\documentclass[11pt]{article}
\usepackage{lmodern} 
\usepackage[colorlinks=true,linkcolor=blue,citecolor=blue]{hyperref}
\usepackage{amsmath, amssymb, amsthm}
\usepackage{mathtools}
\usepackage{array}
\usepackage{float}
\usepackage{fullpage}
\usepackage{graphics}
\usepackage{pifont}
\usepackage{tikz}
\usepackage[T1]{fontenc}

\DeclareMathOperator*{\argmin}{arg\,min}
\usetikzlibrary{arrows.meta}
\usepackage{environ}
\usepackage{framed}
\usepackage{url}
\usepackage[lined,boxed,ruled,norelsize,algo2e,linesnumbered]{algorithm2e}
\usepackage[capitalize]{cleveref}
\crefname{equation}{}{}

\usepackage[noend]{algpseudocode}
\usepackage[labelfont=bf]{caption}
\usepackage{cite}
\usepackage{framed}
\usepackage[framemethod=tikz]{mdframed}
\usepackage{appendix}
\usepackage{graphicx}
\usepackage{tcolorbox}
\usepackage{thm-restate}
\usepackage{pgfplots}
\allowdisplaybreaks[1]

\newcommand\remove[1]{}

\newcommand{\envalias}[2]{\newenvironment{#1}{\begin{#2}}{\end{#2}}}

\newtheorem{theorem}{Theorem}[section]
\envalias{thm}{theorem}
\newtheorem*{thm*}{Theorem}
\newtheorem{lemma}[theorem]{Lemma}
\envalias{lem}{lemma}

\newtheorem*{lemma*}{Lemma}
\newtheorem{corollary}[lemma]{Corollary}
\newtheorem*{corollary*}{Corollary}

\theoremstyle{definition}

\newtheorem*{theorem*}{Theorem}
\newtheorem{definition}[lemma]{Definition}
\envalias{defn}{definition} 

\newtheorem*{rem*}{Remark}

\newcommand\Normal{\mathcal{N}}
\renewcommand\P{\mathbb{P}}

\newcommand{\vu}{\boldsymbol{u}}
\newcommand{\vw}{\boldsymbol{w}}
\newcommand{\valpha}{\boldsymbol{\alpha}}
\newcommand{\vbeta}{\boldsymbol{\beta}}
\newcommand{\vdelta}{\boldsymbol{\delta}}
\newcommand{\va}{\boldsymbol{a}}
\newcommand{\msc}{\mathbf{SC}}
\newcommand{\vx}{\boldsymbol{x}}
\newcommand{\vf}{\boldsymbol{f}}
\newcommand{\vs}{\boldsymbol{s}}
\newcommand{\vd}{\boldsymbol{d}}
\newcommand{\vc}{\boldsymbol{c}}
\newcommand{\veta}{\boldsymbol{\eta}}
\newcommand{\vl}{\boldsymbol{\ell}}
\newcommand{\vy}{\boldsymbol{y}}
\newcommand{\ve}{\boldsymbol{e}}

\newcommand{\vz}{\boldsymbol{z}}

\newcommand{\ovv}{\overline{\boldsymbol{v}}}
\newcommand{\ovw}{\overline{\boldsymbol{w}}}
\newcommand{\hess}{\nabla^{2}}
\newcommand{\vzero}{\boldsymbol{0}}

\newcommand{\nnz}{\mathrm{nnz}}

\newcommand\R{\mathbb{R}}

\newcommand\Z{\mathbb{Z}}

\newcommand{\eps}{\epsilon}
\renewcommand{\O}{\widetilde{O}}
\newcommand{\pe}{\preceq}

\newcommand{\bs}{\backslash}
\newcommand{\assign}{\leftarrow}

\renewcommand{\forall}{\mathrm{\text{ for all }}}

\newcommand{\g}{\nabla}
\newcommand{\diag}{\mathrm{diag}}

\newcommand{\new}{\mathrm{new}}

\newcommand{\mSC}{\mathbf{SC}}

\newcommand\mL{\boldsymbol{L}}
\newcommand\mR{\boldsymbol{R}}

\renewcommand\deg{\mathbf{deg}}

\newcommand{\HH}{\mathcal{H}}
\newcommand{\dspar}{D^{\mathrm{(s)}}}
\newcommand{\IP}{\boldsymbol{T}}

\newcommand\pphi{\boldsymbol{\phi}}
\newcommand\pphitil{\widetilde{\boldsymbol{\phi}}}
\newcommand\ppsi{\boldsymbol{\psi}}
\newcommand\ppsitil{\widetilde{\boldsymbol{\psi}}}

\newcommand\bb{\boldsymbol{b}}
\newcommand\cc{\boldsymbol{c}}
\newcommand\dd{\boldsymbol{d}}

\newcommand\ff{\boldsymbol{f}}

\newcommand\pp{\boldsymbol{p}}
\newcommand\qq{\boldsymbol{q}}
\newcommand\rr{\boldsymbol{r}}

\renewcommand\ss{\boldsymbol{s}}

\newcommand\uu{\boldsymbol{u}}
\newcommand\bell{\boldsymbol{\ell}}

\newcommand\ww{\boldsymbol{w}}
\newcommand\wwbar{\overline{\boldsymbol{w}}}
\newcommand\vv{\boldsymbol{v}}
\newcommand\vvbar{\overline{\boldsymbol{v}}}
\newcommand\vvhat{\widehat{\boldsymbol{v}}}
\newcommand\xx{\boldsymbol{x}}
\newcommand\yy{\boldsymbol{y}}

\newcommand{\mWbar}{\overline{\mW}}
\newcommand{\wwtil}{\widetilde{\ww}}
\newcommand{\mPi}{\mathbf{\Pi}}

\newcommand{\mGamma}{\mathbf{\Gamma}}
\newcommand{\bgamma}{\boldsymbol{\gamma}}
\newcommand{\wt}{\widetilde}

\newcommand{\norm}[1]{\left\| #1 \right\|}

\newcommand{\init}{{\mathrm{(init)}}}

\newcommand{\DynamicSC}{\textsc{DynamicSC}}

\newcommand{\dsc}{D^{\mathrm{(sc)}}}
\newcommand{\dphi}{D^{\mathrm{(\phi)}}}
\newcommand{\dpsi}{D^{\mathrm{(\psi)}}}

\newcommand{\bq}{\boldsymbol{q}}

\foreach \x in {A,...,Z}{
	\expandafter\xdef\csname m\x\endcsname{\noexpand\mathbf{\x}}
}

\newif\ifrandom
\randomtrue

\newcommand{\defeq}{\stackrel{\mathrm{\scriptscriptstyle def}}{=}}

\newcommand{\poly}{{\mathrm{poly}}}

\newcommand{\Locator}{\textsc{Locator}}

\global\long\def\mm{\mathbf{M}}%
\global\long\def\ml{\mathbf{L}}%
\global\long\def\mi{\mathbf{I}}%
\global\long\def\mzero{\mathbf{0}}%
\global\long\def\mproj{\mP}%
\global\long\def\mb{\mathbf{B}}%
\global\long\def\mdiag{\mathbf{diag}}%
\global\long\def\msc{\mathbf{SC}}%
\global\long\def\im{\mathrm{im}}%
\global\long\def\grad{\nabla}%
\global\long\def\ox{\overline{\vx}}%
\global\long\def\os{\overline{\vs}}%
\global\long\def\mw{\mathbf{W}}%
\global\long\def\X{\mathcal{X}}%
\newcommand{\Evaluator}{\textsc{Evaluator}}

\crefname{algocf}{Algorithm}{Algorithms}

\newcommand{\bdelta}{\boldsymbol{\delta}}
\renewcommand{\hat}{\widehat}
\renewcommand{\bar}{\overline}

\begin{document}

\title{
Faster Maxflow via Improved Dynamic Spectral Vertex Sparsifiers
}

\author{
Jan van den Brand\thanks{\texttt{vdbrand@berkeley.edu}. Simons Institute \& UC Berkeley, USA.}
\and
Y{u} Gao\thanks{\texttt{ygao380@gatech.edu}. Georgia Institute of Technology, USA}
\and
Arun Jambulapati\thanks{\texttt{jmblpati@stanford.edu}. Stanford University, USA.}
\and
Yin Tat Lee\thanks{\texttt{yintat@uw.edu}. University of Washington and Microsoft Research Redmond, USA.}
\and 
Yang P.~Liu\thanks{\texttt{yangpliu@stanford.edu}. Stanford University, USA.}
\and
Richard Peng\thanks{\texttt{rpeng@cc.gatech.edu}. Georgia Institute of Technology, USA \&  University of Waterloo, Canada.}
\and
Aaron Sidford\thanks{\texttt{sidford@stanford.edu}. Stanford University, USA.} 
}

\hypersetup{pageanchor=false}

\pagenumbering{gobble}
\begin{titlepage}
\clearpage\maketitle
\thispagestyle{empty}

\begin{abstract}
We make several advances broadly related to the maintenance of electrical flows in weighted graphs undergoing dynamic resistance updates, including:
\begin{enumerate}
\item More efficient dynamic spectral vertex sparsification, achieved by faster length estimation of random walks in weighted graphs using Morris counters [Morris 1978, Nelson-Yu 2020].
\item A direct reduction from detecting edges with large energy in dynamic electric flows to dynamic spectral vertex sparsifiers.
\item A procedure for turning algorithms for estimating a sequence of vectors under updates from an oblivious adversary to one that tolerates adaptive adversaries via the Gaussian-mechanism from differential privacy.
\end{enumerate}
Combining these pieces with modifications to prior robust interior point frameworks gives an algorithm that on graphs with $m$ edges computes a mincost flow with edge costs and capacities in $[1, U]$ in time $\widetilde{O}(m^{3/2-1/58} \log^2 U)$. In prior and independent work, [Axiotis-M\k{a}dry-Vladu FOCS 2021] also obtained an improved algorithm for sparse mincost flows on capacitated graphs. Our algorithm implies a $\widetilde{O}(m^{3/2-1/58} \log U)$ time maxflow algorithm, improving over the $\widetilde{O}(m^{3/2-1/328}\log U)$ time maxflow algorithm of [Gao-Liu-Peng FOCS 2021].
\end{abstract}
\end{titlepage}

\hypersetup{pageanchor=true}

\newpage
\addtocontents{toc}{\protect\setcounter{tocdepth}{2}}
\tableofcontents

\newpage
\pagenumbering{arabic}

\section{Introduction}
\label{sec:intro}

The maximum flow (maxflow) problem asks to route the maximum amount of flow between two vertices $s$ and $t$ in a directed graph $G$ such that the amount of flow on every edge is at most its capacity. The more general minimum cost flow (mincost flow) problem asks to route a fixed demands in a directed graph $G$ without sending more flow on any edge than its capacity, while minimizing a linear cost. Together these well-studied problems cover a wide range of combinatorial and numerical problems, including maximum cardinality bipartite matching, minimum $s$-$t$ cut, shortest paths in graphs with negative edge length, and optimal transport (see e.g. \cite{BLNPSSSW20,BLLSSSW21}).

While classical algorithms for these problems revolved around using augmenting paths or cycle primitives (such as blocking flows) \cite{K73,ET75,GT90,GR98}, the last decade has seen significant runtime improvements for maxflow and mincost flow in various settings based on \emph{electrical flows}.
For a graph $G$ with vertex set $V$ and edge set $E$, with edge resistances $\rr \in \R^E_{\ge0}$, and a demand vector $\dd \in \R^E$, the electric flow on $G$ is the flow that routes a fixed demand while minimizing the energy $\sum_{e \in E} \rr_e\ff_e^2$. Better maxflow algorithms have been given in several regimes using eletrical flows and stronger primitives \cite{AKPS19,KPSW19}, including unit capacity graphs \cite{M13,M16,CMSV17,LS20_STOC,KLS20,AMV20}, approximate maxflow on undirected graphs \cite{CKMST11,LRS13,S13,KLOS14,P16,S17area,ST18}, and dense graphs \cite{LS19,BLNPSSSW20,BLLSSSW21}.

However, it has been particularly challenging to obtain running time improvements for solving mincost flow and maxflow to high precision in sparse capacitated graphs. Recently, \cite{GLP21} gave an $\O(m^{1.5-1/328} \log U)$ time algorithm for sparse graphs with capacitaties bounded by $U$, the first improvement over $\O(m^{1.5} \log U)$ for maxflow on sparse graphs with arbitrary, polynomially bounded capacities. Their improvement involved an intricate interaction of dynamic data structures for electrical flows, a modification of the standard interior point method (IPM) outer loop which builds an maxflow via $\sqrt{m}$ approximate electric flows \cite{K84,Vaidya89}, and sketching techniques. Additionally, issues relating to randomness in data structures and combinatorial reasoning about errors resulting from random walks required careful analysis that signficantly increase the runtime and resulted in the small improvement over $m^{1.5}$.

Our main result is an algorithm that while has a similar high-level picture as \cite{GLP21}, significantly simplifies the major pieces described previously and their interactions. Specifically, we give a general purpose sketching tool for electrical flows, a graph theoretic (instead of algebraic) way of constructing the random walks at the core of the data structure, and handle randomness dependencies using ideas from differnetial privacy. Further, we provide modifications to prior robust interior point frameworks to do $\ell_2$-based recentering (\cref{def:solution_approximation}) within a robust IPM via additional spectral vertex sparsification techniques.
As a result, we achieve a faster runtime than in \cite{GLP21}. Also, as a result of our simplified electric flow data structure, our algorithm and IPM seamlessly extend to mincost flow. In constrast, the data structure complications in \cite{GLP21} restricted their algorithm to be applied to maxflow.
\footnote{\cite{AMV21}, in FOCS 2021,
claims an improvement to sparse capacitated mincost flow in the title.
A preprint was recently made available at~\url{https://arxiv.org/abs/2111.10368v1}.
The results in this paper were derived independently:
we defer detailed comparisons to a future version.}
\begin{theorem}
\label{thm:main}
There is an algorithm which given any $m$-edge directed graph $G$ with integral capacities in $[1, U]$, feasible demand vector $\dd \in \Z^E$, and an integral cost vector $\cc \in [-U, U]^E$, computes a flow $\ff$ that routes demand $\dd$, satisfies the capacity constraints, and minimizes $\cc^\top\ff$. The algorithm succeeds whp. and runs in time $\O(m^{3/2-1/58}\log^2 U)$.
\end{theorem}
Overall, this paper simplifies the key pieces of \cite{GLP21} and, as a result, clarifies the important components used to achieve faster maxflow algorithms via dynamic electric flows. Further, we consider each of these pieces to be interesting in their own right: dynamic maintenance of Schur complements, sketching and maintenance of high energy edges in dynamic electric flows, and understanding reductions between adaptive and oblivious adversaries. Ultimately, this paper be read independently of \cite{GLP21} and the proofs are simpler and more natural in many cases.

\subsection{Key Algorithmic Pieces}
\label{sec:keypieces}

Here we cover the key algorithmic pieces underlying our algorithm. The first major piece is a faster algorithm for generating random walks of a fixed length from a vertex, a core primitive in all random-walk based approaches to dynamic electric flows \cite{DGGP19,GLP21}. Our second key contribution is an algorithm for detecting large energy edges in electric flows for graphs with dynamically changing resistances and demands based on a direct reduction to dynamic spectral vertex sparsifiers (Schur complements). As our data structures for maintaining dynamic electric flows naturally work only against oblivious adversaries, we develop an approach that black-box reduces dynamic electric flows against adaptive adversaries to the same problem against oblivious adversaries, at the cost of a small runtime increase.

\paragraph{Faster generation of random walks and Schur complements.} All previous algorithms for dynamic electric flows \cite{DGGP19,GLP21} require dynamic maintenance of spectral vertex sparsifiers or Schur complements, which approximate the electric flow and potentials onto a smaller set of terminal vertices.
The Schur complement is generated by sampling several random walks from vertices $v$ with exit probabilities proportional to inverse resistances until the walk visited a fixed number $L$ of distinct vertices, and by estimating the sum of resistances of edges along the walk. 

Because there may be edges with very large or small resistances, a na\"{i}ve simulation may get stuck for polynomially many steps. Consequently, \cite{DGGP19} gave an algorithm for this based on taking high powers of the random walk matrix (which \cite{GLP21} applied in a black-box fashion). This generated large factors in the runtime of the data structures. We give an approach to signficantly speed up the sampling of vertices and length estimation by applying a \emph{Morris counter} from the streaming/sketching literature \cite{NY20}, and reducing the problem to solving a sequence of electric flow computations (Laplacian systems) as opposed to the more expensive matrix multiplication operations of \cite{DGGP19}. This signficant runtime improvement immediately translates to our dynamic electric flow data structure described above, which directly uses dynamic Schur complements.

\paragraph{Simplified electric flow heavy hitter.} To design our dynamic algorithm for detecting edges with large energy in dynamic electric flows, we maintain an $\ell_2$ heavy hitter sketch of the electric flow vector. The algorithm of \cite{GLP21} maintained this sketch by using a dynamic spectral vertex sparsifier or Schur complement, which approximates the electric flow and potentials on a smaller set of terminal vertices, and several random walks for ``moving'' the heavy-hitter sketch vector to the terminal set. This latter piece (maintaining random walks for moving the heavy hitter vector) introduced several complications into the analysis and generated a large overall running time for the data structure. On the other hand, our algorithm is more directly based on spectral approximations. In particular, we show how to dynamically maintain the result of moving the heavy hitter vector onto the terminal set by simply calling another dynamic Schur complement data structure,
and carefully reasoning about spectral approximations to bound how that affects the resulting error.

\paragraph{Simplified IPM outer loop.} As a result of our more linear algebraic approach to maintaining electric flows,
our data structure for detecting large energy edges in electric flows works for both dynamic resistance changes and dynamic demands, while the algorithm of \cite{GLP21} required restricting to only $s$-$t$ flows. Our generalization also allows us to use a more standard and efficient robust IPM (from \cite{DLY21}) to implement the outer loop utilizing the data structure, while \cite{GLP21} had to redesign the IPM to carefully only use $s$-$t$ electric flows to interact with their data structure. Our robust IPM implements an additional batching, or $\ell_2$-based recentering step, by computing the changes on a small subset of edges to higher accuracy by using spectral vertex sparsifiers.

\paragraph{Black-box reduction from adaptive to oblivious adversaries.} As we are applying randomized data structure inside an algorithmic outer loop, their previous responses may affect future updates. This is referred to an \emph{adaptive adversary} in the literature.
On the other hand, our data structures which are based on random walks naturally only work against \emph{oblivious adversaries}, where the input sequence does not depend on the outputs and randomness of the data structure. The algorithm of \cite{GLP21} handled this issue in their data structures by carefully controlling the total number of adaptive phases of their algorithm before snapping back to a deterministic state.

Our approach on the other hand is more black-box, and gives a more general approach for converting data structures against oblivious adversaries to handle adaptive queries.
We build a \textsc{Locator} which returns a superset of edges with large energies, and several \textsc{Evaluator}s with differing accuracy parameters which separately estimate the energies of the edges. By leveraging ideas from the Gaussian-mechanism from differential privacy \cite{DR14} we show how to apply the \textsc{Evaluator} data structures to simulate estimating adding Gaussian noise to the true energy vector that we wish to output. We simulate this by making several queries to the \textsc{Evaluator}s, where we query the least accurate \textsc{Evaluator}s most often, and only query more accurate \textsc{Evaluator}s when the estimate of the energy vector is close to certain thresholds and we require finer estimates to decide how to round. Because we are simulating adding noise to the true output, the algorithm succeeds against an adaptive adversary.

\subsection{Related Work}
\label{sec:related}

We briefly survey the lines of work most relevant to our results, and refer the reader to \cite{GLP21} for more comprehensive discussion. Recently, \cite{planar22} gave a mincost flow algorithm on planar graphs running in nearly linear time. Similar to our paper, it is based on the robust IPM framework of \cite{DLY21} and dynamic Schur complements. However, \cite{planar22} relies on the fact that the terminal set $C$ is small due to the existence of planar separators, while our paper relies on the fact that $C$ is slowly changing.

\paragraph{Data structures for IPMs.} IPMs are a powerful framework which reduces linear programming with $m$ variables to a sequence of $\O(\sqrt{m})$ linear system solutions \cite{Ren88}. For maxflow and mincost flow, these linear systems correspond to computing electrical flows, and Daitch-Spielman \cite{DS08} leveraged this observation to give a $\O(m^{1.5}\log U)$ mincost flow algorithm. Recently, several works have leveraged the key fact dating back to early works of Karmarkar \cite{K84} and Vaidya \cite{Vaidya89} that the linear systems change slowly and only need to be solved approximately, both in the context of linear programs \cite{LS15,CLS19,LSZ19,Brand20,BLSS20,JSWZ21,BLLSSSW21,Brand21} 
and mincost flows \cite{BLNPSSSW20,BLLSSSW21,GLP21}.

\paragraph{Dynamic electrical flows.} Recent works applying dynamic data structures to IPMs for maxflow require maintaining various properties of electrical flows on dynamically changing graphs. The improvements on dense graphs \cite{BLNPSSSW20,BLLSSSW21} required dynamically maintaining spectral sparsifiers of the Laplacian in $\O(1)$ time per edge update and $\O(n)$ per query, as well as detecting edges with large electrical energies in $\O(n)$ time per query. Both of these pieces were done using dynamic expander decompositions \cite{NS17,W17,NSW17,SW19,CGLNPS19,BBG20}. The work of \cite{GLP21} desired \emph{sublinear} time per query and hence required dynamically maintaining Schur complements, whose study was initiated in \cite{DGGP19} to dynamically maintain approximate effective resistances.

\paragraph{Adaptivity and differential privacy.} There has been significant work towards building techniques to apply oblivious data structures in the context of an algorithmic outer loop, which requires adaptivity. To date, most approaches to this problem involve either making the algorithm deterministic \cite{BC16,BC17,CK19,GWN20,C21,BGS21}, or resparsifying \cite{BBG20}, both of which heavily leverage properties provided by dynamic expander decompositions \cite{NS17,W17,NSW17,SW19,CGLNPS19,BBG20}. Our work takes a different perspective and instead more carefully analyzes whether the adversary can learn any randomness leaked from the distribution of our output vector. This perspective is motivated by ideas from differential privacy, and in fact our key result is an adaptation of the Gaussian mechanism \cite{DR14} which simulates adding unbiased Gaussian noise to the true output vector by using a sequence of oblivious estimates. Our recursive scheme is also broadly related to the idea of multilevel Monte Carlo \cite{Giles15,BG15} and its recent applications in leveraging approximate optimization procedures to obtain nearly unbiased estimates of minimizers \cite{ACJJS21}.

\subsection{General Notation}
\label{sec:notation}

We use plaintext to denote scalars, bold lowercase for vectors, and bold uppercase for matrices. For resistances $\rr$ and conductances $\ww \defeq \rr^{-1}$, the corresponding capital matrices are diagonal matrices with the vector entries on the diagonal, i.e. $\mR \defeq \mathrm{diag}(\rr)$ and $\mW \defeq \mathrm{diag}(\ww)$. As our algorithm heavily use approximations, we will use $\widetilde{\cdot}$ to denote the approximate versions of true variables.

We use $\O(\cdot)$ to suppress logarithmic factors in $m$ and $\widetilde{\Omega}(\cdot)$ to suppress inverse logarithmic factors in $m$. For vectors $\xx,\yy$ we sometimes let $\xx\yy$ denote the entry-wise product of $\xx,\yy$, so $(\xx\yy)_i \defeq \xx_i\yy_i$. Similarly, we let $(\xx/\yy)_i \defeq \xx_i/\yy_i$. We say that an event holds with high probability (whp.) if for any constant $C > 0$, the event succeeds with probability at least $1-n^{-C}$ by adjusting parameters. We let $[n] = \{1, 2, \dots, n\}.$ We denote the (unweighted) degree of a vertex $v$ as $\deg(v)$.

We say that a symmetric matrix $\mM \in \R^{n \times n}$ is positive semidefinite (PSD) if $\xx^\top\mA\xx \ge 0$ for all $\xx \in \R^n.$ For PSD matrices $\mA, \mB$ we write $\mA \pe \mB$ if $\mB - \mA$ is PSD. For positive real numbers $a, b$ we write $a \approx_{\gamma} b$ to denote $\exp(-\gamma)b \le a \le \exp(\gamma)b.$
For PSD matrices $\mA, \mB$ we write $\mA \approx_{\gamma} \mB$ if $\exp(-\gamma)\mB \pe \mA \pe \exp(\gamma)\mB$.

\subsection{Organization}
\label{sec:organization}

In \cref{sec:overview}, we give a technical overview of each of our improvements to each of the key components of \cite{GLP21}: faster sampling of Schur complements, operator-based electric flow heavy hitters, and black-box reduction of adaptive to oblivious adversaries. We also overview the robust IPM we use. In \cref{sec:prelim} we give preliminaries for maxflow, mincost flow, and electric flows that we require for the remainder of our paper. In \cref{sec:dynamicsc} we give our algorithm for faster sampling of random walks and Schur complements, and we combine this with an operator-based heavy hitter in \cref{sec:er} to give a faster algorithm for detecting edges with large energy. In \cref{sec:adaptive} we show how to black-box reduce adaptive to oblivious adversaries for the problem of estimating dynamic vectors. We give our robust IPM in \cref{sec:ipm}, which is an adaptation of that in \cite{DLY21}, and additional tools to apply it. Finally, we combine all pieces and compute the final runtime in \cref{sec:combine}.

\section{Overview}
\label{sec:overview}

Here w,e provide a technical overview of our contributions.

\subsection{Overview of Faster Schur Complements via the Morris walk}
\label{subsec:morrisoverview}
Our data structures, as in \cite{GLP21}, heavily rely on dynamically maintaining spectral vertex sparsifiers (Schur complements) of $G$, which approximate the inverse spectral form of $G$ onto a subset of the vertices. This was achieved using the algorithm of \cite{DGGP19}, which showed how to dynamically maintain an approximate Schur complement under edge resistance updates.
The main primitive behind the dynamic Schur complement data structure was a procedure to sample random walks from a vertex with exit probability proportional to inverse resistances, i.e. the probability of going from a vertex $v$ to a neighbor $u$ is given by \[ \frac{\rr_{vu}^{-1}}{\sum_{w \text{ neighbor of } v} \rr_{vw}^{-1}}. \] For this walk and a parameter $L$, we must run the walk until the total degree of visited vertices is $L$, and to estimate the total \emph{resistive length} of the walk up to a $(1+\eps)$-factor, where resistive length refers to the sum of resistances of edges on the walk.
Directly simulating the random walk is not efficient enough, because there may be polynomially large and small resistances, which cause the walks to ``get stuck'' on a small set of edges, without visiting new vertices. Thus it can take a long time to visit $L$ distinct vertices.
Despite this, \cite{DGGP19} showed how to sample the walk and resistive length in $\O(L^4\eps^{-2})$ time per vertex, and this large runtime directly led to the fact that \cite{GLP21} only achieved a small $1/328$ improvement in the exponent.

Interestingly, if one is only interested in obtaining distinct vertices on the walk (and not the resistive length) until the total degree is $L$, this can be done in $\O(L^2)$ time by solving a Laplacian system (corresponding to computing an electric flow) for each of at most $L$ steps to compute the next exit vertex. 
However, the approach of \cite{DGGP19} which also computed the resistive length, i.e. the sum of resistances of edges on the walk, was based on matrix-powering/matrix multiplication, and instead had a larger $\O(L^4\eps^{-2})$ runtime.
Our main idea is to resolve this runtime discrepancy between sampling the distinct vertices and computing a $(1+\eps)$-resistive length estimate by giving an algorithm that computes both quantities by solving a sequence of Laplacian systems.
In total, we solve $\O(L + \eps^{-2})$ systems, for a runtime of $\O(L^2 + L\eps^{-2})$. 
In our settings, $L$ will generally be $\Omega(\eps^{-2})$, so our runtime is $\O(L^2)$, matching the time to generate the first $L$ vertices using a sequence of Laplacian systems, and significantly improving over the $\O(L^4\eps^{-2})$ runtime of \cite{DGGP19,GLP21}.

Our algorithm for this task is derived from the Morris counter \cite{Morris78a,NY20}, a probabilistic algorithm for maintaining low-space approximations to a counter $N$ undergoing increments.
For simplicity of exposition, we assume our input graph has integer, polynomially-bounded edge weights. 
Our algorithm intuitively begins by running a random walk in $G$. 
However, we replace the naive procedure for computing the resistive length with a Morris counter. 
More precisely, assume we have run a random walk starting from a vertex $u$ for $k$ steps and have estimated the resistive length of the walk via the Morris counter. 
To estimate the resistive length of this walk after a further step, we simply sample one new step of the walk: if we sampled an edge of resistance length $w$, we increment the Morris counter $w$ times. 
In this way, the Morris counter enables us to maintain estimates of the resistive length of a random walk. 

We use the following properties of Morris counters as shown in \cite{NY20}. First, they take discrete values of $\frac{1}{a} \left( (1+a)^i - 1 \right)$ for some real number $a > 0$ and integers $i \ge 0$, and if $a = \eps^2/\poly\log n$, the value of the Morris counter is always a $(1+\eps)$-approximation of the true value with high probability. In particular, for graphs with polynomially bounded weights, the Morris counters takes at most $\O(a^{-1}) = \O(\eps^{-2})$ distinct values.

Our key insight we can simulate incrementing of the Morris counter and the sequence of vertices visited as a random walk on a ``lifted'' graph with $\O(\eps^{-1})$ layers. Each time we explore a ``new'' neighbor in this lifted space, we either find a new unexplored vertex along the random walk or increment the Morris counter. 
However, there are only $\O(\eps^{-2})$ distinct values of the counter: thus we must explore only an additional $\O(\eps^{-2})$ distinct vertices in this lifted space to obtain the desired guarantee on the number of new vertices seen. 
We obtain our final algorithm by replacing the explicit random walks with a subroutine based on Laplacian linear system solvers: in this way our final complexity of $\O(L^2 + L \eps^{-2} )$ follows. Overall, this provides a graph-theoretic approach for estimating lengths of random walks on graphs, as opposed to the previous algorithm in \cite{DGGP19} which was based on matrix mutiplication.

\subsection{Overview of Operator-based Electric Flow Heavy Hitters} 
\label{subsec:locatoroverview}

Our next major improvement over \cite{GLP21} is a data structure that detects large energy edges in electric flows on graphs with dynamic resistances and demands by direct reduction to maintaining dynamic Schur complements. To be precise, we give a data structure that on a graph $G$ with dynamically changing resistances and demands, solves a \emph{electric flow heavy hitter} problem, by returning a set $S$ of $O(\eps^{-2})$ edges containing all edges $e$ with at least $\eps^2$ fraction of the electric energy, i.e. $\rr_e\ff_e^2 \ge \eps^2 \sum_{e \in E} \rr_e\ff_e^2$ where $\ff$ is the electric flow vector. \cite{GLP21} gave a data structure that solved this problem in sublinear time per resistance update and query as a core piece of their algorithm. We give improved runtimes for solving this problem and shed light on its complexity by directly reducing to dynamically maintaining Schur complements.

Note that the dynamic electric flow heavy hitter problem is equivalent to detecting large coordinates of the vector $\mR^{1/2}\ff$ compared to its $\ell_2$ norm. Hence, it is natural to apply an $\ell_2$ heavy-hitter sketch \cite{KNPW11}, which at a high-level consists of $\O(\eps^{-2})$ Johnson-Lindenstrauss $\ell_2$ sketches. In total, this consists of maintaining the value of $\qq^\top \mR^{1/2}\ff$ for $\O(\eps^{-2})$ random sketch vectors $\qq \in \{-1,0,1\}^E.$ The flow $\ff$ can be represented as $\ff = \mR^{-1}\mB\pphi$ for electric potentials $\pphi$ and edge-vertex incidence matrix $\mB$, so \[ \qq^\top \mR^{1/2}\ff = \langle \mB^\top\mR^{-1/2}\qq, \pphi \rangle. \] Let $\yy = \mB^\top\mR^{-1/2}\qq$, so that we focus on maintaining $\yy^\top\pphi$. However, $\pphi$ is still a $|V|$-dimensional vector, so in order to achieve sublinear time \cite{GLP21} used a smaller terminal set $C$ to estimate $\yy^\top\pphi$. In particular, they write $\pphi = \HH_C\pphi_C$ where $\pphi_C$ is the restriction of $\pphi$ to $C$, and $\HH_C \in \R^{V(G) \times C}$ is the \emph{harmonic extension} (\cref{def:h}) operator which extends $\pphi_C$ to $\pphi$ using that for any vertex $v$, $\pphi_v$ is the average of its neighbors, weighted proportional to inverse resistances. This way, $\yy^\top\pphi = \langle \HH_C^\top\yy, \pphi_C \rangle$. Assuming that we can approximate maintain $\pphi_C$ (which we discuss towards of the end of this section's overview), it suffices to maintain $\HH_C^\top\yy$.

Our major difference from \cite{GLP21} is in how we maintain $\HH_C^\top\yy$. While \cite{GLP21} used the combinatorial interpretation of the operator $\HH_C^\top$ as using random walks to ``move'' the mass from vector $\yy$ onto $C$, we use the spectral fact (\cref{lemma:htosc}) that
\begin{align*} \HH_C = \mL^\dagger\begin{bmatrix} 0 & \mSC(\mL, C) \end{bmatrix}, \end{align*} where $\mL$ is the graph Laplacian and $\mSC(\mL, C)$ is the Schur complement of $\mL$ onto $C$.
Thus, we get
\[ \HH_C^\top\yy = \begin{bmatrix} 0 \\ \mSC(\mL, C) \end{bmatrix}\mL^\dagger\yy. \]
Thus, we could optimistically precompute $\mL^\dagger\yy$ and then compute $\HH_C^\top\yy$ as long as we can dynamically maintain the Schur complement $\mSC(\mL, C)$, which is a size $C$ object. The remaining issue is that the Laplacian $\mL$ may change because the resistances change. However, the operator $\HH_C^\top$ does not depend on the resistances of edges completely inside $C$ (by definition), so we may actually let $\mL$ be the Laplacian of the \emph{original graph} as long as all endpoints of edges with resistance changes are added to $C$. Finally we are able to show that using an approximate Schur complement in place of $\mSC(\mL, C)$ still suffices for our data structures (\cref{lemma:approxpi}).

Finally we discuss the (approximate) maintenance of the potential $\pphi_C$. For an electric flow $\ff$ routing demand $\dd$, i.e. $\mB^\top\ff = \dd$, the potentials $\pphi_C$ are given by
\[ \pphi_C = \mSC(\mL, C)^\dagger\HH_C^\top\dd. \] In other words, we first ``move'' the demands to the terminal set using $\HH_C^\top$ just as above, and then invert the Schur complement on it. Therefore we can maintain $\pphi_C$ as follows: maintain $\HH_C^\top\dd$ approximately as above, and then also approximate maintain $\mSC(\mL, C)$ using an approximate Schur complement data structure. In all, this reduces the maintenance of the heavy hitter vector to three calls to an approximate Schur complement oracle.

\subsection{Overview of Reduction from Adaptive to Oblivious Adversaries}
\label{subsec:adaptiveoverview}
\label{sec:overview:dp}

The dynamic electric flow data structures built in \cref{subsec:locatoroverview} na\"{i}vely only work against oblivious adversaries, i.e. the inputs must be independent of the outputs and randomness of the data structure. In \cite{GLP21} this was handled by carefully designing the data structures to utilize the fact that the IPM central path is a deterministic object. However, we take a more general approach, by applying ideas from the Gaussian-mechanism from differential privacy \cite{DR14} to build versions of these data structures that work directly against adaptive adversaries, allowing them to be applied within the interior point outer loop. In fact, we give a generic reduction for estimating vectors against adaptive adversaries to oblivious adversaries.

Consider an oblivious data structure that outputs vectors $\ovv \in \R^m$ that are supposed to approximate a true underlying vector $\vv \in \R^m$. In our dynamic electric flow setting, this corresponds to a data structure which detects edges with large electric energy, and approximates their flow values. Consider what would happen if instead of $\ovv$, our algorithm uses $\vz \sim \Normal(\vv, \sigma^2)$ for small enough $\sigma$ (i.e. the vector $\vv$ with some Gaussian noise added to it).
If $\sigma$ is small enough, then $\vz$ would be an accurate approximation for our algorithm to work. Additionally, the vector $\vz$ obviously does not depend on the internal randomness of the data structure, since it is defined with respect to $\vv$, not the approximation $\ovv$.
Unfortunately, computing $\vz$ by computing $\vv$ and adding noise is rather inefficient since $\vv$ is the exact solution, not an approximation. We now explain how to obtain vector $\vz$ more efficiently from oblivious estimate vectors $\ovv$ by using the Gaussian-mechanism from differential privacy \cite{DR14}.

Specifically, it is known that for any $\sigma > 0$ there is small enough $\alpha > 0$ such that if $d$ is the density function of $\Normal(\vv, \sigma^2)$ and $\bar{d}$ is the density function of $\Normal(\ovv, \sigma^2)$, then $\bar{d}(\vx) \le \exp(\alpha) d(\vx)$ for all $\vx$.\footnote{
This is actually only true for $\vx \in D$ for some event $D$ that holds whp. We ignore this here for simplicity. 
}
For example, \Cref{fig:overview:dp} shows density function $d(\vx)$ and the scaled density function $\exp(-\alpha)\cdot \bar{d}(\vx)$ for the $1$-dimensional case.
Note that we can pick a random $\vz \sim \Normal(\vv, \sigma^2)$ by picking uniformly at random a point below the curve of $d(\ox)$ and returning the $x$-coordinate.
We can also split this sampling scheme into two phases:
(i) With probability $1-\exp(-\alpha)$, sample from the area between the two curves.
(ii) Alternatively, with probability $\exp(-\alpha)$ sample from the area below the bottom curve $\exp(-\alpha)\bar{d}(\vx)$ in \Cref{fig:overview:dp}.

\pgfmathdeclarefunction{gausso}{2}{\pgfmathparse{1/(#2*sqrt(2*pi))*exp(-((x-#1)^2)/(2*#2^2))}}
\pgfmathdeclarefunction{smallgausso}{2}{\pgfmathparse{0.5/(#2*sqrt(2*pi))*exp(-((x-#1)^2)/(2*#2^2))}}

\begin{figure}
\center
\begin{tikzpicture}
\begin{axis}[every axis plot post/.append style={
  mark=none,domain=-1.5:3.5,samples=50,smooth},
  axis x line*=bottom,
  axis y line*=left,
  enlargelimits=upper,
  yticklabels={,,},
  xticklabels={,,},
  ticks=none]
  \addplot[black] {smallgausso(1.1,0.75)};
  \addplot[black] {gausso(1,0.75)};
\end{axis}
\draw[dotted] (3.1,0) -- (3.1,5.5);
\draw[dotted] (3.25,0) -- (3.25,3);
\node[text width=1cm] at (3.3, -0.25) {$\vv$ $\ovv$};
\end{tikzpicture}
\caption{\label{fig:overview:dp}
Density function $d$ of $\Normal(\vv,\sigma^2)$, and density function $\bar{d}$ of $\Normal(\ovv,\sigma^2)$ scaled by some $\exp(-\alpha)$, $\alpha > 0$ so that $\bar{d}(\vx) \exp(-\alpha) \le d(\vx)$.}
\end{figure}

When case (i) happens, we handle it directly by computing $\vv$ exactly (which is expensive), which gives us the distributions $d$ and $\bar{d}$ explicitly. However, note that if $\alpha$ is close to $0$, then this case only occurs infrequently: with probability $1 - \exp(-\alpha) = O(\alpha)$, which balances out the expensive cost of computing $\vv$. On the other hand, case (ii), which occurs with probability $\exp(-\alpha)$, corresponds to flipping an unbalanced coin and with probability $\exp(-\alpha)$ we sample a $\vz' \sim \Normal(\ovv, \sigma^2)$.
So with probability $\exp(-\alpha)$ we do not need to know/compute the exact vector $\vv$ in order to obtain a sample with distribution $\Normal(\vv, \sigma^2)$ and just knowing the approximate result $\ovv$ already suffices.

Now, this scheme can be extended recursively to handle case (ii), i.e. sampling from $\vz \sim \Normal(\ovv, \sigma^2)$. We can use the same scheme again via some $\ovv'$, i.e.~sampling from $\Normal(\ovv',\sigma^2)$ with probability $\exp(-\alpha)$ instead of $\Normal(\ovv, \sigma^2)$.
This leads to another speed-up because of the following reason: the probability $\exp(-\alpha)$ depends on the approximation quality of $\ovv'$ compared to $\vv$. We want to use a large $\alpha$ in order to reduce the probability of computing $\vv$, but this requires $\ovv'$ to be a better approximation. Thus, we are able to compute higher accuracy approximations (which take more runtime) less frequently, and this leads to a speedup. Overall, by using this scheme, our data structures will work against an adaptive adversary because the output has distribution $\Normal(\vv, \sigma^2)$, i.e.~a distribution that is independent of the internal randomness of the data structures.

\subsection{Overview of IPM Outer Loop}
\label{subsec:ripmoverview}
\label{subsec:overviewipm}

Here we overview how we apply the above primitives in a robust IPM to give an algorithm for algorithm, which reduces solving maxflow to computing a sequence of $\O(\sqrt{m})$ approximate electric flows. The IPM of \cite{GLP21} required several nonstandard modifications, including restricting to using $s$-$t$ flows, which resulted in using more than $\O(\sqrt{m})$ steps, and overall higher runtime. On the other hand, our algorithm is based on the more standard robust IPM of \cite{DLY21}, with an additional procedure that allows for recentering in the context of a robust IPM that allows us to control errors that accumulate over longer periods of time.

We start by briefly introducing a standard robust IPM setup for the mincost flow problem based on \cite{DLY21} (in \cref{sec:ipm} we change notation slightly to work with general linear programs)
\begin{equation}
\min_{\vf\in\R^{m} : \mb^{\top}\vf=\vd\text{ and }\vl\leq\vf\leq\vu}\vc^{\top}\vf, \label{eq:mincostflow}
\end{equation}
where $\vc \in \R^E$ is the cost vector, and $\vl, \vu \in \R^E$ are lower/upper capacities on edges. For $e \in E$ and real number $f \in \R$, define the logarithmic barrier function $\phi_e(f) \defeq -\log(f-\vl_e)-\log(\vu_e-f)$, and for flow $\ff \in \R^E$ define $\phi(\ff) \defeq \sum_{e \in E} \phi_e(\ff_e).$ For a path parameter $\mu$ that decreases towards $0$ over the course of $\O(\sqrt{m})$ steps, the robust IPM maintains an approximate minimizer to the expression
\begin{equation}
\vf_{\mu} \defeq \min_{\vf\in\R^{m} : \mb^{\top}\vf=\vd\text{ and }\vl\leq\vf\leq\vu} \vc^{\top}\vf + \mu\phi(\vf)\,. \label{eq:reg}
\end{equation}
Since $\phi$ is convex, the KKT conditions for \eqref{eq:reg} give that there is a vector $\yy$ such that $\vc + \mu \g\phi(\vf_{\mu}) = \mu\mB\yy$. Thus there is a vector $\vs_{\mu} \in \R^E$ such that $\mB\yy + \vs_{\mu} = \vc/\mu$ and $\vs_{\mu} + \g\phi(\ff_{\mu}) = 0$. In this way, we define a $\mu$-\emph{centered point} as a pair $(\vf,\vs)$ such that $\mB\yy + \vs = \vc/\mu$ for some $\yy \in \R^V$ and $\|\g^2\phi(\ff)^{-1/2}(\vs + \g\phi(\ff)) \|_{\infty} \le 1/64$. The robust IPM maintains $\mu$-centered points throughout by tracking the potential function \[ \sum_{e \in E} \cosh\left(\lambda\phi_e''(\ff_e)^{-1/2}(\vs_e + \phi_e'(\ff_e)) \right) \] for $\lambda = 128\log(16m)$. Now IPM steps are taken to simulate gradient descent steps on the potential to keep it small, and hence maintain $\mu$-centered points at all times.

Because our data structures work in time sublinear in the number of vertices, the flows we maintain during the robust IPM are stored implicitly, even without the ability to query in $\O(1)$ time the ``true flow'' on an edge $e$. Further, the error of our flow estimate from the true value accumulates over the steps of our method. Hence we require a procedure to recompute a feasible $\mu$-centered flow in a robust IPM every $k$ steps in $\O(m)$ time for some $k = m^{\Omega(1)}$ (\cref{thm:graph_solution_approximation}). To see why this could be possible, note that the true step per iteration is an electric flow with some resistances and demands. Additionally, over the course of $k$ steps, these resistances and demands will only change at most $\poly(k)$ total times. Thus, we can put all edges whose resistance or demand changed into a terminal set $C$ and compute an $\eps$-approximate Schur complement onto $C$. \cite{LS18} shows that such a Schur complement (onto a slightly larger set) can be computed in time $\O(m + |C|/\eps^2) = \O(m + \poly(k)/\eps^2) = \O(m)$ for some $k = m^{\Omega(1)}$. Leveraging this, we show that we can recover a centered point in the context of a robust IPM in $\O(m)$ time every $k$ steps.

Overall, our algorithm splits the $\O(\sqrt{m})$ robust IPM steps in $\O(\sqrt{m}/k)$ batches of $k$ steps. Within each batch, we ensure that at most $\poly(k)$ edges have their resistances change in the graph $G$ (but there may be more resistance updates in between batches). Each step in the batch is maintained using the dynamic electric flow heavy hitter data structure we built, as described in \cref{subsec:morrisoverview,subsec:locatoroverview,subsec:adaptiveoverview}. At the end of each batch, we use the approximate recentering procedure described in the previous paragraph. Combining these pieces along with the standard bound that over $T$ IPM steps, at most $\O(T^2)$ resistances change by a constant factor, gives our final runtimes.
\section{Preliminaries}
\label{sec:prelim}

We give preliminaries on maxflow, mincost flow, electric flows, and Schur complements.

\paragraph{Maxflow and mincost flow.} Throughout, we let $G = (V, E)$ be our graph with $n = |V|$ vertices and $m = |E|$ edges. We let $\mB \in \R^{E \times V}$ denote the edge-vertex incidence matrix of $G$. Additionally, we let $\bell, \uu \in \Z^E$ denote the lower/upper capacities on edges in $G$. We assume that $\|\bell\|_\infty, \|\uu\|_\infty \le U$. A \emph{flow} $\ff \in \R^E$ is any assignment of real numbers of edges of $G$. We say that a flow $\ff$ is \emph{feasible} if $\bell_e \le \ff_e \le \uu_e$ for all $e \in E$. We say that $\ff$ routes the demand $\dd \in \R^V$ if $\mB^\top\ff = \dd$.

The maximum flow problem asks to find a feasible flow routing the maximum multiple of a demand $\dd$ (generally assumed to be $s$-$t$). Written linear algebraically, this asks to find the largest $F^*$ such that there is a flow $\ff$ satisfying $\mB^\top\ff = F^*\dd$ and $\bell_e \le \ff_e \le \uu_e$ for all $e \in E$. The minimum cost flow problem asks to minimize a linear cost $\cc$ over flows routing a fixed demand $\dd$. Linear algebraically, this can be written as
\[ \min_{\substack{\mB^\top\ff=\dd \\ \bell_e \le \ff_e \le \uu_e \forall e \in E}} \cc^\top\ff. \]
We work with mincost flow throughout, as it is known to generalize maxflow. We also focus on finding high-accuracy solutions in runtime depending logarithmically on $U$ and $\|\cc\|_\infty$, as it is known that this suffices to get an exact solution with linear time overhead \cite{DS08,BLLSSSW21}.

\paragraph{Electric flows and Schur complements.} Electric flows are $\ell_2$-minimization analogues of maxflow on undirected graphs, and are used in all current state-of-the-art high accuracy maxflow algorithms \cite{KLS20,AMV20,BLNPSSSW20,BLLSSSW21,GLP21} based on IPMs. On a graph $G$ with resistances $\rr$, the electric flow routing demand $\dd$ is given by
\begin{align} \argmin_{\mB^\top\ff=\dd} \sum_{e \in E} \rr_e\ff_e^2. \label{eq:electricflow} \end{align}
The minimizer in \eqref{eq:electricflow} is given by the solution to a linear system: $\ff = \mR^{-1}\mB(\mB^\top\mR^{-1}\mB)^\dagger\dd.$ The matrix $\mB^\top\mR^{-1}\mB$ is known as the \emph{Laplacian} of $G$, which can be solved in nearly-linear time \cite{ST04,KMP10,KMP11,KOSZ13,LS13,CKMPPRX14,KLPSS16,KS16,JS21}. Precisely, solving a Laplacian system gives high accuracy \emph{vertex potentials}, defined as $\pphi = (\mB^\top\mR^{-1}\mB)^\dagger\dd$.
\begin{theorem}
\label{thm:lap}
Let $G$ be a graph with $n$ vertices and $m$ edges. Let $\rr \in \R_{>0}^E$ denote edge resistances. For any demand vector $\dd$ and $\eps>0$ there is an algorithm which computes in $\O(m \log \eps^{-1})$ time \emph{potentials} $\pphi$ such that $\|\pphi - \pphi^*\|_{\mL} \le \eps\|\pphi^*\|_\mL$, where $\mL = \mB^\top \mR^{-1}\mB$ is the Laplacian of $G$, and $\pphi^* = \mL^\dagger \dd$ are the true potentials determined by the resistances $\rr$.
\end{theorem}
For notational convenience, we define the \emph{conductances} $\ww \defeq \rr^{-1}$, and let $\mL(\ww) \defeq \mB^\top\mW\mB$.

Several of our algorithms want to solve Laplacian systems in \emph{sublinear time}. This can be done in the following natural sense: instead of returning the full potential vector $\pphi$, we only wish to determine $\pphi$ restricted to a subset of vertices $C \subseteq V$. This is captured by a \emph{Schur complement}, which is defined as $\mSC(\mL, C) \defeq \mL_{CC} - \mL_{CF}\mL_{FF}^{-1}\mL_{CC}$, where $F = V \backslash C$ and $\mL_{FF}, \mL_{CF}, \mL_{FC}, \mL_{CC}$ are blocks of the Laplacian $\mL$ corresponding to rows/columns in $F, C$. Schur complements satisfy two key properties which are essential for our algorithm: they are also graph Laplacians, and they are directly related to $\mL^\dagger$ via the Cholesky factorization.
\begin{lemma}[Cholesky factorization]
\label{lemma:cholesky}
For a connected graph $G$ with Laplacian $\mL \in \R^{V \times V}$, subset $C \subseteq V$, and $F \defeq V \bs C$,
\[ \mL^\dagger = \begin{bmatrix} \mI & -\mL_{FF}^{-1} \mL_{FC} \\ 0 & \mI \end{bmatrix}
\begin{bmatrix} \mL_{FF}^{-1} & 0 \\ 0 & \mSC(\mL, C)^\dagger \end{bmatrix}
\begin{bmatrix} \mI & 0 \\ -\mL_{CF}\mL_{FF}^{-1} & \mI \end{bmatrix}\,. \]
\end{lemma}
The matrix $\begin{bmatrix}-\mL_{FF}^{-1} \mL_{FC} \\ \mI\end{bmatrix}$ appearing in the Cholesky factorization corresponds to mapping the potentials on $C$ back to the whole graph via a \emph{harmonic extension}. In other words, a random walk on $G$, with exit probabilities proportional to conductances is a martingale on potentials. We give a more formal definition and properties later in \cref{sec:er}.

Finally, it is very useful intuition that electric flows are inherently connected with the following 
 random walk on $G$: a vertex $v$ goes to a neighbor $u$ with probability proportinal to conductance (inverse resistances), i.e. $\frac{\ww_{uv}}{\sum_{w \in N(v)}\ww_{wv}}$, where $N(v)$ are the neighbors of $v$ in $G$. This random walk is the one used to define the harmonic extension, and also is used more directly in our algorithm for sampling Schur complements (see \cref{lemma:scapprox}). Throughout, any mention of random walks refers to this random walk.
\section{Improved Dynamic Schur Complements}
\label{sec:dynamicsc}
\newcommand{\Morris}{\textsc{Morris}}
\newcommand{\QCounter}{\textsc{ApproxVal() }}
\newcommand{\InitCounter}{\textsc{InitCounter() }}
\newcommand{\IncCounter}{\textsc{Increment() }}
\newcommand{\SlowMorrisWalk}{\textsc{ConceptualMorrisWalk}}
\newcommand{\MorrisWalk}{\textsc{MorrisWalk}}
\newcommand{\W}{\mathcal{W}}
\newcommand{\floor}[1]{\left\lfloor#1\right\rfloor}

In this section, we give our main algorithm for maintaining a Schur complement in a dynamic graph. Our main result is the following (see \Cref{thm:dynasc} for a more precise statement):

\begin{theorem}[Dynamic Schur complement (informal)]
\label{thm:informal_dynsc}
There is a data structure that supports the following operations against oblivious adversaries given a graph $G = (V, E)$ with dynamic edge conductances $\ww \in \R^{E(G)}$ and parameters $\beta < \eps^2 < 1$.
\begin{itemize}
\item $\textsc{Initialize}(G, \ww, \eps, \beta)$. Initializes the data structure with accuracy parameter $\eps$, and chooses a set of $O(\beta m)$ terminals $C$. $\wwbar$ is initialized as $\ww$. Runtime: $\O(m\beta^{-2}\eps^{-2})$.
\item $\textsc{AddTerminal}(v)$. Makes $v$ a terminal, i.e. $C \assign C \cup \{v\}.$ Runtime: amortized $\O(\beta^{-2}\eps^{-2})$.
\item $\textsc{Update}(e, \wwbar^\new)$. Under the guarantee that both endpoints of $e$ are terminals in $C$, updates $\wwbar_e \assign \wwbar^\new$. Runtime: amortized $\O(1)$.
\item $\textsc{SC}().$ Returns a Laplacian $\wt{\mSC}$ with $\O(\beta m\eps^{-2})$ edges which $(1+\eps)$-spectrally approximates the Schur complement of $\mL$ with terminal set $C$ in time $\O(\beta m\eps^{-2})$.
\end{itemize}
All outputs and runtimes are correct with high probability if $|C| = O(\beta m)$ at all times and there are at most $O(\beta m)$ total calls to $\textsc{Update}$.
\end{theorem}

Our proof is organized in two parts. In \Cref{subsec:morriswalk} we give an algorithm to efficiently generate useful attributes of a random walk in graphs with polynomially bounded edge weights. Next, in \Cref{subsec:sc} we describe how to use these walk attributes to maintain a Schur complement under modifications to the terminal set and edge weights. 

\subsection{Approximate Random Walks with Morris Counters}
\label{subsec:morriswalk}
Our main contribution in this section is an improved algorithm to sample random walks in weighted graphs, based on the Morris counter of \cite{Morris78a, NY20}. The main technical result of this section is the following:

\begin{theorem}[Morris Walk]
\label{thm:morris_walk}
Let $G = (V,E,w,\ell)$ be a graph with edge weights $w$ and edge lengths $\ell$ bounded between $1$ and $n^{O(1)}$. For any vertex $u$, and parameters $L, \eps \geq 0$, \Cref{alg:fast-morris} with high probability runs in $\O(L^2 + L \eps^{-2} )$ time and generates the following attributes of a random walk $\W_u$ in $G$ which starts from $u$, samples the edges it traverses with probabilities proportional to $w_e$, and stops when $\sum_{v \in \W_u} \deg(v) > L$\footnote{Here, $\deg$ denotes the \emph{unweighted} degree of a vertex in $G$.}:  
\begin{itemize}
    \item $u_1, u_2, \dots$, the $O(L)$ distinct vertices of $\W_u$ in order of their encounter.
    \item For each $u_i$, $\delta_{u_i}$ is a $(1+\eps)$-approximation of 
    \[
    \sum_{k = 1}^{f_i - 1} \ell_{(u_k, u_{k+1})},
    \]
    where $f_i$ is the index of the first visit of $u_i$ in $\W_u$. 
\end{itemize}
\end{theorem}

Our algorithm is based on simulating random walks in a graph by repeatedly solving linear systems, a technique that has been used in prior work on sampling random spanning trees and dynamically maintaining Schur complements \cite{KM09,MST14,DKPRS16,S18,DGGP19,GLP21}. However, a difficulty in applying this approach to our setting is the need to estimate the length of the resulting random walk. We address this issue by appealing to an approximate counter algorithm to estimate the length of prefixes of the walk by simulating random walks on a larger graph.

To aid our exposition, we begin by recalling a variant of the Morris counter algorithm and an improved analysis of such from \cite{NY20}, which we present in \cref{alg:morris}.
\begin{algorithm2e}[!ht]
\caption{Morris Counter \Morris \label{alg:morris}}
\SetKwProg{Globals}{global variables}{}{}
\SetKwProg{Proc}{procedure}{}{}
\SetKwIF{Prob}{ProbElseProb}{ElseProb}{with probability}{do}{else with probability}{else}{end}

\Globals{}{
    $X$: current counter value\\
    $a \geq 0$: accuracy parameter\\
}

\Proc{\InitCounter$(\eps, \delta)$}{
    $X = 0$.
    $a \assign \frac{\eps^2}{8 \log(1/\delta)}.$
}

\Proc{\IncCounter}{
	\tcp{Probabilistically updates the Morris counter $X$}
	\Prob{$(1+a)^{-X}$}{
		$X = X+ 1$
	}
}
\Proc{\QCounter}{
	\tcp{Returns unbiased estimator for number of times \IncCounter was called.}
	\Return {$\frac{1}{a} \left( (1+a)^X - 1 \right)$}
}
\end{algorithm2e}

\begin{theorem}[Modification of Theorem~1.2 from \cite{NY20}]
\label{thm:morris_counter}
Consider an instantiation of \cref{alg:morris} for parameters $\eps, \delta$, where \IncCounter has been called $N$ times after one call to \InitCounter. Then \QCounter returns a value $\hat{N}$ satisfying $\mathbb{E}[\hat{N}] = N$ and
\[
(1-\eps) N \leq \hat{N} \leq (1+\eps) N
\]
with probability $1-\delta$. 
\end{theorem}

The above theorem may be recovered directly from the analysis of Section~2.2 in \cite{NY20}. We will employ this algorithm in a white-box fashion to estimate the length of a random walk in a graph. To do this, we condense the behavior of the counter over a collection of $w$ increments into an explicit probability distribution:

\begin{definition}[Morris Increment Probabilities]
\label{def:mor_inc_prob}
Given a parameter $a \geq 0$, for integers $Y,Z$ we define the Morris increment probabilities 
\[
p_{a, Y}^{Z}(\ell) = \Pr \left( \Morris.X = Z \text{ after processing $\ell$ \IncCounter calls } | \Morris.X = Y\text{originally} \right)
\]
\end{definition}

We remark that these probabilities may be nontrivial to compute. However, in our algorithms we only require the ability to sample a $Z$ with probability proportional to $p_{a,Y}^Z(\ell)$: we will later show how this may be done efficiently.

\begin{definition}[Layer Graph]
For weighted graph $G = (V,E,w,\ell)$ and parameter $a \geq 0$, the $a$-layer graph is an (infinite) weighted directed graph $\widehat{G}$ with vertex set $\widehat{V} = V \otimes \{ 0, 1, \dots, \}$\footnote{We use the notation $A \otimes B$ to denote the Cartesian product of $A$ and $B$: it consists of all tuples $(i,j)$ for $i \in A$ and $j \in B$.} and edge set $\widehat{E}$ constructed in the following fashion: For each edge $(u,v) \in E$ of weight $w$ and length $\ell$, each $0 \leq i$, and each $i \leq j$, add a directed edge $(u,i) \rightarrow (v,j)$ of weight $w \cdot p_{a,i}^{j}\left(\floor{ \frac{8}{a} \ell} \right)$ to $\widehat{E}$.
\end{definition}

We remark that although the layer graph as defined is infinite, we only access finite subgraphs of it in our algorithms. Our proof strategy in this section is in two parts. First, we describe an idealized algorithm (\Cref{alg:slow-morris}) that directly runs a random walk in a graph and generates an output matching the requirements of \Cref{thm:morris_walk}. We then provide an efficient variant (\Cref{alg:fast-morris}) which with high probability returns an output matching that of \Cref{alg:slow-morris} in distribution: the result follows.

\begin{algorithm2e}[t!]
\caption{Conceptual Morris Walk \label{alg:slow-morris}}
\SetKwProg{Globals}{global variables}{}{}
\SetKwProg{Proc}{procedure}{}{}
\SetKwIF{Prob}{ProbElseProb}{ElseProb}{with probability}{do}{else with probability}{else}{end}

\Proc{\SlowMorrisWalk$(G, L, \eps, u, c)$}{
    $a = \frac{\eps^2}{8 \log(n^{1+c}) }$ \\
    $\widehat{G} = a-$layer graph of $G$, $u_{rw} = (u,0)$   \\
    $S_{\mathrm{visited}} = [u]$ \\
    \While{$\sum_{v \in S_{\mathrm{visited}}} \deg_G(v) \leq L$}{
        $(v_{rw} , i_{rw}) \gets$ random neighbor of $u_{rw}$ in $\widehat{G}$ \label{line:rw} \\
        \If{$v_{rw} \notin S_{\mathrm{visited}}$}{
        $S_{\mathrm{visited}} = [S_{\mathrm{visited}}; v_{rw}]$ \\
        $\delta_{v_{rw}} = \frac{1}{8} \left( (1+a)^{i_{rw}} - 1 \right)$ \\
        }
        $u_{rw} = (v_{rw}, i_{rw})$ \\
    }
    \Return{$S_{\mathrm{visited}}, \{\delta\}$}
}
\end{algorithm2e}

\begin{theorem}
\label{thm:slow_walk}
Let $G = (V,E,w,\ell)$ be a graph with edge weights $w$ and edge lengths $\ell$ bounded between $1$ and $n^{O(1)}$. For parameters $L, \eps,c \geq 0$ and starting vertex $u$, \cref{alg:slow-morris} with high probability returns the following attributes of $\W_u$, a random walk in $G$ which starts from $u$, samples the edges it traverses with probabilities proportional to $w_e$, and stops once $\sum_{v \in \W_u} \deg_G(v) \geq L$:
\begin{itemize}
    \item A set $S$ of the first $O(L)$ distinct vertices in $\W_u$, in order of encounter
    \item With probability $1-n^{-c}$, values $\{ \delta \}$ such that for each $v \in S$, $\delta_v$ $(1+\eps)$-approximates
    the length in $\W_u$ (measured with respect to $\ell$) from $u$ to the first encounter of $v$. 
\end{itemize} 
\end{theorem}

\begin{proof}
Let $a = \frac{\eps^2}{8 \log(n^{1+c})}$, and let $\widehat{\W}$ be the ordred collection of vertices $(v_{rw}, i_{rw}) \in \widehat{G}$ encountered on \Cref{line:rw}: note that these vertices form  a random walk on $\widehat{G}$ by construction. We define an auxillary random walk $\W$ in $G$ as follows: if the $k^{th}$ in $\widehat{\W}$ is $(v,i)$, the $k^{th}$ node in $\W$ is $v$. We will show the following two facts: 
\begin{itemize}
    \item $\W$ is a random walk in $G$ which starts from $u$, samples its edges with probabilities proportional to $w_e$, and stops once $\sum_{v \in \W} \deg_G(v) \geq L$. 
    \item For any $k$, let the $k^{th}$ node in $\widehat{\W}$ be $(v_k,i_k)$, and let $R_k$ be inductively defined by $R_1 = 0$, $R_{i+1} = R_i + \floor{ \frac{8}{a} \ell_{(v_i, v_{i+1})} }$.  Then $i_k$ is distributed as $\Morris(a).X$ after processing $R_k$ increments. 
\end{itemize}
The first of these claims follows immediately: if we sample a random neighbor of $(v,i)$ in $\widehat{G}$, the probability that it is of the form $(v', j)$ for some $j$ is simply 
\[
\frac{\sum_{j=0}^\infty w_{(v',v)} p_{a,i}^j \left( \floor{ \frac{8}{a} \ell_{(v',v)} } \right)}{\sum_{x \in V} \sum_{j=0}^\infty w_{(x,v)} p_{a,i}^j \left( \floor{ \frac{8}{a} \ell_{(x,v)} } \right)} =  \frac{w_{(v',v)}}{\sum_{w \in V} w_{(x,v)}}. 
\]
Thus the $k^{th}$ element of $\W$ is a neighbor of the $(k-1)^{st}$ sampled proportional to $w$: since $\W$ starts from $u$ the claim follows. 

For the second claim, we proceed by induction on $k$. The claim is trivially true for $k=1$ (as the first node in $\widehat{\W}$ is $(u,0)$). It remains to show the induction step. Let the $k^{th}$ node of $\widehat{\W}$ be $(v_k,i_k)$: by the induction hypothesis $i_k$ is distributed as $\Morris(a).X$ after $R_k$ increments. Now conditioned on the value of $v_{k+1}$, we have
\begin{align*}
& \Pr\left(i_{k+1} = x\right) = \sum_{y = 0}^x \Pr \left(i_{k+1} = x | i_k = y\right) \Pr \left(i_k = y \right) \\
&= \sum_{y=0}^x \frac{ w_{(v_{k}, v_{k+1})} p_{a,y}^x \left( \floor{ \frac{8}{a} \ell_{(v_{k}, v_{k+1})} } \right)}{\sum_{z = 0}^\infty  w_{(v_{k}, v_{k+1})} p_{a,y}^z \left( \floor{ \frac{8}{a} \ell_{(v_{k}, v_{k+1})} } \right) } \Pr \left(i_k = y \right) =  \sum_{y=0}^x p_{a,y}^x \left( \floor{ \frac{8}{a} \ell_{(v_{k}, v_{k+1})} } \right) \Pr \left(i_k = y \right).
\end{align*}

But by the induction hypothesis, each term of the expression is the probability that $\Morris(a).X$ equals $y$ after $R_k$ increments and also equals $y$ after a further $\floor{ \frac{8}{a} \ell_{(v_k , v_{k+1} ) } }$ increments. Since $R_{k+1} = R_k + \floor{ \frac{8}{a} \ell_{(v_k , v_{k+1} ) } }$ conditioned on the value of $v_{k+1}$, the claim follows by the law of total probability.

We finally show how these claims imply the theorem. First, note that $S$ consists of the vertices in $\W$ in the order of their encounter: since $\W$ is a random walk in $G$ the correctness of $S$ follows. Next, for each $v \in S$ let $u_{v} = (v, i_v)$ be the value of $u_{rw}$ set on \Cref{line:rw} where $v$ was first encountered. Observe that each edge in $G$ has weight at least $1$: thus 
\[
\left(1- \frac{a}{8}\right) \ell_{(v_k , v_{k+1})} \leq \ell_{(v_k , v_{k+1})} - \frac{a}{8} \leq \frac{a}{8} \floor{ \frac{8}{a} \ell_{(v_k , v_{k+1} )} } \leq \ell_{(v_k , v_{k+1})}.
\]
Thus for any $k$, 
\[
\left( 1 - \frac{a}{8} \right) \frac{8}{a} \sum_{i=0}^k \ell_{(v_i, v_{i+1})} \leq R_k \leq \frac{8}{a} \sum_{i=0}^k \ell_{(v_i, v_{i+1})}.
\]
By the second claim, we see that $i_v$ is distributed as $\Morris(a).X$ after processing $N_{uv}$ increments, where $N_{uv} = R_k$ if $k$ is the smallest index where $v$ appears in $\W$. This number of increments is larger than $\frac{8}{a}$: by \Cref{thm:morris_counter} we thus have 
\[
\Pr \left( \left| \frac{1}{a} \left( (1+a)^{i_v} - 1\right) - \frac{8}{a} N_{uv} \right| \geq \frac{8\eps}{a} N_{uv} \right) \leq 1 - n^{-3}.
\]
But now, $\frac{a}{8} N_{uv}$ is within a $1+\frac{a}{4} \leq 1+\eps$ factor of $\mathcal{L}_{uv}$, the length of $\W$ from $u$ to the first visit of $v$. Thus,
\[
\Pr \left( \left| \delta_v - \mathcal{L}_{uv} \right| \geq 2 \eps \mathcal{L}_{uv} \right) \leq 1- n^{-C-1}.
\]
As there are at most $n$ vertices in $S$, the claim follows by scaling down $\eps$ and union bounding over these failure probabilities. 
\end{proof}

With \Cref{thm:slow_walk} in hand, we prove the main result of this section by giving an efficient implementation of \Cref{alg:slow-morris}. Our algorithm works by simulating a random walk over the $a$-layer graph $\widehat{G}$ using a Laplacian linear system solver. We will employ the following (standard) lemma on the hitting probabilities of a random walk in undirected graphs:

\begin{lemma}[Corollary of Lemma~5.6 from \cite{GLP21}]
    \label{lem:sample_exit}
Let $G = (V,E,w)$ be a weighted undirected graph, and let $x$ be any vertex in $V$. For any $C \subseteq V$,  the probability that a random walk starting from $x$ first enters $C$ at a vertex $y$ is given by  
\[
-\left[\mL_{C, V\backslash C} \mL^{-1}_{C,C} \pmb{\chi}_x \right]_y.
\]
Thus, in $\O(|E|)$ time we may sample a vertex $y \in C$ with probability equal to a random walk starting from $x$ first entering $C$ at $y$. 
\end{lemma}

We will use this fact within the framework of $\SlowMorrisWalk$ to replace the explicit sampling of random walk (which may take $\poly(n,W)$ time) with a computationally efficient subroutine. We now describe the graphs on which we apply \Cref{lem:sample_exit}:
In the below, $N_G(S)$ denotes the vertices which are neighbors of $S$ but do not themselves belong to $S$. 

\begin{definition}[$(a,\iota,S)$-shortcut graph]
\label{def:ALS-shortcut}
Let $G = (V,E,w,\ell)$ be an undirected graph with edge weights $w$ and edge lengths $\ell$. Let $a \geq 0$ be a parameter, and let $\iota \geq 0$ be an integer. For $S \subseteq V$, we define the $(a,\iota,S)$-shortcut graph $H = (V_H, E_H, w_H)$ as follows:
\begin{itemize}
    \item For each $v \in S \cup N_G(S)$, add $v$  to $V_H$.
    \item For each edge $e = (u,v) \in E$ of weight $w$ and length $\ell$ with $u, v \in S$, add vertices $v^+_u, u^+_v$ to $V_H$, an edge $(u,v)$ of weight $w \cdot p_{a,\iota}^\iota \left( \floor{ \frac{8}{a} \ell } \right)$ to $E_H$, and edges $(u,v^+_u), (u^+_v, v)$ of weight $w \cdot \left( 1- p_{a,\iota}^\iota \left( \floor{ \frac{8}{a} \ell } \right) \right)$ to $E_H$.  
    \item For each edge $(u,v) \in E$ of weight $w$ and length $\ell$ with $u \in S$, $v \in N(S)$, add a vertex $v_u^+$ to $V_H$, an edge $(u,v)$ of weight $w \cdot p_{a,\iota}^\iota \left( \floor{ \frac{8}{a} \ell }\right)$ to $E_H$, and an edge $(u,v^+_u)$ of weight $w \cdot \left( 1- p_{a,\iota}^\iota \left( \floor{ \frac{8}{a} \ell } \right) \right)$ to $E_H$. 
\end{itemize}
Let $V_S^+$ denote the set of vertices of the form $v_u^+ \in H$ for $v \in S$, $N_S$ denote vertices of the form $v \in H$ for $v \in N_G(S)$, and $N_S^+$ denote vertices of the form $v_u^+ \in H$ for $v \in N_G(S)$. 
\end{definition}

Note that computing the shortcut graph defined above only requires computing Morris increment probabilities of the form $p_{a, \iota}^\iota(s)$. We will show that this admits a simple closed form, and that we may sample a variable proportional to the increment probabilities efficiently. 
\begin{lemma}
\label{lem:pfacts}
Given a parameter $a \geq 0$ and integers $Y, \ell$, the Morris increment probabilities (\Cref{def:mor_inc_prob}) satisfy
\[
p_{a, Y}^Y(\ell) = \left(1 - (1+a)^{-Y} \right)^\ell.
\]
In addition, we may sample an integer $Z \geq Y+1$ such that 
\[
\Pr \left( Z = \Gamma \right) = \frac{ p_{a, Y}^\Gamma (\ell) }{1- p_{a, Y}^Y (\ell)}
\]
in time $\O(Z - Y)$. 
\end{lemma}

\begin{proof}
For the first claim, note that each call to $\IncCounter$ increments $\Morris.X$ with probability $(1+a)^{-Y}$. The probability that $\ell$ such increments fails to increase $\Morris.X$ is therefore $\left(1 - (1+a)^{-Y} \right)^\ell$ as claimed.

For the second claim, we describe an algorithm to sample from the claimed distribution. We first observe that the desired distribution is precisely the value of $\Morris.X$ after processing $\ell$ increments, conditioned on 
\begin{itemize}
    \item The initial value of $\Morris.X$ was $Y$.
    \item The final value of $\Morris.X$ is strictly larger than $Y$.
\end{itemize}
We will sample from this distribution by implicitly simulating the Morris counter algorithm itself: we repeatedly sample from the distribution over the number of $\IncCounter$ calls required to increase $\Morris.X$, and return the final value of $\Morris.X$ after $\ell$ simulated increments were processed. For the below, we let $Geom(p)$ denote the geometric random variable over $\{1, 2, \dots \}$ with failure probability $p$ and let $Geom^k(p)$ denote $Geom(p)$ conditioned on the output being at most $k$: note that both distributions may be sampled from in $\O(1)$ time. 

Assume that $Y' = \Morris.X$ at some point. Let $p_{Y'}$ be a random variable representing the number of $\IncCounter$ calls required to increase $\Morris.X$: note that 
\[
\Pr\left( p_{Y'} > s \right) = p_{a, Y'}^{Y'}(s)
\]
by definition. By the closed-form representation of these probabilities, we may therefore conclude that $p_{Y'}$ is distributed as $Geom((1+a)^{-Y'})$. 

By the definition of $\Morris$, it is therefore clear that we may sample $Z ~ p_{a,Y}^Z(\ell)$ by repeating the following operations:
\begin{itemize}
    \item Initialize a running increment counter $\ell' = 0$ and a counter value $Z = Y$.
    \item Generate a sample $k_Z \sim Geom( (1+a)^{-Z} )$ and set $\ell' = \ell' + k_Z$. 
    \item If $\ell' > \ell$, return $Z$. Else, increment $Z$ by $1$ and go back to the previous line. 
\end{itemize}
To sample $Z$ conditioned on $Z \neq Y$, it is thus sufficient to sample the first $k_Y \sim Geom^\ell( (1+a)^{-Z} )$ to ensure $Z$ is not incremented $0$ times. To bound the running time, we additionally observe that the number of geometric and truncated geometric random variables sampled is proportional to $Z-Y$: as the total work performed is $\O(1)$ times this the claim follows. 
\end{proof}

\begin{algorithm2e}[t!]
\caption{Morris Walk \label{alg:fast-morris}}
\SetKwProg{Globals}{global variables}{}{}
\SetKwProg{Proc}{procedure}{}{}
\SetKwIF{Prob}{ProbElseProb}{ElseProb}{with probability}{do}{else with probability}{else}{end}

\Proc{\MorrisWalk$(G, L, \eps, u)$}{
    $a = \frac{\eps^2}{8 \log(n^3) }$ \\
    $S= [u]$, $\iota = 0$, $u_{rw} = u$ \\
    \While{$\sum_{v \in S} \deg_G(v) \leq L$}{
        $G_S^\iota \gets (a,\iota,S)$-shortcut graph for $G$ (\Cref{def:ALS-shortcut}) \\
        $C =  V_S^+ \cup N_S \cup N_S^+$ \\
        $x \gets$ vertex sampled with probability a random walk starting from $u_{rw}$ in $G_S^\iota$ first enters $C$ at $x$ (\Cref{lem:sample_exit}) \label{line:sample_exit} \\
        \If{$x \in N_S$ \Comment{Added new vertex to $S$}}{
        $S= [S; x]$ \Comment{Interpret $x$ as a vertex in $G$} \\
        $\delta_{x} = \frac{1}{8} \left( (1+a)^{\iota} - 1 \right)$ \\
        $u_{rw} = x$ \\
        }
        \If{$x \in V_S^+$ \Comment{Incremented $\iota$}}{
        $v_s^+ \defeq x$\\
        $\ell \defeq $ length of edge $(s,v) \in G$ \\
        $\iota \gets \iota'$ sampled with probability $\propto  p_{a,\iota}^{\iota'}(\floor{ \frac{8}{a} \ell })$, conditioned on $\iota' > \iota$ (\Cref{lem:pfacts}) \label{line:sample1} \\
        $u_{rw} = v$ \\
        }
        \If{$v \in N_S^+$ \Comment{Incremented $\iota$ and added vertex to $S$}}{
        $v_s^+ \defeq x$\\
        $\ell \defeq $ length of edge $(s,v) \in G$ \\
        $S = [S, v]$ \\
        $\iota \gets \iota'$ sampled with probability $\propto  p_{a,\iota}^{\iota'}(\floor{ \frac{8}{a} \ell })$, conditioned on $\iota' > \iota$ (\Cref{lem:pfacts}) \label{line:sample2} \\
        $\delta_{v} = \frac{1}{8} \left( (1+a)^{\iota} - 1 \right)$ \\
        $u_{rw} = v$
        }
    }
    \Return{$S, \{\delta\}$}
}
\end{algorithm2e}

\begin{proof}[Proof of \Cref{thm:morris_walk}]
Our proof proceeds in two steps. We will first show that the vertices added to $S$ and the values $\{ \delta \}$ have the same distribution as the output of \Cref{alg:slow-morris}. We will then bound the runtime of the algorithm.

Let $\widehat{G}$ be the $a$-layer graph of $G$, and fix a parameter $\iota$ and visited set $S$ during a single iteration of \Cref{alg:fast-morris}. We consider the subgraph $\widehat{G}_{\iota, S}$ consisting of all directed edges with tail of the form $(v, \iota)$ for $v \in S$. Consider the process of running a random walk from $(v, \iota) \in \widehat{G}_{\iota, S}$ until a vertex not of the form $(u, \iota)$ with $u \in S$ is reached. It is self-evident that the only such vertices in $\widehat{G}_{\iota, S}$ belong to three classes:

\begin{itemize}
    \item $(v, \iota)$ where $v \in N_G(S)$
    \item $(v, \iota')$ where $v \in S$ and $\iota' > \iota$
    \item $(v, \iota')$ where $v \in N_G(S)$ and $\iota' > \iota$. 
\end{itemize}

Let $C$ denote the collection of vertices of this type. Note that the subgraph of $\widehat{G}_{\iota, S}$ induced on vertices of the form $(v, \iota)$ for $v \in S$ is essentially undirected (since each directed edge $(x,y)$ is matched by a directed edge $(y,x)$ of the same weight). Let $G_{\iota, S}$ be the graph obtained by replacing these parallel directed edges with an undirected edge of the same weight, and by removing edge directions from all other edges. It is clear that the probability distribution over vertices that a random walk starting from $(u, \iota)$ for $u \in S$ enters $C$ at is induced by a Laplacian linear system solve via \Cref{lem:sample_exit}. By direct calculation, it may be verified that these sampling probabilities are equivalent to the sampling performed on \Cref{line:sample_exit}: when sampling the number of increments to the counter, \Cref{alg:fast-morris} simply samples the event that the counter is incremented at least once and then samples from the appropriate conditional distribution for the true number of increments to apply. 

We now bound the running time of our algorithm. We observe that the termination condition of the while loop ensures that $G_S^\iota$ never contains more than $O(L)$ edges: thus the call to \Cref{lem:sample_exit} on on \Cref{line:sample_exit} can be implemented in $\O(L)$ time. We additionally see via the remaining operations in the loop that each linear system we solve ensures that we either add a new vertex to $S$ or increase the value of $\iota$. Next, we note that since $G$'s weights and lengths are polynomially-bounded, the total length of a random walk which covers the entirety of $G$ is bounded by $\poly(n)$. Thus for any $v$ in the returned set $S$, $\delta_v$ is a $(1+\eps)$-approximation to a quantity which is also bounded by $\poly(n)$. But this implies that the variable $\iota$ satisfies 
\[
(1+a)^\iota \leq \poly(n) \implies \iota \leq \O \left( \eps^{-2} \right) 
\]
with high probability. Thus at most $\O(\eps^{-2})$ calls to \Cref{lem:sample_exit} can increase the value of $\iota$: as the while loop must terminate after adding $L$ vertices to $S$ it follows that \Cref{alg:fast-morris} solves at most $\O(L + \eps^{-2})$ linear systems with high probability. Finally, the only remaining nontrivial computation of the algorithm is performed on \Cref{line:sample1} and \Cref{line:sample2}. But as $\iota \leq \O(\eps^{-2})$ by \Cref{lem:pfacts} these lines cost $\O(\eps^{-2})$ amortized over the whole algorithm. The claimed runtime follows. 
\end{proof}

\subsection{Improved Dynamic Schur Complement}
\label{subsec:sc}

Here we provide our main result regarding the dynamic maintenance of Schur complements under edge resistance changes in $G$. We achieve this by plugging in our improved algorithm \cref{thm:morris_walk} for estimating lengths of random walks visiting a fixed number of vertices into previous frameworks \cite{DGGP19,GLP21}. Below, the additional operation $\textsc{InitialSC}$ maintains the approximate Schur complement ignoring edge updates, but still tracking terminal additions. It is useful for our dynamic \textsc{Evaluator} and \textsc{Locator} data structures in \cref{sec:er}. We use the notation $\mSC_{\HH}$ for the approximation as it eventually gets used to approximately compute a harmonic extension $\HH$.

\begin{restatable}[Dynamic Schur complement]{theorem}{dynsc}
\label{thm:dynasc}
There is a data structure $\DynamicSC$ that supports the following operations against oblivious adversaries given a graph $G = (V, E)$ with dynamic edge conductances $\ww \in \R^{E(G)}$ and parameters $\beta < \eps^2 < 1$.
\begin{itemize}
\item $\textsc{Initialize}(G, \ww, \eps, \beta)$. Initializes the data structure with accuracy parameter $\eps$, and chooses a set of $O(\beta m)$ terminals $C$. $\wwbar$ is initialized as $\ww$. Runtime: $\O(m\beta^{-2}\eps^{-2})$.
\item $\textsc{AddTerminal}(v)$. Makes $v$ a terminal, i.e. $C \assign C \cup \{v\}.$ Runtime: amortized $\O(\beta^{-2}\eps^{-2})$.
\item $\textsc{Update}(e, \wwbar^\new)$. Under the guarantee that both endpoints of $e$ are terminals in $C$, updates $\wwbar_e \assign \wwbar^\new$. Runtime: amortized $\O(1)$.
\item $\textsc{SC}().$ Returns a Laplacian $\wt{\mSC} \approx_{\eps} \mSC(\mL(\wwbar), C)$ with $\O(\beta m\eps^{-2})$ edges in time $\O(\beta m\eps^{-2})$.
\item $\textsc{InitialSC}().$ Returns a Laplacian $\wt{\mSC}_{\HH}$ with $\O(\beta m\eps^{-2})$ edges in time $\O(\beta m\eps^{-2})$.  Let $Z$ be the set of edges which were input to $\textsc{Update}$ after initialization. Define $\ww_{\overline{Z}}$  as $(\ww_{\overline{Z}})_e=0$ for $e\in Z$ and $(\ww_{\overline{Z}})_e=\ww_e$ otherwise. Then  $\wt{\mSC}_{\HH}$ satisfies \begin{align} \mSC(\mL(\ww), C) - \eps \mSC(\mL(\ww_{\overline{Z}}), C) \pe \wt{\mSC}_{\HH} \pe \mSC(\mL(\ww), C) + \eps \mSC(\mL(\ww_{\overline{Z}}), C). \label{eq:scthmblah} \end{align}
\end{itemize}
All outputs and runtimes are correct whp. if $|C| = O(\beta m)$ at all times and there are at most $O(\beta m)$ total calls to $\textsc{Update}$.
\end{restatable}

We note that we could achieve the tighter approximation guarantee in \eqref{eq:scthmblah} for the operation $\textsc{SC}()$. However, we do not need use it in this paper (eg. \cref{sec:er}) and therefore, do not state it.

We require the following process which samples Schur complements using random walks.
\begin{lemma}[Schur complement approximation, \cite{DGGP19} Theorem 3.1]
\label{lemma:scapprox}
Let $G = (V, E, \rr)$ be an undirected, weighted multigraph with a subset of vertices $C$. Furthermore, let $\eps \in (0, 1)$ and let $\rho = 1000\eps^{-2}\log n$. Let $H$ be an initially empty graph with vertices $C$, and for each edge $e = (u, v) \in E(G)$ repeat the following procedure $\rho$ times.
\begin{enumerate}
    \item Simulate a random walk from $u$ until it hits $C$ at $c_1$.
    \item Simulate a random walk from $v$ until it hits $C$ at $c_2$.
    \item Combine these random walks (along with edge $e = (u, v)$) to form a walk $W$.
    \item Add edge $(c_1, c_2)$ to $H$ with resistance $\rho\sum_{e \in W} \rr_e.$
\end{enumerate}
The resulting graph $H$ satisfies $\mL(H) \approx_{\eps} \mSC(\mL(\ww), C)$ with probability at least $1 - n^{-10}.$
\end{lemma}

Finally, we require a dynamic spectral sparsification procedure.
\begin{lemma}[\!\!{\cite[Lemma 4.10]{GLP21}}]
\label{lemma:sparsifier}
There is a data structure that supports insertions and deletions of edges on a graph $G$ which have underlying conductances/resistances in amortized $\O(\log U)$ time per operation. Additionally, it can output a $(1+\eps)$-spectral sparsifier of $G$ in $\O(n\eps^{-2}\log U)$ time.
\end{lemma}
Now, we can show \cref{thm:dynasc} exactly as done in \cite{DGGP19,GLP21} by sampling random walks using \cref{lemma:scapprox} and shortcutting them as terminals get added.
\begin{proof}[Proof of \cref{thm:dynasc}]
We explain how to implement each operation in \cref{thm:dynasc}.

\textsc{Initialize}: Randomly sample an initial terminal set $C$ of size $O(\beta m)$. From each edge $e = (u, v) \in G$, sample $\rho = \O(\eps^{-2})$ random walks from $u, v$ to $C$ as in \cref{lemma:scapprox}, and record $(1+\eps)$ approximations of the sums of resistances of all prefixes. Note that these walks visit $\O(\beta^{-1})$ distinct vertices whp. Initialize the data structure $\dspar$ in \cref{lemma:sparsifier}. Based on these random walks, add edges to $C$ using the data structure $\dspar$. Additionally, maintain a set $Z$ of updated edges, whose original and final conductances we track explicitly.

The runtime of \textsc{Initialize} is dominated by the time to sample the random walks, which is $\O(m\eps^{-2}(\beta^{-2} + \beta^{-1}\eps^{-2})) = \O(m\eps^{-2}\beta^{-2})$ by \cref{thm:morris_walk} (the length $L = \O(\beta^{-1})$) and $\beta < \eps^2$.

$\textsc{AddTerminal}(v)$: Update $C \assign C \cup \{v\}$ and shortcut all walks passing through $v$. The total length of all walks is $\O(m\eps^{-2}\beta^{-1})$, so over the course of $O(\beta m)$ terminal insertions, the amortized runtime is $\O(m\eps^{-2}\beta^{-1}/(\beta m)) = \O(\beta^{-2}\eps^{-2})$. Finally, pass all edge insertions/deletions in $C$ to $\dspar$.

$\textsc{Update}(e, \wwbar^\new)$: Delete the edge $e$ (do not insert an edge with conductance $\wwbar^\new$), and pass the deletion to $\dspar$. Update $Z \assign Z \cup \{e\}.$ From now on, the algorithm explicitly stores in memory the original and current conductances of edge $e$. Clearly, the update time is $\O(1)$.

$\textsc{SC}()$: Call $\dspar$ to output a $(1+\eps)$-approximation of $\mSC(\mL(\ww_{\overline{Z}}), C)$ with high probability. The approximation guarantee follows from \cref{lemma:scapprox} and the guarantee of \cref{thm:morris_walk} that the total resistive length of each random walk is correct up to $(1+\eps/10)$ with high probability. Finally, add the edges $e \in Z$ back in with the current conductances. The runtime is $\O(\beta m\eps^{-2})$ by \cref{lemma:sparsifier} as $|C| = O(\beta m)$.

$\textsc{InitialSC}()$: Same as $\textsc{SC}()$, except we add back edges in $Z$ with their original conductances. The tighter approximation holds because the algorithm is returning a $(1+\eps)$-approximation of $\mSC(\mL(\ww_{\overline{Z}}), C)$ and the edges $e \in Z$ that are added in contribute no error.
\end{proof}
\section{Data Structures for Dynamic Electrical Flows}
\label{sec:er}

The goal of this section is to apply the dynamic Schur complement data structure of \cref{thm:dynasc} to give algorithms that dynamically maintain electric potentials and edges with large electric energies in dynamic electrical flows. In \cref{subsec:harmonic}, we will introduce the harmonic extension and use it to decompose the energy vector we need to maintain for the outer IPM. In \cref{subsec:potential}, we show how to maintain a potential vector which is a key component for the following subsections. In \cref{subsec:evaluator}, we build the \textsc{Evaluator} that estimates the energy of any edge. In \cref{subsec:locator}, we build the \textsc{Locator} that returns a superset of edges with large energies.

\subsection{Harmonic Extension}
\label{subsec:harmonic}

A key notion we use throughout is the harmonic extension, 
which is a linear operator that maps the potentials restricted to a terminal set to the full electric potentials $\pphi$. We use $\IP$ to denote the projection orthogonal to the all-ones vector.
\begin{definition}[Harmonic extension]
\label{def:h}
For a graph $G=(V, E)$ with edge conductances $\ww \in \R^E_{>0}$ and $C \subseteq V(G)$, define the \emph{harmonic extension} operator $\HH_C \in \R^{V(G) \times C}$ as
\[ \HH_C \defeq \begin{bmatrix} -\mL(\ww)_{FF}^{-1}\mL(\ww)_{FC} \IP\\ \IP \end{bmatrix} .\]
\end{definition}
Note that the harmonic extension does not depend on edges with both endpoints in $C$. Leveraging this yields the following alternative characterization of the harmonic extension. These properties are crucial for our data structures as they maintain a growing terminal set where are resistance changes are on edges completely inside the terminal set. In this section, we use $\wwbar$ to denote modified conductances and $\ww$ to denote initial conductances.

\begin{lemma}[Alternate definition of harmonic extension]
\label{lemma:htosc}
For a graph $G=(V, E)$ with edge conductances $\ww \in \R^E_{>0}$ and $C \subseteq V$, let $\wt{G}$ be a graph with the same edge set as $G$ whose conductances $\wwtil$ agree with $\ww$ \emph{except} potentially on edges with both endpoints inside $C$. Then
\begin{align} \HH_C = \mL(\wwtil)^\dagger\begin{bmatrix} 0 \\ \mSC(\mL(\wwtil), C) \end{bmatrix}. \label{eq:hformula} \end{align}
\end{lemma}

\begin{proof}
By \cref{def:h}, the harmonic extension does not depend on the edges inside $C$. Hence, we can simply show the lemma for the Laplacian $\mL = \mL(\wwtil)$. By the Cholesky factorization (\cref{lemma:cholesky}), we have
\begin{align*} \mL^\dagger\begin{bmatrix} 0 \\ \mSC(\mL, C) \end{bmatrix} &= \begin{bmatrix} \mI & -\mL_{FF}^{-1} \mL_{FC} \\ 0 & \mI \end{bmatrix}
\begin{bmatrix} \mL_{FF}^{-1} & 0 \\ 0 & \mSC(\mL, C)^\dagger \end{bmatrix}
\begin{bmatrix} \mI & 0 \\ -\mL_{CF}\mL_{FF}^{-1} & \mI \end{bmatrix} \begin{bmatrix} 0 \\ \mSC(\mL, C) \end{bmatrix} \\
&= \begin{bmatrix} \mI & -\mL_{FF}^{-1} \mL_{FC} \\ 0 & \mI \end{bmatrix}\begin{bmatrix} \mL_{FF}^{-1} & 0 \\ 0 & \mSC(\mL, C)^\dagger \end{bmatrix}
\begin{bmatrix} 0 \\ \mSC(\mL, C) \end{bmatrix}\\
&= \begin{bmatrix} \mI & -\mL_{FF}^{-1} \mL_{FC} \\ 0 & \IP \end{bmatrix}
\begin{bmatrix} 0 \\ \IP \end{bmatrix}
 = \HH_C.
\end{align*}
\end{proof}
This is why we use the notation $\HH_C$ without reference to $G$ -- when we use $\HH_C$  in our dynamic data structures all changed edges will lie inside $C$. Consequently, the actual graph (beyond initialization) does not affect $\HH_C$!

The inverse of the Laplacian can be represented by a contribution from the Schur complement, plus $\mL_{FF}^{-1}$. This is essentially just a restatement of the Cholesky factorization (\cref{lemma:cholesky}).
\begin{lemma}
\label{lemma:HSH}
Let $G=(V, E, \ww)$ be a graph. Then
\[
\mL(\ww)^\dagger=\HH_C\mSC(\mL(\ww), C)^\dagger\HH_C^\top+\left[\begin{array}{cc}\mL(\ww)_{F, F}^{-1} & 0 \\ 0 & 0\end{array}\right].
\]
\end{lemma}
\begin{proof}
The Cholesky factorization (\cref{lemma:cholesky}) says that
\begin{align*} \mL(\ww)^\dagger &= \begin{bmatrix} \mI & -\mL_{FF}^{-1} \mL_{FC} \\ 0 & \mI \end{bmatrix}
\begin{bmatrix} \mL_{FF}^{-1} & 0 \\ 0 & \mSC(\mL, C)^\dagger \end{bmatrix}
\begin{bmatrix} \mI & 0 \\ -\mL_{CF}\mL_{FF}^{-1} & \mI \end{bmatrix} .\\
\end{align*}
As $\IP\mSC(\mL, C)^\dagger\IP=\mSC(\mL, C)^\dagger$, the equation above is equal to 
\begin{align*}
&\begin{bmatrix} \begin{bmatrix} \mI \\ 0 \end{bmatrix} & \HH_C \end{bmatrix}
\begin{bmatrix} \mL_{FF}^{-1} & 0 \\ 0 & \mSC(\mL, C)^\dagger \end{bmatrix}
\begin{bmatrix} \begin{bmatrix} \mI & 0 \end{bmatrix} \\ \HH_C^\top \end{bmatrix} \\
&= \HH_C\mSC(\mL(\wwbar), C)^\dagger\HH_C^\top+\begin{bmatrix} \mI \\ 0 \end{bmatrix}\mL(\ww)_{F, F}^{-1}\begin{bmatrix} \mI & 0 \end{bmatrix}\\
&= \HH_C\mSC(\mL(\wwbar), C)^\dagger\HH_C^\top+ \left[\begin{array}{cc}\mL(\ww)_{F, F}^{-1} & 0 \\ 0 & 0\end{array}\right]
\end{align*}
\end{proof}
Let $\vv \in \R^E$ be a (dynamic) vector. To implement the outer IPM, we must be able to maintain a heavy-hitter sketch on the following vector
\[\mPi\vv \defeq \mPi(\wwbar)\vv = \mWbar^{1/2}\mB\mL(\wwbar)^\dagger \mB^\top\mWbar^{1/2}\vv.
\]
For this, we decompose $\mPi\vv$ into three terms. Let $\vvhat$ be any vector that agrees with $\vv$ in $F=V\setminus C$. Then
We first decompose $\mPi\left(\vv\right)$ by \cref{lemma:HSH}.
\begin{align*}
\mPi\left(\vv\right)
=&
\mWbar^{1/2}\mB\mL\left(\wwbar\right)^\dagger \mB^\top\mWbar^{1/2}\vv\\
=& \mWbar^{1/2}\mB\left(\HH_C\mSC(\mL(\wwbar), C)^\dagger\HH_C^\top+\left[\begin{array}{cc}\mL(\wwbar)_{F, F}^{-1} & 0 \\ 0 & 0\end{array}\right]\right)\mB^\top\mWbar^{1/2}\vv \tag{by \cref{lemma:HSH}}\\
= &\mWbar^{1/2}\mB\HH_C\mSC(\mL(\wwbar), C)^\dagger\HH_C^\top\mB^\top\mWbar^{1/2}\vv+\mWbar^{1/2}\mB\left[\begin{array}{cc}\mL(\wwbar)_{F, F}^{-1} & 0 \\ 0 & 0\end{array}\right]\mB^\top\mWbar^{1/2}\vvhat.\\
\end{align*}
Since $\ww_C$ does not affect the value of the second term, we have 
\begin{align*}
&\mWbar^{1/2}\mB\HH_C\mSC(\mL(\wwbar), C)^\dagger\HH_C^\top\mB^\top\mWbar^{1/2}\vv+\mWbar^{1/2}\mB\left[\begin{array}{cc}\mL(\wwbar)_{F, F}^{-1} & 0 \\ 0 & 0\end{array}\right]\mB^\top\mWbar^{1/2}\vvhat\\
=&\mWbar^{1/2}\mB\HH_C\mSC(\mL(\wwbar), C)^\dagger\HH_C^\top\mB^\top\mWbar^{1/2}\vv+\mW^{1/2}\mB\left[\begin{array}{cc}\mL(\ww)_{F, F}^{-1} & 0 \\ 0 & 0\end{array}\right]\mB^\top\mW^{1/2}\vvhat. \\
\end{align*}
Then we use \cref{lemma:HSH} in the other direction to get 
\begin{align*}
&\mWbar^{1/2}\mB\HH_C\mSC(\mL(\wwbar), C)^\dagger\HH_C^\top\mB^\top\mWbar^{1/2}\vv+\mW^{1/2}\mB\left[\begin{array}{cc}\mL(\ww)_{F, F}^{-1} & 0 \\ 0 & 0\end{array}\right]\mB^\top\mW^{1/2}\vvhat \\
 =&\mWbar^{1/2}\mB\HH_C\mSC(\mL(\wwbar), C)^\dagger\HH_C^\top\mB\mWbar^{1/2}\vv
+\mW^{1/2}\mB\left(\mL(\ww)^\dag-\HH_C\mSC(\mL(\ww), C)^\dagger\HH_C^\top\right)\mB^\top\mW^{1/2}\vvhat\tag{by \cref{lemma:HSH}}\\
 =&\mWbar^{1/2}\mB\HH_C\mSC(\mL(\wwbar), C)^\dagger\HH_C^\top\mB^\top\mWbar^{1/2}\vv
-\mW^{1/2}\mB\HH_C\mSC(\mL(\ww), C)^\dagger\HH_C^\top\mB^\top\mW^{1/2}\vvhat\\
+&\mW^{1/2}\mB\mL(\ww)^\dag\mB^\top\mW^{1/2}\vvhat.
 \end{align*} 
 We will use 
 \begin{equation}
 \begin{aligned}
 \mPi\vv 
 =\mWbar^{1/2}\mB\HH_C\mSC(\mL(\wwbar), C)^\dagger\HH_C^\top\mB^\top\mWbar^{1/2}\vv
-\mW^{1/2}\mB\HH_C\mSC(\mL(\ww), C)^\dagger\HH_C^\top\mB^\top\mW^{1/2}\vvhat\\
+\mW^{1/2}\mB\mL(\ww)^\dag\mB^\top\mW^{1/2}\vvhat \label{eq:expansion}
 \end{aligned}
 \end{equation} in two cases 
 \begin{itemize}
 \item where $\vvhat=\vv$, and
 \item where $\vv$ being the current vector and $\vvhat$ being the initial $\vv$.
 \end{itemize}

At a high level, our approach will use several spectral approximations of the RHS of \eqref{eq:expansion}. We will replace $\mSC(\mL(\wwbar), C)$ with an approximate Schur complement using \cref{thm:dynasc}. Additionally, we will replace $\HH_C$ and $\HH_C^\top$ by replacing the Schur complements in \eqref{eq:hformula} with approximate Schur complements given by \cref{thm:dynasc}.

We now focus on approximating the ``right'' of the first two terms of the RHS \eqref{eq:expansion}, i.e. the induced potentials on $C$
\begin{equation}
\pphi=\mSC(\mL(\wwbar), C)^\dagger\HH_C^\top\mB^\top\mWbar^{1/2}\vv \label{eq:pphidef}
\end{equation}
and 
\begin{equation}
\ppsi=\mSC(\mL(\ww), C)^\dagger\HH_C^\top\mB^\top\mW^{1/2}\vvhat. \label{eq:ppsidef}
\end{equation}

\cref{lemma:approxphi} below defines the approximation of $\pphi$ that our data structures maintain. To analyze the quality of the approximation, we will need a standard spectral approximation inequality, proven for completeness.
\begin{lemma}[Spectral approximation of differences]
\label{lemma:approxlemma}
For PSD matrices $\mX \approx_{\eps} \mY$, we have that
\[ (\mX - \mY)\mX^\dagger(\mX - \mY) \pe \eps^2\mX. \]
\end{lemma}
\begin{proof}
The desired inequality follows from $\left\|\mI - \mX^{\dagger/2}\mY\mX^{\dagger/2}\right\|_2 \le \eps$ and multiplying the LHS and RHS by $\mX^{1/2}$ on the left and right. Now, this follows because $\mX \approx_{\eps} \mY$ implies that $(1-\eps)\mI \pe \mX^{\dagger/2}\mY\mX^{\dagger/2} \pe (1+\eps)\mI$ as desired.
\end{proof}

\begin{lemma}[Approximate potential]
\label{lemma:approxphi}
Let $G$ be a graph with weights $\wwbar \in \R^E$ which differ from weights $\ww \in \R^E$ except on an edge subset $Z \subseteq E(G)$. Let $C \subseteq V(G)$ contain all endpoints of edges in $Z$. Let $\ww_{\overline{Z}} \in \R^E$ be defined as $(\ww_{\overline{Z}})_e = 0$ for $e \in Z$ and $(\ww_{\overline{Z}})_e = \ww_e$ otherwise.
Let $\wt{\mSC} \approx_\eps \mSC(\mL(\wwbar), C)$ and let $\wt{\HH} = \mL(\ww)^\dagger\begin{bmatrix} 0 \\ \wt{\mSC}_{\HH} \end{bmatrix}$ for some $\wt{\mSC}_{\HH}$ satisfying \begin{align} \mSC(\mL(\ww), C) - \eps \mSC(\mL(\ww_{\overline{Z}}), C) \pe \wt{\mSC}_{\HH} \pe \mSC(\mL(\ww), C) + \eps \mSC(\mL(\ww_{\overline{Z}}), C). \label{eq:happrox} \end{align}
Then the vectors $\pphi = \mSC(\mL(\wwbar), C)^\dagger\HH_C^\top\mB^\top\mWbar^{1/2}\vv$ (\eqref{eq:pphidef}) and 
$\pphitil = \wt{\mSC}^\dagger \mB^\top(\mWbar^{1/2} - \mW^{1/2})\vv + \wt{\mSC}^\dagger \wt{\HH}^\top \mB^\top \mW^{1/2}\vv$ in $\R^C$ 
satisfy $\left\|\pphi - \pphitil \right\|_{\mSC(\mL(\wwbar), C)} \le 3\eps\|\vv\|_2$.
\end{lemma}
\begin{proof}
We first calculate that
\begin{align*}
\pphi &= \left(\mSC(\mL(\wwbar), C)^\dagger - \wt{\mSC}^\dagger\right)\HH_C^\top\mB^\top\mWbar^{1/2}\vv + \wt{\mSC}^\dagger\HH_C^\top\mB^\top\mWbar^{1/2}\vv \\
&= \left(\mSC(\mL(\wwbar), C)^\dagger - \wt{\mSC}^\dagger\right)\HH_C^\top\mB^\top\mWbar^{1/2}\vv + \wt{\mSC}^\dagger\HH_C^\top\mB^\top\left(\mWbar^{1/2}-\mW^{1/2}\right)\vv + \wt{\mSC}^\dagger\HH_C^\top\mB^\top\mW^{1/2}\vv \\
&\overset{(i)}{=} \left(\mSC(\mL(\wwbar), C)^\dagger - \wt{\mSC}^\dagger\right)\HH_C^\top\mB^\top\mWbar^{1/2}\vv + \wt{\mSC}^\dagger\mB_{:,C}^\top\left(\mWbar^{1/2}-\mW^{1/2}\right)\vv + \wt{\mSC}^\dagger\HH_C^\top\mB^\top\mW^{1/2}\vv \\
\end{align*}
$(i)$ is because $\HH_C^\top\mB^\top\left(\mWbar^{1/2}-\mW^{1/2}\right) = \mB^\top_{:,C}\left(\mWbar^{1/2}-\mW^{1/2}\right)$ as $\wwbar = \ww$ except on $C$.
We extract $\pphitil$ by swapping the $\wt{\HH}$ by $\wt{\HH}_{C}$ in the last term:
\begin{align*}
\pphi &=\left(\mSC(\mL(\wwbar), C)^\dagger - \wt{\mSC}^\dagger\right)\HH_C^\top\mB^\top\mWbar^{1/2}\vv + \wt{\mSC}^\dagger\mB_{:,C}^\top\left(\mWbar^{1/2}-\mW^{1/2}\right)\vv + \wt{\mSC}^\dagger\HH_C^\top\mB^\top\mW^{1/2}\vv \\
&= \pphitil + \left(\mSC(\mL(\wwbar), C)^\dagger - \wt{\mSC}^\dagger\right)\HH_C^\top\mB^\top\mWbar^{1/2}\vv + \wt{\mSC}^\dagger\left(\HH_C^\top - \wt{\HH}^\top\right)\mB^\top\mW^{1/2}\vv.
\end{align*}
Hence
\begin{align} \pphi - \pphitil = \left(\mSC(\mL(\wwbar), C)^\dagger - \wt{\mSC}^\dagger\right)\HH_C^\top\mB^\top\mWbar^{1/2}\vv + \wt{\mSC}^\dagger \left(\HH_C^\top - \wt{\HH}^\top\right)\mB^\top\mW^{1/2}\vv. \label{eq:diffphi} \end{align}
We bound both terms separately. For the first term,
\begin{align*}
&\left\|\left(\mSC(\mL(\wwbar), C)^\dagger - \wt{\mSC}^\dagger\right)\HH_C^\top\mB^\top\mWbar^{1/2}\vv\right\|_{\mSC(\mL(\wwbar), C)} \overset{(i)}{\le} \eps\left\|\HH_C^\top\mB^\top\mWbar^{1/2}\vv\right\|_{\mSC(\mL(\wwbar), C)^\dagger} \\
\overset{(ii)}{\le}~& \eps\left\|\mB^\top\mWbar^{1/2}\vv\right\|_{\mL(\wwbar)^\dagger} \le \eps\|\vv\|_2.
\end{align*}
where $(i)$ follows from $\wt{\mSC} \approx_\eps \mSC(\mL(\wwbar), C)$ and \cref{lemma:approxlemma} for $\mX = \mSC(\mL(\wwbar), C)^\dagger$ and $\mY = \wt{\mSC}^\dagger$,
and $(ii)$ follows from \cref{lemma:HSH} and the fact that $\mL_{F, F}^{-1}$ is positive definite. For the second term,
\begin{align*}
&\left\|\wt{\mSC}^\dagger\left(\HH_C^\top - \wt{\HH}^\top\right)\mB^\top\mW^{1/2}\vv\right\|_{\mSC(\mL(\wwbar), C)} \le 2\left\|\left(\HH_C^\top - \wt{\HH}^\top\right)\mB^\top\mW^{1/2}\vv\right\|_{\mSC(\mL(\wwbar), C)^\dagger} \\
\le ~& 2\left\|\left(\HH_C^\top - \wt{\HH}^\top\right)\mB^\top\mW^{1/2}\vv\right\|_{\mSC(\mL(\ww_{\overline{Z}}), C)^\dagger} \overset{(i)}{\le} 2\eps\left\|\left[\mL(\ww)^\dagger\mB^\top\mW^{1/2}\vv\right]_C\right\|_{\mSC(\mL(\ww_{\overline{Z}}), C)} \\
\overset{(ii)}{\le} ~& 2\eps\left\|\left[\mL(\ww)^\dagger\mB^\top\mW^{1/2}\vv\right]_C\right\|_{\mSC(\mL(\ww), C)} \overset{(iii)}{\le} 2\eps\|\vv\|_2,
\end{align*}
where $(i)$ follows from \cref{lemma:htosc} and \eqref{eq:happrox}, $(ii)$ follows from $\mSC(\mL(\ww), C)-\mSC(\mL(\ww_{\overline{Z}}), C) =\mL(G[C])$ being positive semidefinite, and $(iii)$ follows from the fact that the Schur complement is spectrally smaller than the Laplacian: $\mSC(\mL(\ww), C) \pe 
\mL(\ww)$.
\end{proof}

By \cref{lemma:approxphi} with $\wwbar=\ww$, we can approximate the other potential vector in the RHS of  \eqref{eq:expansion}.
\begin{corollary}
\label{cor:approxpsi}
Let $\ppsi = \mSC(\mL(\ww), C)^\dagger\HH_C^\top\mB\mW^{1/2}\vv$ (\eqref{eq:ppsidef}) and 
$\ppsitil =\wt{\mSC}(\mL(\ww), C)^\dagger \wt{\HH}^\top \mB \mW^{1/2}\vv$ in $\R^C$ where 
$\wt{\mSC}(\mL(\ww), C)$ satisfies $\wt{\mSC}(\mL(\ww), C) \approx_\eps \mSC(\mL(\ww), C)$.
We have $\left\|\ppsi - \ppsitil \right\|_{\mSC(\mL(\ww), C)} \le 3\eps\|\vv\|_2$.
\end{corollary}
\begin{proof}
Apply \cref{lemma:approxphi} with $\wwbar=\ww$, $\ppsi=\pphi$ and $\ppsitil=\pphitil$. The first term of $\ppsi$
\[\wt{\mSC}^\dagger \mB(\mWbar^{1/2} - \mW^{1/2})\vv
\] equals $0$ because $\mWbar=\mW$.
\end{proof}

We can use our approximate potential $\pphitil$, $\ppsitil$ in \cref{lemma:approxphi} and \cref{cor:approxpsi} to define a full approximate projection of $\mPi\vv$. Our starting point is that $\mPi\vv = \mWbar^{1/2}\mB\HH_C\pphi+\mW^{1/2}\mB\HH_C\ppsi$ for $\pphi$, $\ppsi$ as in \cref{lemma:approxphi} and \cref{cor:approxpsi}.
\begin{lemma}[Approximate projection]
\label{lemma:approxpi}
Let $\ww, \wwbar, Z, \ww_{\overline{Z}}, \wt{\HH}, \pphi, \pphitil$ be as in \cref{lemma:approxphi}, let $\ppsi, \ppsitil$ be as in \cref{cor:approxpsi}, and let
\begin{align*}
 \wt{\mPi}\vv =& \left(\mWbar^{1/2} - \mW^{1/2}\right)\mB\pphitil + \mW^{1/2}\mB\wt{\HH}\pphitil\\
 &+\mW^{1/2}\mB\wt{\HH}\ppsitil\\
 &+\mW^{1/2}\mB\mL(\ww)^\dag\mB^\top\mW^{1/2}\vv
\end{align*}
where $\pphitil$ is padded with zeroes for computing $\mB\pphitil$.
Then \[ \left\|\mPi\vv - \wt{\mPi}\vv\right\|_2 \le 2\eps\|\vv\|_2 + (1+\eps)\left\|\pphi - \pphitil \right\|_{\mSC(\mL(\wwbar), C)}+(1+\eps)\left\|\ppsi - \ppsitil \right\|_{\mSC(\mL(\ww), C)}. \]
\end{lemma}
\begin{proof}
We first prove the first two terms of $\wt{\mPi}\vv$
\[\wt{\mPi}_\phi(\vv)\defeq \left(\mWbar^{1/2} - \mW^{1/2}\right)\mB\pphitil + \mW^{1/2}\mB\wt{\HH}\pphitil
\] approximates $\mPi_\phi(\vv)\defeq \mWbar^{1/2}\mB\HH_C\pphi$. Specifically, 
\[
\left\|\mPi_\phi(\vv) - \wt{\mPi}_\phi(\vv)\right\|_2 \le \eps\|\vv\|_2 + (1+\eps)\left\|\pphi - \pphitil \right\|_{\mSC(\mL(\wwbar), C)}.
\]
We start by calculating
\begin{align*}
\mPi_\phi(\vv) &= \mWbar^{1/2}\mB\HH_C\pphi = \mWbar^{1/2}\mB\HH_C\left(\pphi-\pphitil\right) + \mWbar^{1/2}\mB\HH_C\pphitil \\
&= \mWbar^{1/2}\mB\HH_C\left(\pphi-\pphitil\right) + \left(\mWbar^{1/2}-\mW^{1/2}\right)\mB\HH_C\pphitil + \mW^{1/2}\mB\HH_C\pphitil \\
&\overset{(i)}{=} \mWbar^{1/2}\mB\HH_C\left(\pphi-\pphitil\right) + \left(\mWbar^{1/2}-\mW^{1/2}\right)\mB\pphitil + \mW^{1/2}\mB\HH_C\pphitil \\ 
&= \mWbar^{1/2}\mB\HH_C\left(\pphi-\pphitil\right) + \mW^{1/2}\mB(\HH_C - \wt{\HH})\pphitil + \wt{\mPi}_\phi(\vv),
\end{align*}
where $(i)$ follows from $\left(\mWbar^{1/2}-\mW^{1/2}\right)\mB\HH_C = \left(\mWbar^{1/2}-\mW^{1/2}\right)\mB$ as $\wwbar = \ww$ outside $C$. Hence
\begin{align}
\mPi_\phi(\vv) - \wt{\mPi}_\phi(\vv) = \mWbar^{1/2}\mB\HH_C\left(\pphi-\pphitil\right) + \mW^{1/2}\mB(\HH_C - \wt{\HH})\pphitil. \label{eq:diffpi}
\end{align}
We bound both terms of \eqref{eq:diffpi} separately. For the first term, note that
\begin{align*}
&\left\|\mWbar^{1/2}\mB\HH_C\left(\pphi-\pphitil\right)\right\|_2 \overset{(i)}{\le} \left\|\pphi-\pphitil\right\|_{\mSC(\mL(\wwbar), C)},
\end{align*}
where $(i)$ follows from properties of $\HH_C$. For the second term of \eqref{eq:diffpi},
\begin{align*}
&\left\|\mW^{1/2}\mB(\HH_C - \wt{\HH})\pphitil\right\|_2 \overset{(i)}{\le} \eps\left\|\pphitil\right\|_{\mSC(\mL(\ww_{\overline{Z}}), C)}
\le \eps\left\|\pphitil\right\|_{\mSC(\mL(\wwbar), C)} \overset{(ii)}{\le} \eps\|\vv\|_2 + \left\|\pphi-\pphitil\right\|_{\mSC(\mL(\wwbar), C)}
\end{align*}
where $(i)$ follows from \eqref{eq:happrox}, and $(ii)$ is because 
\[
\|\pphitil\|_{\mSC(\mL(\wwbar), C)}\le 
\|\pphitil-\pphi\|_{\mSC(\mL(\wwbar), C)}+\|\pphi\|_{\mSC(\mL(\wwbar), C)}
\le \left\|\pphi-\pphitil\right\|_{\mSC(\mL(\wwbar), C)} + \|\vv\|_{\mP_{\wwbar}}
\le \|\vv\|_2
\] because $\mP_{\mWbar}$ is an orthogonal projection matrix. Summing these errors completes the proof for 
\begin{equation}
\left\|\mPi_\phi(\vv) - \wt{\mPi}_\phi(\vv)\right\|_2 \le \eps\|\vv\|_2 + (1+\eps)\left\|\pphi - \pphitil \right\|_{\mSC(\mL(\wwbar), C)} .\label{eq:phiapprox}
\end{equation}
We then define 
\[\wt{\mPi}_\psi(\vv)\defeq \left(\mW^{1/2} - \mW^{1/2}\right)\mB\pphitil + \mW^{1/2}\mB\wt{\HH}\pphitil
\] which is the third term of $\mPi\vv$ and 
\[
\mPi_\psi(\vv)\defeq \mWbar^{1/2}\mB\HH_C\ppsi.
\] Then, by the proof above with $\wwbar$ replaced by $\ww$ (and $\mWbar$ replaced by $\mW$), we get 
\begin{equation}
\left\|\mPi_\psi(\vv) - \wt{\mPi}_\psi(\vv)\right\|_2 \le \eps\|\vv\|_2 + (1+\eps)\left\|\ppsi - \ppsitil \right\|_{\mSC(\mL(\ww), C)}.\label{eq:psiapprox}
\end{equation}
Recall that 
\begin{align*}
\mPi\vv=&\mWbar^{1/2}\mB\HH_C\mSC(\mL(\wwbar), C)^\dagger\HH_C^\top\mB^\top\mWbar^{1/2}\vv
-\mW^{1/2}\mB\HH_C\mSC(\mL(\ww), C)^\dagger\HH_C^\top\mB^\top\mW^{1/2}\vv\\
+&\mW^{1/2}\mB\mL(\ww)^\dag\mB^\top\mW^{1/2}\vv.
\end{align*} 
Its first two terms are approximated respectively by $\wt{\mPi}_\phi(\vv)$ and $\wt{\mPi}_\psi(\vv)$, the first two terms in $\wt{\mPi}\vv$. Its last term is exactly the last term of  $\wt{\mPi}\vv$. By triangle inequality and \cref{eq:phiapprox}, \cref{eq:psiapprox}, we have 
\[
\left\|\mPi\vv - \wt{\mPi}\vv\right\|_2 \le 2\eps\|\vv\|_2 + (1+\eps)\left\|\pphi - \pphitil \right\|_{\mSC(\mL(\wwbar), C)}+(1+\eps)\left\|\ppsi - \ppsitil \right\|_{\mSC(\mL(\ww), C)}.
\]
\end{proof}

In the following sections, we will use \textsc{Solve}$(\mL, \bb)$ to denote a high accuracy Laplacian solver that returns $\xx$ such that $\mL\xx=\bb$ and runs in nearly linear time. We will overload notation to extend any dimension of a matrix from a subset of $V$ to $V$, or from a subset of $E$ to $E$ by padding zeroes.

\subsection{Dynamic Potential Maintanence}
In this section, we show how to maintain the vector $\pphitil$ (\cref{lemma:approxphi}) that approximates the potential vector $\pphi$. This data structure can also be used for $\ppsi$ (\cref{cor:approxpsi}).
\label{subsec:potential}
\begin{lemma}[Dynamic Potential]
\label{thm:potential}
For a graph $G = (V, E)$ with dynamic edge conductances $\wwbar \in \R^{E(G)}_{\ge0}$ and a 
dynamic vector $\vvbar \in \R^{E(G)}$ for some constant $C$, there is a data structure (\cref{alg:potential}) that supports the following operations against an oblivious adversary for parameters $\beta < \eps^2 < 1$.
\begin{itemize}
\item $\textsc{Initialize}(G, \ww, \vv^\init, \beta, \eps)$. Initializes the data structure in time $\O(m\beta^{-2}\eps^{-2})$ with an empty set $Z \assign \emptyset$ of marked edges. Initialize $\wwbar$ as $\ww$ and $\vvbar$ as $\vv^\init$. 
\item $\textsc{UpdateV}(e, \vvbar^\new)$. Updates $\vvbar_e \assign \vvbar^\new$ in $\O(\beta^{-2}\eps^{-2})$ time.
\item $\textsc{UpdateW}(e, \wwbar^\new)$. Updates $\wwbar_e \assign \wwbar^\new$ in $\O(\beta^{-2}\eps^{-2})$ time.
\item $\textsc{QueryPotential}()$. For $C \subseteq V$ with $|C| = O(\beta m)$ and $Z\subseteq E(C)$, returns in $\O(\beta m\eps^{-2})$ time a vector $\pphitil$ satisfying $\left\|\pphi - \pphitil \right\|_{\mSC(\mL(\wwbar), C)} \le \eps\|\vvbar\|_2$ where $\pphi = \mSC(\mL(\wwbar), C)^\dagger\HH_C^\top\mB\mWbar^{1/2}\vvbar$.
\end{itemize}
Runtimes and output correctness hold w.h.p. if there are at most $O(\beta m)$ calls to $\textsc{UpdateV} and \textsc{UpdateW}$ in total.
\end{lemma}
\begin{proof}
The pseudocode for the proof of \cref{thm:potential} is in \cref{alg:potential}. At a high-level, it simply maintains $\pphitil$ as in \cref{lemma:approxphi}. The one difference is that it handles changes to $\vv^\init$ directly because both endpoints of all edges changed in $\vv^\init$ are in the marked set $Z$. We start by analyzing the correctness of the algorithm, then move the runtime.

\paragraph{Correctness.} We only need to analyze the $\textsc{QueryPotential}()$ operation. In this proof, we show the weaker bound $\left\|\pphitil -\pphi\right\|_{\mSC(\mL(\wwbar), C)} \le 4\eps\left(\|\vvbar\|_2 + \|\vv^\init\|_2 \right)$. 
However, this suffices because we can build $\O(1)$ copies of the data structure. For $-\O(1) \le j \le \O(1)$, the $j$-th instance initializes and answers queries only when $\|\vvbar\|_2 \in (2^j, 2^{j+1}]$. Updates are passed to all instances. When the number of updates exceeds $O(\beta m)$ for an instance but it cannot be initialized because $\|\vvbar\|_2$ does not fall in its range, it simply ignore following updates until it can be initialized. This only increases the runtime by $\O(1)$ factors.

Let $\wt{\mSC}_{\HH}$ be the value of $\dsc.\textsc{InitialSC}()$ returned in line \ref{line:getinitialscphi} of \cref{alg:potential}. It satisfies condition \eqref{eq:happrox} by the guarantees of \cref{thm:dynasc}. Also, by inspection of the procedure $\textsc{QueryPotential}()$ in \cref{alg:potential}, the returned vector $\pphitil$ is defined as
\begin{align*} \pphitil &= \wt{\mSC}^\dagger \begin{bmatrix} \wt{\mSC}_{\HH} & 0 \end{bmatrix} \dd^\init + \wt{\mSC}^\dagger\mB\mW^{1/2}(\vvbar_Z-\vv^\init_Z) + \wt{\mSC}^\dagger\mB(\mWbar^{1/2}-\mW^{1/2}) \vvbar_Z \\
&= \wt{\mSC}^\dagger \begin{bmatrix} \wt{\mSC}_{\HH} & 0 \end{bmatrix} \mL(\ww)^\dagger\mB\mW^{1/2}\vv^\init + \wt{\mSC}^\dagger\mB\mW^{1/2}(\vvbar_Z-\vv^\init_Z) + \wt{\mSC}^\dagger\mB(\mWbar^{1/2}-\mW^{1/2}) \vvbar_Z \\
&= \wt{\mSC}^\dagger\wt{\HH}_C\mB\mW^{1/2}\vv^\init + \wt{\mSC}^\dagger \mB(\mWbar^{1/2}-\mW^{1/2})\vv^\init + \wt{\mSC}^\dagger\mB\mWbar^{1/2}(\vvbar_Z-\vv^\init_Z).
\end{align*}
Additionally, because $\vv^\init - \vvbar$ is supported on $Z$, the true $\pphi$ can be written as
\begin{align*}
\pphi &= \mSC(\mL(\wwbar), C)^\dagger\HH_C^\top\mB\mWbar^{1/2}\vvbar \\
&= \mSC(\mL(\wwbar), C)^\dagger\HH_C^\top\mB\mWbar^{1/2}\vv^\init + \mSC(\mL(\wwbar), C)^\dagger\mB\mWbar^{1/2}(\vvbar - \vv^\init).
\end{align*}
Hence we get that
\begin{align*}
&\left\|\pphitil -\pphi\right\|_{\mSC(\mL(\wwbar), C)} \\ \le~&\left\|\wt{\mSC}^\dagger\wt{\HH}_C\mB\mW^{1/2}\vv^\init + \wt{\mSC}^\dagger \mB(\mWbar^{1/2}-\mW^{1/2})\vv^\init - \mSC(\mL(\wwbar), C)^\dagger\HH_C^\top\mB\mWbar^{1/2}\vv^\init \right\|_{\mSC(\mL(\wwbar), C)} \\ +~&\left\|\wt{\mSC}^\dagger\mB\mWbar^{1/2}(\vvbar_Z-\vv^\init_Z) - \mSC(\mL(\wwbar), C)^\dagger\mB\mWbar^{1/2}(\vvbar - \vv^\init)\right\| \\
\overset{(i)}{\le}~& 3\eps\|\vv^\init\|_2 + \eps\|\vvbar - \vv^\init\|_2 \le 4\eps\left(\|\vvbar\|_2 + \|\vv^\init\|_2 \right),
\end{align*}
where $(i)$ follows from \cref{lemma:approxphi} for the first term, and $\wt{\mSC} \approx_{\eps} \mSC(\mL(\wwbar), C)$ from the guarantee of \cref{thm:dynasc} for the second term. This suffices because \cref{alg:potential} set $\eps \assign \eps/10$ in line \ref{line:eps4}.

\paragraph{Runtime.} The runtimes of \textsc{UpdateV}, \textsc{UpdateW} are trivially the same as the runtime of \textsc{Mark}. The runtimes of \textsc{Mark} and \textsc{Initialize} follows from the \textsc{AddTerminal} and \textsc{Initialize} operations respectively of \cref{thm:dynasc}. The runtime of \textsc{QueryPotential} is $\O(\beta m\eps^{-2})$ by the runtime guarantees of \textsc{SC} and \textsc{InitialSC} of \cref{thm:dynasc}, and the fact that $\wt{\mSC}$ and $\wt{\mSC}_{\HH}$ all have $\O(\beta m\eps^{-2})$ edges, and hence solving or multiplying by them costs $\O(\beta m\eps^{-2})$ time.
\end{proof}

\begin{algorithm2e}[!ht]
\caption{Dynamic Potential \label{alg:potential}}
\SetKwProg{Proc}{procedure}{}{}
\SetKwIF{Prob}{ProbElseProb}{ElseProb}{with probability}{do}{else with probability}{else}{end}
\tcp{This implementation assumes that $\|\vvbar\|_2 \approx \|\vv^\init\|_2$ always. However, this can be achieved by duplicating the data structure $\O(1)$ times, one handling each range $\|\vvbar\|_2 \in [2^j, 2^{j+1}]$ for $-\O(1) \le j \le \O(1)$.}
\Proc{\textsc{Initialize}$(G, \ww, \vv^\init, \beta, \eps)$}{
	$\eps\leftarrow \eps/10$. \label{line:eps4} \\
	Let $\dsc$ be a instance of the dynamic Schur complement data structure of \cref{thm:dynasc}.\\
	$\dsc$.\textsc{Initialize}$(G, \ww, \eps, \beta)$.\\ 
	$\vvbar\leftarrow \vv^\init.$\\ \tcp{$\vv^\init$ is the initial vector and $\vvbar$ to denote the current vector $\vv$ throughout the algorithm.}
	$\wwbar\leftarrow \ww.$\\ \tcp{We use $\ww$ to denote the initial vector and $\wwbar$ to denote the current vector $\ww$ throughout the algorithm.}
	$\dd^\init \leftarrow \textsc{Solve}(\mL(\ww), \mB\mW^{1/2}\vv^\init).$\\
	$Z \assign \emptyset$. \tcp{Marked edges.}
}
\Proc{\textsc{Mark}$(e)$}{
	$\dsc.\textsc{AddTerminal}(u).$\\
	$\dsc.\textsc{AddTerminal}(v).$\\
	$Z \assign Z \cup \{e\}$. \\
	$\dsc.\textsc{Update}(e, \ww_e)$. \tcp{Make sure the $\dsc$ puts edge $e$ in $Z$.}
}
\Proc{\textsc{UpdateV}$(e, \vv^\new)$}{
	\textsc{Mark}$(e)$.\\
	$\vvbar_e \leftarrow \vv^\new$.
}
\Proc{\textsc{UpdateW}$(e, \ww^\new)$}{
	\textsc{Mark}$(e)$.\\
	$\dsc.\textsc{Update}(e, \ww^\new)$.\\
	$\wwbar_e\leftarrow \ww^\new$.
}
\Proc{\textsc{QueryPotential}$()$} {
	$\wt{\mSC} \assign \dsc.\textsc{SC}()$. \label{line:approxscphi} \\
	$\pphitil\leftarrow \textsc{Solve}\left(\wt{\mSC}, \begin{bmatrix} \dsc.\textsc{InitialSC}() & 0 \end{bmatrix} \dd^\init\right).$ \label{line:getinitialscphi}\\
	$\pphitil\leftarrow \pphitil + \textsc{Solve}(\wt{\mSC}, \mB\mW^{1/2}(\vvbar_Z-\vv^\init_Z)).$\\
	$\pphitil\leftarrow \pphitil + \textsc{Solve}(\wt{\mSC}, \mB(\mWbar^{1/2}-\mW^{1/2}) \vvbar_Z).$\\
	\Return $\pphitil.$
}
\end{algorithm2e}

\subsection{Dynamic Evaluator}
\label{subsec:evaluator}
\begin{theorem}[Dynamic Evaluator]
\label{thm:evaluator}
For a graph $G = (V, E)$ with dynamic edge conductances $\wwbar \in \R^{E(G)}_{\ge0}$ and a dynamic vector $\vvbar \in \R^{E(G)}$, there is a data structure $\Evaluator$ that supports the following operations against an oblivious adversary for parameters $\beta < \eps^2 < 1$. 
\begin{itemize}
\item $\textsc{Initialize}(G, \ww, \vv^\init, \beta, \eps)$. Initializes the data structure in time $\O(m\beta^{-2}\eps^{-2})$ with an empty set $Z \assign \emptyset$ of marked edges. Initializes $\wwbar$ as $\ww$, $\vvbar$ as $\vv^\init$.
\item $\textsc{UpdateV}(e, \vvbar^\new)$. Updates $\vvbar_e \assign \vvbar^\new$ in $\O(\beta^{-2}\eps^{-2})$ time.
\item $\textsc{UpdateW}(e, \wwbar^\new)$. Updates $\wwbar_e \assign \wwbar^\new$ in $\O(\beta^{-2}\eps^{-2})$ time.
\item $\textsc{Query}()$. Returns a vector $\uu \in \R^Z$ satisfying $\|\uu - \left[\mP_{\wwbar}\vvbar\right]_Z \|_2 \le \eps\|\vvbar\|_2+\eps\|\vv\|_2$ in time $\O(\beta m\eps^{-2})$.
\end{itemize}
Runtimes and output correctness hold w.h.p.\ if there are at most $O(\beta m)$ calls to $\textsc{UpdateV}, \textsc{UpdateW}$ in total.
\end{theorem}

\begin{proof}
We decompose $\mP_{\wwbar}\vvbar$ by \cref{eq:expansion}. The $\vv$ in \cref{eq:expansion} is the current vector $\vvbar$ and the $\vvhat$ in \cref{eq:expansion} is the initial vector $\vv$ here. We create two instances of \cref{alg:potential} $\dphi$ and $\dpsi$ maintaining 
\[\pphitil=\wt{\mSC}^\dagger \mB(\mWbar^{1/2} - \mW^{1/2})\vvbar + \wt{\mSC}^\dagger \wt{\HH}^\top \mB \mW^{1/2}\vvbar
\] (\cref{lemma:approxphi}) and 
\[\ppsitil=\wt{\mSC}(\mL(\ww), C)^\dagger \wt{\HH}^\top \mB \mW^{1/2}\vv
\] (\cref{cor:approxpsi}) 
respsectively. 
$\textsc{Initialize}, \textsc{Mark}$ are forwarded to both $\dphi$ and $\dpsi$. The operations $\textsc{UpdateV}, \textsc{UpdateW}$ are forwared only to $\dphi$ as $\dpsi$ maintains $\ppsitil$ where $\ww$ and $\vv$ do not change. We also compute the exact value of the last term $\mW^{1/2}\mB\mL(\ww)^\dag\mB^\top\mW^{1/2}\vv$ of $\mP_{\wwbar}\vvbar$ by 
\[
\xx=\mW^{1/2}\mB\textsc{Solve}(\mL(\ww), \mB^\top\mW^{1/2}\vv).
\]
For $\textsc{Query}()$, let $\pphitil = \dphi.\textsc{QueryPotential}()$, $\ppsitil = \dpsi.\textsc{QueryPotential}()$. The \textsc{Evaluator} returns 
\[\uu = \left[\mWbar^{1/2}\mB\pphitil\right]_Z+\left[\mW^{1/2}\mB\ppsitil\right]_Z+\xx_Z.\] Clearly all runtimes transfer exactly from \cref{thm:potential}. It suffices to show the correctness of $\textsc{Query}()$.

For the true potentials $\pphi = \mSC(\mL(\wwbar), C)^\dagger\HH_C^\top\mB\mWbar^{1/2}\vvbar$ and $\ppsi = \mSC(\mL(\ww), C)^\dagger\HH_C^\top\mB\mW^{1/2}\vv$ we have $\left[\mP_{\wwbar}\vvbar\right]_Z = \left[\mWbar^{1/2}\mB\pphi+\mW^{1/2}\mB\ppsi\right]_Z+\xx_Z$. 
Thus,
\begin{align*} \left\|\uu - \left[\mP_{\wwbar}\vvbar\right]_Z\right\|_2 =& \left\|\mWbar_Z^{1/2}\mB_Z\left(\pphitil - \pphi\right)+\mW^{1/2}\mB_Z\left(\ppsitil-\ppsi\right)\right\|_2 \\
\le &\left\|\mWbar_Z^{1/2}\mB_Z\left(\pphitil - \pphi\right)\right\|_2+\left\|\mW^{1/2}\mB_Z\left(\ppsitil-\ppsi\right)\right\|_2 \\
&= \left\|\pphitil - \pphi\right\|_{\mB_Z^\top\mWbar_Z\mB_Z}+\left\|\ppsitil - \ppsi\right\|_{\mB_Z^\top\mW_Z\mB_Z}\\
& \overset{(i)}{\le} \left\|\pphitil - \pphi\right\|_{\mSC(\mL(\wwbar), C)}+\left\|\ppsitil - \ppsi\right\|_{\mSC(\mL(\ww), C)}\\
& \overset{(ii)}{\le} \eps\|\vvbar\|_2+\eps\|\vv\|_2
\end{align*}
where $(i)$ follows from the fact that $\mL(\wwbar_Z) \pe \mSC(\mL(\wwbar), C)$ (and $\mL(\ww_Z) \pe \mSC(\mL(\ww), C)$) as $Z$ is completely inside $C$, and $(ii)$ follows from the guarantee of $\textsc{QueryPotential}()$ of \cref{thm:potential}. This completes the proof.
\end{proof}
\subsection{Dynamic Locator}
\label{subsec:locator}
\begin{theorem}[Dynamic Locator]
\label{thm:locator}
For a graph $G = (V, E)$ with dynamic edge conductances $\wwbar \in \R^{E(G)}_{\ge0}$ and a dynamic vector $\vvbar \in \R^{E(G)}$, there is a data structure $\Locator$ (given in \cref{alg:locator}) that supports the following operations against an oblivious adversary for parameters $\beta < \eps^2 < 1$.
\begin{itemize}
\item $\textsc{Initialize}(G, \ww, \vv^\init, \beta, \eps)$. Initializes the data structure in time $\O(m\beta^{-2}\eps^{-2})$ and sets $\wwbar \assign \ww$ and $\vvbar \assign \vv^\init$.
\item $\textsc{UpdateV}(e, \vvbar^\new)$. Updates $\vvbar_e \assign \vvbar^\new$ in $\O(\beta^{-2}\eps^{-2})$ time.
\item $\textsc{UpdateW}(e, \wwbar^\new)$. Updates $\wwbar_e \assign \wwbar^\new$ in $\O(\beta^{-2}\eps^{-2})$ time.
\item $\textsc{Locate}().$ Returns in time $\O(\beta m\eps^{-2})$ a set $S \subseteq E(G)$ with $|S| \le O(\eps^{-2})$ containing all edges $e$ with $\left|[\mP_{\ww}\vvbar]_e\right| \ge \eps\|\vvbar\|_2$ whp.
\end{itemize}
Runtimes and output correctness hold w.h.p. if there are at most $O(\beta m)$ calls to $\textsc{UpdateV}, \textsc{UpdateW}$ in total.
\end{theorem}
The following lemma is implicit in \cite{KNPW11} and allows us to recover the large entries of $\xx$ by a low-dimensional projection of it.

\begin{lemma}[$\ell_{2}$-heavy hitter, \cite{KNPW11}]
\label{lem:heavyhitter} There exists a function $\textsc{Sketch}(\epsilon,n)$
that given $\epsilon>0$ explicitly returns a random matrix $\mQ \in\R^{N\times m}$
with \textup{$N=O(\epsilon^{-2}\log^3 m)$} and column sparsity
$c=O(\log^3 m)$ in $\O(N + m)$ time, and uses $\O(N + m)$
spaces to store the matrix $\mQ$. There further exists a function
$\textsc{Recover}(\mQ \xx)$ that in time $O(\epsilon^{-2}\log^3 m)$
reports a list $S\subset[m]$ of size $O(\epsilon^{-2})$. For any fixed $\xx$, the list includes all $i$ with $|\xx_i| \ge \eps\|\xx\|_2$ with high probability over the randomness of $Q$.
\end{lemma}

\begin{algorithm2e}[!ht]
\caption{Dynamic Locator \label{alg:locator}}
\SetKwProg{Proc}{procedure}{}{}
\SetKwIF{Prob}{ProbElseProb}{ElseProb}{with probability}{do}{else with probability}{else}{end}
\Proc{\textsc{Initialize}$(G, \ww, \vv^\init, \beta, \eps)$}{
	$\eps\leftarrow \eps/10$. \label{line:scaleeps} \\
	Let $\dphi$ be an instance of the dynamic potential data structure of \cref{thm:potential}. \\
	Let $\dpsi$ be an instance of the dynamic potential data structure of \cref{thm:potential}. \\
	Let $\dsc$ be an instance of the dynamic Schur complement data structure of \cref{thm:dynasc}. \\
	$\dphi$.\textsc{Initialize}$(G, \ww, \vv^\init, \eps, \beta)$. \\
	$\dpsi$.\textsc{Initialize}$(G, \ww, \vv^\init, \eps, \beta)$. \\
	$\dsc$.\textsc{Initialize}$(G, \ww, \eps, \beta)$. \\
	Initialize an $N=O(\eps^{-2}\log ^3 m)$ by $m$ matrix $\mQ$ with rows $\bq^{(1)}, \bq^{(2)}, \dots, \bq^{(N)} \in \mathbb{R}^m$ by \cref{lem:heavyhitter}. \label{line:initializeQ} \\
	\For{$i \in [N]$}{
		$\bgamma^{(i)} \assign \textsc{Solve}(\mL(\ww), \mB^\top\mW^{1/2}\bq^{(i)})$. \label{line:initializePsi} \tcp{$\bgamma^{(i)}$ are rows of $\mGamma \defeq \mQ\mW^{1/2}\mB\mL(\ww)^\dagger$.}
	}
	$\wwbar \assign \ww, \vvbar \assign \vv$.
	$\yy\assign \mB^\top \mW \vv$.
}
\Proc{\textsc{UpdateV}$(e, \vv^\new)$}{
	$\dphi.\textsc{UpdateV}(e, \vv^\new).$ \\
	$\dpsi.\textsc{UpdateV}(e, \vv^\new).$ \\
	$\dsc.\textsc{UpdateV}(e, \vv^\new).$ \\
	$\vvbar_e \assign \vv^\new$.
	$\yy\assign \yy+\mB^\top \mW (\vv^\new-\vv_e)$.
}
\Proc{\textsc{UpdateW}$(e, \ww^\new)$}{
	$\dphi.\textsc{UpdateW}(e, \vv^\new).$ \\
	\tcp{$\dpsi$ does not update $\bar{\ww}_e$ to $\ww^\new$.}
	$\dsc.\textsc{UpdateW}(e, \ww^\new).$ \\
	$\wwbar_e\leftarrow \ww^\new$.\\
}
\Proc{\textsc{Locate}$()$}{
	$\pphitil \assign \dphi.\textsc{QueryPotential}()$. \\
	\tcp{$\pphitil$ is padded with zeroes for computing $\mB\pphitil$}
	$\ppsitil \assign \dpsi.\textsc{QueryPotential}()$. \\
	$\pp \assign \mQ\left(\mWbar^{1/2} - \mW^{1/2}\right)\mB\pphitil + \mGamma \begin{bmatrix} 0 \\ \dsc.\textsc{InitialSC}() \end{bmatrix}\pphitil + \mGamma \begin{bmatrix} 0 \\ \dsc.\textsc{InitialSC}() \end{bmatrix}\ppsitil+\mGamma\yy.$ \label{line:computepp} \\
	Return the set $S$ returned by calling \textsc{Recover}$(\pp)$ of \cref{lem:heavyhitter}. \
}
\end{algorithm2e}

\begin{proof}[Proof of \cref{thm:locator}]
At a high level, \cref{alg:locator} simply maintains the formula given by \cref{lemma:approxpi} for $\pphitil$ and $\ppsitil$ given by the output of the dynamic potential maintenance data structure in \cref{thm:potential}. Let us first show correctness and then analyze runtime.
We only have to check correctness of $\textsc{Locate}()$. We follow the decomposition \cref{eq:expansion} with both $\vv$ and $\vvhat$ being the current vector $\vvbar$ here. The $\pphitil$ and $\ppsitil$ maintained by \cref{alg:locator} satisfy

\[\pphitil=\wt{\mSC}^\dagger \mB(\mWbar^{1/2} - \mW^{1/2})\vv + \wt{\mSC}^\dagger \wt{\HH}^\top \mB \mW^{1/2}\vvbar
\] and 
\[\ppsitil=\wt{\mSC}^\dagger \wt{\HH}^\top \mB \mW^{1/2}\vvbar. 
\] (Note that $\ppsitil$ is defined differently from \cref{thm:evaluator}.)

Thus, for $\wt{\mSC}_{\HH} = \dsc.\textsc{InitialSC}()$, $\pp$ as defined in $\textsc{Locate}()$ of \cref{alg:locator} satisfies
\begin{align*}
\pp &= \mQ\left(\mWbar^{1/2} - \mW^{1/2}\right)\mB\pphitil + \mGamma \begin{bmatrix} 0 \\ \dsc.\textsc{InitialSC}() \end{bmatrix}\pphitil +\mGamma \begin{bmatrix} 0 \\ \dsc.\textsc{InitialSC}() \end{bmatrix}\ppsitil+\mGamma \yy\\
&= \mQ\left(\mWbar^{1/2} - \mW^{1/2}\right)\mB\pphitil + \begin{bmatrix} 0 \\ \mQ\mW^{1/2}\mB\mL(\ww)^\dagger \wt{\mSC}_{\HH} \end{bmatrix}\pphitil +\begin{bmatrix} 0 \\ \mQ\mW^{1/2}\mB\mL(\ww)^\dagger \wt{\mSC}_{\HH} \end{bmatrix}\ppsitil \\ &~+\mQ\mW^{1/2}\mB\mL(\ww)^\dagger \mB^\top \mW^{1/2}\vvbar\\
&= \mQ\wt{\mPi}\vvbar
\end{align*}
for $\wt{\mPi}\vvbar$ as defined in \cref{lemma:approxpi}. Note that $\pphitil$ is padded with zeroes for computing $\mB\pphitil$ as in \cref{alg:locator}. Because
\begin{align} \left\|\pphi - \pphitil \right\|_{\mSC(\mL(\wwbar), C)} \le \eps\|\vvbar\|_2 \label{eq:philocate} \end{align}
and
\begin{align} \left\|\ppsi - \ppsitil \right\|_{\mSC(\mL(\ww), C)} \le \eps\|\vvbar\|_2 \label{eq:psilocate}, \end{align}
 by the guarantee of \cref{thm:potential} we have that
\begin{align*}
\left\|\wt{\mPi}\vvbar\right\|_2 &\le \left\|\mPi\vvbar\right\|_2 + \left\|\wt{\mPi}\vvbar - \mPi\vvbar\right\|_2 \\
&\le \|\vvbar\|_2 + 2\eps\|\vvbar\|_2 + (1+\eps)\left\|\pphi - \pphitil \right\|_{\mSC(\mL(\wwbar), C)} + (1+\eps)\left\|\ppsi - \ppsitil \right\|_{\mSC(\mL(\ww), C)}\\
&\le 2\|\vvbar\|_2.
\end{align*}

By \cref{lem:heavyhitter}, the set $S \assign \textsc{Recover}(\pp)$ contains all $e$ such that $\left|\left[\wt{\mPi}\vvbar\right]_e\right|$ is at least \[ \eps\left\|\wt{\mPi}\vvbar\right\|_2 \le 2\eps\|\vvbar\|_2. \]
Finally, if $e$ satisfies $\left|\left[\mPi\vvbar\right]_e\right| \ge 10\eps\|\vvbar\|_2$ then
\begin{align*} 
	\left|\left[\wt{\mPi}\vvbar\right]_e\right| &\ge \left|\left[\mPi\vvbar\right]_e\right| - \left\|\wt{\mPi}\vvbar - \mPi\vvbar\right\|_2 \\
&\ge 10\eps\|\vvbar\|_2 - 2\eps\|\vvbar\|_2 - (1+\eps)\left\|\pphi - \pphitil \right\|_{\mSC(\mL(\wwbar), C)} - (1+\eps)\left\|\ppsi - \ppsitil \right\|_{\mSC(\mL(\ww), C)} \ge 2\eps\|\vvbar\|_2
\end{align*}
where the final step follows from \cref{lemma:approxpi} with \eqref{eq:philocate}. Thus $e \in S$ as desired.

Now we bound the runtimes. The runtimes of \textsc{UpdateV} and \textsc{UpdateW} follow directly from \cref{thm:dynasc,thm:potential}. The cost of \textsc{Initialize} is the cost of \textsc{Initialize} in \cref{thm:dynasc,thm:potential} plus the cost of computing $\mGamma$. This involves solving $N$ Laplacian systems, which costs $\O(Nm) = \O(m\eps^{-2})$ time. This is dominated by $\O(m\beta^{-2}\eps^{-2})$. Finally, the cost of $\textsc{Locate}()$ is $\O(\beta m\eps^{-2})$ for computing $\pphitil, \ppsitil$ by \cref{thm:potential}, and the cost of computing $\pp$ in line \ref{line:computepp} of \cref{alg:locator}. The first term in line \ref{line:computepp} can be computed in time $O(N\beta m) = O(\beta m\eps^{-2})$ as $\mWbar^{1/2} - \mW^{1/2}$ is supported on $O(\beta m)$ entries and $\mQ$ has $N$ rows. The second and third terms in line \ref{line:computepp} can be computed by first multiplying $\dsc.\textsc{InitialSC}()$ times $\pphitil$ (or $\ppsitil$) in time $\O(\beta m\eps^{-2})$, as $\dsc.\textsc{InitialSC}()$ has $\O(\beta m\eps^{-2})$ edges, and then multiplying by $\mGamma$ which is a $N$-by-$O(\beta m)$ size matrix in time $O(N\beta m) = \O(\beta m\eps^{-2})$ time. Thus the total runtime of $\textsc{Locate}()$ is $\O(\beta m\eps^{-2})$ as desired.
\end{proof}
\section{Reducing Adaptive to Oblivious Adversaries}
\label{sec:adaptive}
In this section we show a blackbox reduction that is able to transform any dynamic algorithm that maintains some sequence of vectors $(\vv^t)_{t\ge1}$ against oblivious adversaries to one that can maintain the vectors against adaptive adversaries.
We formalize the requirements of the dynamic algorithm via \Cref{def:differential:oracle}.
Roughly, \Cref{def:differential:oracle} states that the dynamic algorithm must support two operations: (i) find the entries of the current vector $\vv^t$ with large absolute value, and (ii) query some set of the entries approximately.

\begin{definition}\label{def:differential:oracle} 
We call a dynamic algorithm an \emph{$\epsilon$-approximate $(L,S)$-locator} for an online\footnote{The sequence is may depend on outputs of the data structures.} sequence of vectors $(\vv^t)_{t\ge1}$, if in each iteration $t\ge1$ the dynamic algorithm returns a set $I \subset [n]$ of size at most $S$ containing all $i$ with $|\vv^t_i| > \epsilon$ in $L$ time.

We call a dynamic algorithm an \emph{$\epsilon$-approximate $C$-evaluator}, if it supports a query operation that, given some $I\subset[n]$, returns all $\ovv_i$ for $i \in I$ in $C(|I|)$ time for some $\ovv \in \R^n$ with $\| \ovv - \vv^t\|_2 \le \epsilon$.
\end{definition}

We show that, given a locator (a dynamic algorithm that can tell us the large entries), and locators (dynamic algorithms that tells us the entries of the vectors) with different accuracies, we can combine these dynamic algorithms to work against an adaptive adversary.
The more accurate locators will be used less frequently, resulting in an expected time complexity faster than the most accurate locator.

\begin{theorem}\label{thm:dp}
Assume we have $\epsilon$-accurate $(L,S)$-locator and $(\epsilon/2^i)$-accurate $C_i$-evaluators for $i=0,...,K$ for an online sequence of $n$-dimensional vectors $(\vv^t)_{t\ge1}$. Both dynamic algorithms hold against an \emph{oblivious} adversary.
Also assume there is an $\epsilon/2^K$-accurate $T$-evaluator against an \emph{adaptive} adversary.

Then there exists a dynamic algorithm against an \emph{adaptive} adversary that in each iteration returns whp.\ some $\ovv^t$ with $\|\ovv^t - \vv^t\|_\infty \le O(\epsilon \log^2 n)$.
Each iteration takes expected time 
$$O\left(SK+ \frac{T(S)}{2^K} + L + \sum_{i=0}^{K} \frac{C_i(S)}{2^i}\right).$$
\end{theorem}

Note that the $T$-evaluator against an adaptive adversary could just be a method to compute the exact solution statically. Alternatively, one could run several dynamic algorithms against an oblivious adversary in parallel, but use each data structure only once to answer a query.

In the overview of \Cref{sec:overview:dp} we outlined how \Cref{thm:dp} is obtained.
We here give a quick recap.
Let $w$ be the result of the $T$-evaluator and $\vw'$ be the result of one of the other evaluators against an oblivious adversary.
We want to construct an output $\ovw$ whose distribution is similar to $\Normal(\vw, \sigma^2)$ for some variance $\sigma = O(\epsilon \log n)$. Note that w.h.p~$\|\ovw - \vv^t\|_\infty \le O(\epsilon \log^2 n)$ because of the random Gaussian noise we added. We wish to improve upon the naive time of explicitly computing $\ovw$ by directly adding Gaussian noise to $\vw$. 
We achieve this by performing this sampling in a different way which guarantees that we compute $\vw$ explicitly only with some small probability.

Let $d$ be the density function of $\Normal(\vw, \sigma^2)$ and $d'$ be the density function of $\Normal(\vw', \sigma^2)$. Then there is some small $\alpha > 0$ and very unlikely event $D$ with $d'(\vx) \le \exp(\alpha) d(\vx)$ for all $\vx \notin D$.
For example, \Cref{fig:dp} shows density function $d(\vx)$ and the scaled density function $\exp(-\alpha)\cdot d'(\vx)$ for the $1$-dimensional case.
If one were to pick uniformly at random a point below the top curve in \Cref{fig:dp} and return its $\vx$-coordinate, then this corresponds to sampling from $\Normal(\vw, \sigma^2)$.
The same distribution can be obtained by first flipping an unbalanced coin, and with probability $\exp(-\alpha)$ we sample from the area below the bottom curve $\exp(-\alpha)d'(\vx)$ in \Cref{fig:dp} (i.e.~sample according to $\Normal(\vw', \sigma^2)$).
Otherwise, with probability $1-\exp(-\alpha)$, we sample from the area between the two curves. This way we are able to sample from $\Normal(\vw, \sigma^2)$ more efficiently because only with probability $1-\exp(-\alpha)$ must we compute $\vw$.
As computing $\vw'$ is faster than computing $\vw$, the expected time complexity improves.

\begin{figure}
\center
\begin{tikzpicture}
\begin{axis}[every axis plot post/.append style={
  mark=none,domain=-1.5:3.5,samples=50,smooth},
  axis x line*=bottom,
  axis y line*=left,
  enlargelimits=upper,
  yticklabels={,,},
  xticklabels={,,},
  ticks=none]
  \addplot[black] {smallgausso(1.1,0.75)};
  \addplot[black] {gausso(1,0.75)};
\end{axis}
\draw[dotted] (3.1,0) -- (3.1,5.5);
\draw[dotted] (3.25,0) -- (3.25,3);
\node[text width=1cm] at (3.3, -0.25) {$\vv$ $\ovv$};
\end{tikzpicture}
\caption{\label{fig:dp}
Density function $d$ of $\Normal(\vv,\sigma^2)$, and density function $\bar{d}$ of $\Normal(\ovv,\sigma^2)$ scaled by some $\exp(-\alpha)$, $\alpha > 0$ so that $\bar{d}(\vx) \exp(-\alpha) \le d(\vx)$.}
\end{figure}

This scheme is proven formally in \Cref{sec:differential:simple} for the general case where the vectors are $n$-dimensional.
The scheme can be extended recursively: note that in order to sample from $\Normal(\vw', \sigma^2)$, we can use the same scheme again via some $\vw''$, i.e.~when sampling from $\Normal(\vw',\sigma^2)$ we can sample from $\Normal(\vw'', \sigma^2)$ instead with probability $\exp(-\alpha)$.
This is why \Cref{thm:dp} has $K$ many different evaluators with increasing accuracy. The evaluators with higher accuracy are used with smaller probability, thus the expected time complexity improves.
This recursive scheme is proven in \Cref{sec:differential:recursive} and we use it in \Cref{sec:differential:proof} to prove \Cref{thm:dp}.

\subsection{Simulating Gaussian Error}
\label{sec:differential:simple}

Here we prove the algorithm outlined in the previous subsection.
We want to construct a variable with distribution $\Normal(\vv,\sigma^2)$.
This is done by flipping a biased coin:
with probability $\exp(-\alpha)$ we return a vector according to $\Normal(\vu,\sigma^2)$.
Alternatively, with probability $1-\exp(-\alpha)$ we must return a random vector whose distribution we pick in such a way, that the result of our algorithm has distribution $\Normal(\vv,\sigma^2)$. The exact algorithm is given in \Cref{alg:basic_firewall} and \Cref{lem:basic_firewall} stated the guarantees of that algorithm.

\begin{algorithm2e}[t!]
\caption{Basic Simulation Algorithm \label{alg:basic_firewall}}
\SetKwProg{Proc}{procedure}{}{}
\SetKwIF{Prob}{ProbElseProb}{ElseProb}{with probability}{do}{else with probability}{else}{end}
\Proc{\textsc{Simulate}$(\vv \in \R^n, \vu \in \R^n,\alpha \ge 0,\sigma > 0)$}{
	\tcp{Simulates $\vv + \vx$ for $\vx \sim \Normal(0,\sigma^2)$.}
	\Prob{$\exp(-\alpha)$}{
		Sample $\vx \sim \Normal(0,\sigma^2)$ \label{line:basic_sample}\\
		\Return $\vu + \vx$ 
	}
	\While{true\label{line:sampling_while}}{
		Sample $\vx \sim \Normal(0,\sigma^2)$ conditioned on $\frac{|\|\vx\|^2 - \|\vx-\vu+\vv\|^2|}{2\sigma^2} \le \alpha$ \\ 
		\Prob{$1-\exp\left(\frac{\|\vx\|^2-\|\vx-\vu+\vv\|^2}{2\sigma^2}-\alpha\right)$}{
			\Return $\vv + \vx$ \label{line:return_sample}
		}
	}\label{line:sampling_while_end}
}
\end{algorithm2e}

\begin{lemma}\label{lem:basic_firewall}
Let $\vz$ be the result of a call to \textsc{Simulate}$(\vv,\vu,\alpha,\sigma)$ (\Cref{alg:basic_firewall})
with $\sigma \ge 2\ln(1.25/\delta) \epsilon / \alpha$ for any $\delta > 0$ and $\epsilon \ge \|\vv - \vu\|_2$. Then the distribution of $\vz$ has total variation distance at most $\delta$ compared to $\Normal(\vv, \sigma^2)$. Further, the expected time complexity is bounded by $O(n)$.
\end{lemma}

To prove \Cref{lem:basic_firewall}, we first consider the distribution of the result returned by \Cref{line:sampling_while} to \Cref{line:return_sample} of \Cref{alg:basic_firewall}.

\begin{lemma}\label{lem:density}
Consider executing \Cref{line:sampling_while} to \Cref{line:sampling_while_end}
of \Cref{alg:basic_firewall} and let $\vz$ be the returned vector, i.e.~$\vz$ is the output of \Cref{alg:basic_firewall} conditioned on being returned in \Cref{line:return_sample}.
Then the distribution of $\vz$ under this condition has density function
\begin{align*}
d(\vz) = \frac{\exp(-\frac{\|\vz-\vv\|^2}{2\sigma^2}) - \exp(\frac{- \|\vz-\vu\|^2}{2\sigma^2}-\alpha)}{\sqrt{2\pi}\sigma\P\left[|\|\vx\|^2-\|\vx-\vu+\vv\|^2| \le 2\sigma^2\alpha\right](1 - \exp(-\alpha))}
\end{align*}
\end{lemma}

\begin{proof}
Up to normalization, the density function of the distribution of $\vz$ is
\begin{align*}
&~
\frac{1}{\sqrt{2\pi}\sigma} \exp\left(-\frac{\|\vz-\vv\|^2}{2\sigma^2}\right) \left(1 - \exp\left(\frac{\|\vz-\vv\|^2 - \|\vz-\vu\|^2}{2\sigma^2}-\alpha\right)\right) \\
=&~
\frac{1}{\sqrt{2\pi}\sigma} \left(\exp\left(-\frac{\|\vz-\vv\|^2}{2\sigma^2}\right) - \exp\left(\frac{- \|\vz-\vu\|^2}{2\sigma^2}-\alpha\right)\right)
\end{align*}
We now compute the normalization factor. For that note the set of possibly returned vectors is $S = \{\vz \mid |\|\vz-\vv\|_2-\|\vz-\vu\|_2| \le 2 \sigma^2 \alpha\}$
and
\begin{align*}
\int_{\vz\in S} \frac{1}{\sqrt{2\pi}\sigma} \exp\left(-\frac{\|\vz-\vu\|^2}{2\sigma^2}\right) \text{d}\vz
= \int_{\vz\in S} \frac{1}{\sqrt{2\pi}\sigma} \exp\left(-\frac{\|\vz-\vv\|^2}{2\sigma^2}\right) \text{d}\vz \\
= \P[|\|\vx\|^2-\|\vx-\vu+\vv\|^2| \le 2\sigma^2\alpha]
\end{align*}
for some $\vx \sim \Normal(0,\sigma^2)$.
So the normalization factor is
\begin{align*}
\P\left[|\|\vx\|^2 - \|\vx-\vu+\vv\|^2| \le 2\sigma^2\alpha\right](1 - \exp(-\alpha))
\end{align*}
and the density function is
\begin{align*}
\frac{\exp\left(-\frac{\|\vz-\vv\|^2}{2\sigma^2}\right) - \exp\left(\frac{- \|\vz-\vu\|^2}{2\sigma^2}-\alpha\right)}{\sqrt{2\pi}\sigma\P\left[\|\vx\|^2,\|\vx-\vu+\vv\|^2 \le 2\sigma^2\alpha\right](1 - \exp(-\alpha))}
\end{align*}
\end{proof}

Our algorithm relies on the fact that the density function of $\Normal(\vu,\sigma^2)$ is smaller than the density function of $\Normal(\vv,\sigma^2)$ when scaled by $\exp(\alpha)$.
This is generally not true, unless we restrict the two density function on to some event $E\subset\R^n$. Using the following result from differential privacy, we show that this event occurs with high probability, if the variance $\sigma^2$ of the added noise and the scaling-parameter $\alpha$ are sufficiently large.
\begin{lemma}[\!\!{\cite[Appendix~A]{DR14}}]
\label{lem:privacybook}
Let $\vu,\vv \in \R^n$, 
$\epsilon \ge \|\vu-\vv\|$, 
$c^2 > 2\ln(1.25/\delta)$,
$\sigma \ge c \epsilon / \alpha$
and $\vx \sim \Normal(0,\sigma^2)$.
Then $\P[|\|\vx\|^2-\|\vx-\vu+\vv\|^2| > 2\alpha\sigma^2] \le \delta$.
\end{lemma}

We now have all tools available to prove \Cref{lem:basic_firewall}.

\begin{proof}[Proof of \Cref{lem:basic_firewall}]
The density function of $z$ conditioned on $|\|\vz-\vv\|^2-\|\vz-\vu\|^2| \le 2\sigma^2\alpha$ is 
\begin{align*}
&~
\exp\left(-\alpha\right) \cdot \frac{\exp\left(-\frac{\|\vz-\vu\|^2}{2\sigma^2}\right)}{\sqrt{2\pi}\sigma\P\left[~\Big|\|\vx\|^2-\|\vx-(\vu-\vv)\|^2\Big| \le 2\sigma^2\alpha\right]}\\
&~+
(1-\exp\left(-\alpha\right)) \cdot \frac{\exp\left(-\frac{\|\vz-\vv\|^2}{2\sigma^2}\right) - \exp\left(\frac{- \|\vz-\vu\|^2}{2\sigma^2}-\alpha\right)}{\sqrt{2\pi}\sigma\P\left[~\Big|\|\vx\|^2-\|\vx-(\vu-\vv)\|^2\Big| \le 2\sigma^2\alpha\right](1 - \exp\left(-\alpha\right))} \\
=&~
\frac{\exp\left(-\frac{\|\vz-\vv\|^2}{2\sigma^2}\right)}{\sqrt{2\pi}\sigma\P\left[~\Big|\|\vx\|^2-\|\vx-(\vu-\vv)\|^2\Big| \le 2\sigma^2\alpha\right]}
\end{align*}
which is also the density function of $\vx \sim \Normal(0,\sigma^2)$ conditioned on $|\|\vx\|^2-\|\vx-(\vu-\vv)\|^2| \le 2\sigma^2\alpha$.
Thus the total variation distance is bounded by $\delta$ via \Cref{lem:privacybook}.

For the time complexity, note that the probability of returning the vector during an iteration of \Cref{line:sampling_while} is $1-\exp(-\alpha)$.  Consequently, if we reach \Cref{line:sampling_while}, it is invoked $(1-\exp(-\alpha))^{-1}$ times in expectation. The probability of reaching \Cref{line:sampling_while} is $1-\exp(-\alpha)$, so \Cref{line:sampling_while} is invoked $1$ time in expectation. As each iteration needs $O(n)$ time to process the $n$-dimensional vectors, the expected runtime of the procedure is $O(n)$.
\end{proof}

\subsection{Recursive Simulation}
\label{sec:differential:recursive}

In this subsection we provide and analyze a recursive variant of \Cref{alg:basic_firewall}.
This variant replaces \Cref{line:basic_sample} of \Cref{alg:basic_firewall}, which samples some $\Normal(\vu,\sigma^2)$, by a recursive invocation of \Cref{alg:basic_firewall}.
We first analyze the distribution of the returned vector in \Cref{lem:recursive_firewall} and then bound the expected time complexity in \Cref{lem:recursive_query_complexity}.

\begin{algorithm2e}[t!]
\caption{Recursive Simulation Algorithm \label{alg:recursive_firewall}}
\SetKwProg{Proc}{procedure}{}{}
\SetKwIF{Prob}{ProbElseProb}{ElseProb}{with probability}{do}{else with probability}{else}{end}
\Proc{\textsc{Simulate}$(\vv_1,...,\vv_k \in \R^n,\alpha \ge 0,\sigma > 0)$}{
	\tcp{If $k=2$ we call \Cref{alg:basic_firewall} instead.}
	\Prob{$\exp(-\alpha)$}{
		\Return $\textsc{Simulate}(\vv_2,...,\vv_k, 2\alpha,\sigma)$ \tcp{Call Alg.~\ref{alg:basic_firewall} if $k=2$.} \label{line:recurse_simulate}
	}
	\While{true}{
		Sample $\vx \sim \Normal(0,\sigma^2)$ conditioned on $\frac{|\|\vx\|^2-\|\vx-\vv_2+\vv_1\|^2|}{2\sigma^2} \le \alpha$ \\
		\Prob{$1-\exp(\frac{\|\vx\|^2-\|\vx-\vv_2+v_1\|^2}{2\sigma^2}-\alpha)$}{
			\Return $\vv_1 + \vx$
		}
	}
}
\end{algorithm2e}

\begin{lemma}\label{lem:recursive_firewall}
Consider a call to \Cref{alg:recursive_firewall} with inputs
$\vv_1,...,\vv_k \in \R^n$, $k \ge 2$,
$\epsilon \ge \|\vv_i - \vv_{i+1}\|_2 / 2^{i-1}$ for all $i=1,...,k-1$,
$\sigma \ge 2\ln(1.25/\delta) \epsilon / \alpha$,
Then the returned value has total variation distance at most $(k-1)\delta$ compared to $\Normal(\vv_1, \sigma^2)$.
\end{lemma}

\begin{proof}
We prove this by induction over $k$, the number of vectors.
\paragraph{Base Case}
For $k=2$ we call \Cref{alg:basic_firewall} instead, so the claim is true by \Cref{lem:basic_firewall}.
\paragraph{Induction}
Now assume \Cref{lem:recursive_firewall} holds for some $k-1$ and consider a call to \textsc{Simulate} with vectors $\vv_1,...\vv_{k}$.
Let $\vv'_1,...,\vv'_{k-1}$, $\alpha' = 2\alpha$ be the parameters of the recursive call in \Cref{line:recurse_simulate} and let $\epsilon'=2\epsilon$.
Then 
\begin{align*}
\epsilon' =&~ \epsilon/2 \ge \|\vv_i - \vv_{i+1}\|_2 / 2^{i} = \|\vv_{i+1} - \vv_{i}\|_2 / 2^{i-1} = \|\vv'_{i} - \vv'_{i+1}\|_2 / 2^{i-1} \text{ and }\\
\sigma \ge&~ 2\ln(1.25/\delta) \epsilon / \alpha = 2\ln(1.25/\delta) \epsilon' / \alpha'
\end{align*}
so the conditions to apply the induction hypothesis are satisfied.
Thus, the vector returned in \Cref{line:recurse_simulate} has the same distribution as $\Normal(v'_1, \sigma^2) = \Normal(\vv_2, \sigma^2)$ up to total variation distance $(k-2)\delta$.

Note that \Cref{alg:recursive_firewall} is the same as \Cref{alg:basic_firewall}, except for \Cref{line:recurse_simulate}, 
so by the same proof as in \Cref{lem:basic_firewall} we return a vector that is distributed like $\Normal(\vv_1, \sigma^2)$ up to total variation distance $(k-2)\delta+\delta=(k-1)\delta$.
\end{proof}

For computational efficiency, note that \Cref{alg:recursive_firewall} does not need to read vector $\vv_1$ when performing the branch of \Cref{line:recurse_simulate}.
The following lemma bounds the probability of accessing any $\vv_i$ for $i<k$.

\begin{lemma}\label{lem:recursive_query_complexity}
Consider a call to \Cref{alg:recursive_firewall} with inputs
$\vv_1,...,\vv_k \in \R^n, \alpha \ge 0$.
The probability that vector $\vv_i$ is accessed is at most $2^i \alpha$ for all $i < k$.
Further, the expected time complexity (ignoring the time for accessing any $\vv_i$) is bounded by $O(kn)$.
\end{lemma}

\begin{proof}
Vector $\vv_1$ is accessed with probability $1-\exp(-\alpha) \le \alpha \le \alpha 2^1$.
Vector $\vv_i$ for $i > 1$ is accessed with probability $1-\exp(-\alpha 2^{i-2})$ when calling \textsc{Simulate}$(\vv_{i-1},...,\vv_k,\alpha,\sigma)$,
or with probability $1-\exp(-\alpha 2^{i-1})$ when calling \textsc{Simulate}$(\vv_{i},...,\vv_k,\alpha,\sigma)$.
Thus the overall probability after a call \textsc{Simulate}$(\vv_1,...,\vv_k,\alpha,\sigma)$ is
\begin{align*}
&~
\underbrace{\prod_{t=1}^{i-1} \exp\left(-\alpha 2^{t-1}\right)}_{
\text{\begin{tabular}{c}Probability of recursing\\ 
to \textsc{Simulate}$(\vv_i,...,\vv_k,\alpha,\sigma)$\end{tabular}}} 
\left(1-\exp\left(-\alpha 2^{i-1}\right)\right) +
\underbrace{\prod_{t=1}^{i-2} \exp\left(-\alpha 2^{t-1}\right)}_{
\text{\begin{tabular}{c}Probability of recursing\\ 
to \textsc{Simulate}$(\vv_{i-1},...,\vv_k,\alpha,\sigma)$\end{tabular}}}
\left(1-\exp\left(-\alpha 2^{i-2}\right)\right) \\
=&~
\exp\left(-\alpha\sum_{t=0}^{i-2} 2^{t}\right) \left(1-\exp\left(-\alpha 2^{i-1}\right)\right) + \exp\left(-\alpha\sum_{t=0}^{i-3} 2^{t}\right) \left(1-\exp\left(-\alpha 2^{i-2}\right)\right) \\
=&~
\exp\left(-\alpha( 2^{i-1}-1)\right) \left(1-\exp\left(-\alpha 2^{i-1}\right)\right) + \exp\left(-\alpha (2^{i-2}-1)\right) \left(1-\exp\left(-\alpha 2^{i-2}\right)\right) \\
=&~
\exp\left(-\alpha( 2^{i-1}-1)\right) -\exp\left(-\alpha( 2^{i}-1)\right) + \exp\left(-\alpha (2^{i-2}-1)\right) - \exp\left(-\alpha (2^{i-1}-1)\right) \\
=&~
\exp\left(-\alpha (2^{i-2}-1)\right)-\exp\left(-\alpha( 2^{i}-1)\right) \\
=&~
\left(1-\exp\left(-\alpha( 2^{i}-2^{i-2})\right)\right)\cdot\exp\left(-\alpha (2^{i-2}-1)\right) \\
\le&~
1-\exp\left(-\alpha( 2^{i}-2^{i-2})\right) \\
\le&~
\alpha( 2^{i}-2^{i-2}) \le \alpha 2^i
\end{align*}

The time expected time complexity is at most $O(kn)$ because each recursion has expected time $O(n)$ by \Cref{lem:basic_firewall}.
\end{proof}

\subsection{Proof of \Cref{thm:dp}}
\label{sec:differential:proof}

We can now prove \Cref{thm:dp} by applying \Cref{alg:recursive_firewall} to the vectors returned by the evaluator data structures.

\begin{proof}[Proof of \Cref{thm:dp}]
Given the locator and evaluators, we construct a new dynamic algorithm $\mathcal{A}$ against an adaptive adversary. The construction is done in a paragraph further below.
For now, we claim that the output of the new dynamic algorithm $\mathcal{A}$
has the following distribution.

Let $\vw^1$ be the output of the $\epsilon/2^K$-accurate oracle against an adaptive adversary.
Then sample $\vx \sim \Normal(\vw^1, (c_1\epsilon\log n)^2)$ for some sufficiently large constant $c_1$.
At last, set all entries of $\vx$ with absolute value smaller than $c_2 \epsilon\log^2 n$ to $0$.
Call the resulting vector $\vu$.
We claim the dynamic algorithm $\mathcal{A}$ will have the distribution of this vector $\vz$.

Vector $\vz$ satisfies w.h.p.~$\|\vz - \vv^t\|_\infty = O(\epsilon\log^2 n)$,
so returning $\vz$ would satisfy the promised approximation guarantees of \Cref{thm:dp}
and the algorithm would work against an adaptive adversary because the output does not depend on any of the oracles that use the oblivious adversary assumption.

We now describe the new dynamic algorithm $\mathcal{A}$ and how it constructs this vector $\vz$ more efficiently than the procedure described above.

\paragraph{Algorithm}
Let $I$ be the set returned by the $\epsilon$-accurate locator.
For $i > 1$ let $\vw^i_I$ be result of the $\epsilon/2^{K-i+2}$-accurate oracles against oblivious adversaries when querying only entries from $I$.
Let $\vx'_I = \textsc{Simulate}(\vw^1_S,...,\vw^{k+1}_S,2^{-K},c_1 \epsilon \log n)$ (\Cref{alg:recursive_firewall}).
Then set all entries of $\vx$ with absolute value smaller than $c_2 \epsilon \log^2n$ to $0$ and let $\vz'$ be the resulting vector.
Here $c_2>c_1$ is picked such that w.h.p.~$|\vx_i-\vw^1_i| < c_2/2\cdot\epsilon \log^2 n$.
Our algorithm returns this vector $\vz'$.

\paragraph{Correctness}
We claim $\vz'$ has the same distribution as $\vz$ up to total variation distance $1/\poly(n)$.
For $\epsilon' = 2\epsilon/2^K$ we have
$\|\vw^i - \vw^{i+1}\|_2 \le \epsilon/2^{K-i} \le \epsilon' 2^{i-1}$.
So $x'_I$ has distribution $\Normal(\vw^1_I,(c_1\sigma\log n)^2)$ up to total variation distance $1/\poly(n)$ by \Cref{lem:recursive_firewall} for some large enough constant $c_1$.
Thus if $I$ only contained indices $i$ where w.h.p.~$\vz_i$ would be $0$ anyway, then $\vz'$ has same distribution as $z$ up to total variation distance $1/\poly(n)$.

Note that by $\|\vw^1 - \vv^t\|_2 < \epsilon$ we have that $I$ (which by definition contains all indices with $|\vv^t_i| > \epsilon$) also contains all indices $i$ with $|\vw^1_i| > 2\epsilon$.
Further, w.h.p.~we have $\|\vw^1-\vx\|_\infty < c_2/2 \epsilon \log^2 n$ by choice of $c_2>c_1$.
So $i\in I$ this would imply
$|\vx_i| \le |\vw^1_i| + |\vw^1_i-\vx_i| \le 2\epsilon + c_2/2 \epsilon \log^2 n < c_2 \epsilon \log^2 n$
so w.h.p.~$\vz_i$ will be set to $0$.
Thus the total variation distance of $\vz$ and $\vz'$ is at most some $1/\poly(n)$.

\paragraph{Complexity}

By \Cref{lem:recursive_query_complexity}, we use each $\vw_i$ with probability at most $2^i/2^K = 2^{i-K}$ for $i \le K$ and running $\textsc{Simulate}$ on $K+1$ many $|I|$-dimensional vectors needs $O(|I|K)=O(SK)$ time. We can delay the query to $\vw^i$ until the vectors actually need to be used.
As $\vw_i$ is obtained from evaluator with complexity $C_{K-i+2}$ for $i > 1$, we obtain time complexity
$$
O(SK + L + C_0(S) + T(S)/2^K + \sum_{i=1}^K \frac{C_i(S)}{2^i}).
$$

\end{proof}

\section{Interior Point Method}

\label{sec:ipm}

In this section we provide the machinery we use to reduce minimum
cost flow to dynamic graph data structure problems. First, in \cref{sec:ipm_framework}
we provide the general IPM framework for linear programming from \cite{DLY21}
that we use. Then, in \cref{sec:ipm_tools} we introduce the
data structures, subroutines, and bounds that we develop in this paper
to implement this framework efficiently and in \cref{sec:ipm_implementation}
we combine these pieces to give the efficient IPM. The proofs for the tools we introduce are provided in \cref{sec:solution_estimation}, \cref{sec:ipm_stability}, and \cref{sec:combine} (for the runtime bound for the graph solution maintainer (\cref{def:sol_maintainer}) in \cref{thm:graph_solution_approximation}).

\subsection{Robust IPM Framework}

\label{sec:ipm_framework}

Here we provide the the linear programming setup that we use to model
minimum cost flow and the IPM framework provided by \cite{DLY21} for
solving them. In particular, throughout the section, we consider the
general linear programming problem. Given $\mb\in\R^{m\times n}$,
$\vc,\vl,\vu\in\R^{m}$, and $\vd\in\R^{n}$ where $\vl<\vu$ entrywise,
we wish to solve.
\begin{equation}
\min_{\vx\in\R^{m}\,|\,\mb^{\top}\vx=\vd\text{ and }\vl\leq\vx\leq\vu}\vc^{\top}\vx\,.\label{eq:lp_formula}
\end{equation}
In the special case where $\mb$ is the incidence matrix of graph
and $\vl=\vzero$, this problem directly corresponds to the minimum
cost flow problem. Many of the reductions we provide in this section
apply to this general linear program and we will explicitly state
in which cases we instead assume that $\mb$ is the incidence matrix
of graph.

To solve (\ref{eq:lp_formula}) we leverage the general robust IPM
framework of \cite{DLY21}. This method crudely follows a central path
by maintaining centered points defined as follows. 
\begin{defn}[Centered Point]
\label{def:centered} For $\X\defeq\{\vx\in\R^{m}\,|\,\vx_{i}\in(\vl_{i},\vu_{i})\}$
we say $(\vx,\vs)\in\R^{m}\times\R^{m}$ is\emph{ $\mu$-feasible
}for $\mu>0$ if $\vx\in\mathcal{X}$, $\mb^{\top}\vx=\vd$, $\mb\vy+\vs=\vc/\mu$
for some $\vy\in\R^{n}$. We say $(\vx,\vs)$ is \emph{$\mu$-centered}\footnote{In \cite{DLY21}, the condition is $\mb\vy+\vs=\vc$ and $\norm{\hess\phi(\vx)^{-1/2}(\vs/\mu+\grad\phi(\vx))}_{\infty}\leq\frac{1}{64}$
instead. We do the replacement from $\vs/\mu$ to $\vs$ to simplify
the algorithm description and notations in the data structures. Further,
the choice of variable names is different in the two papers with variable names chosen here for the application of minimum cost flow.} if $(\vx,\vs)$ is $\mu$-feasible and $\norm{\hess\phi(\vx)^{-1/2}(\vs+\grad\phi(\vx))}_{\infty}\leq\frac{1}{64}$
where $\phi(\vx)\defeq\sum_{i\in[m]}\phi_{i}(\vx)$ with $\phi_{i}(\vx)\defeq-\log(\vu_{i}-\vx_{i})-\log(\vx_{i}-\vl_{i})$
for $\vx\in\X$.
\end{defn}

This definition is motivated for the fact that, \emph{$\mu$-central path point}, defined as
\[
\vx_{\mu} \defeq \argmin_{\vx\in\X\,|\,\mb^{\top}\vx=\vd}\mu\cdot\vc^{\top}\vx+\phi(\vx)
\]
is the unique $\mu$-centered point with $\|\hess\phi(\vx)^{-1/2}(\vs+\grad\phi(\vx))\|_{\infty}=0$.
To see this, note that $\phi$ is convex on $\X$ and the optimality
conditions for $\vx_{\mu}$ are that 
\[
\vx\in\X
\text{ , }
\mb^{\top}\vx_{\mu}=\vd
\text{ and }
\vc+\mu\grad\phi(\vx_{\mu})\perp\ker(\mb^{\top})\,.
\] However,
$\vc+\mu\grad\phi(\vx_{t})\perp\ker(\mb^{\top})$ if and only if $\vc+\mu\grad\phi(\vx_{\mu})\in\im(\mb)$
which we can write equivalently as $\vc+\mu\grad\phi(\vx_{\mu})=\mu\mb\vy_{\mu}$
for some $\vy_{\mu}$. Finally, the condition $\vc+\mu\grad\phi(\vx_{\mu})=\mu\mb\vy_{\mu}$
is equivalent to $\mb\vy_{\mu}+\vs_{\mu}=\vc/\mu$ for $\vs_{\mu}=-\grad\phi(\vx_{\mu})$,
i.e. $\|\hess\phi(\vx_{\mu})^{-1/2}(\vs_{\mu}+\grad\phi(\vx_{\mu}))\|_{\infty}=0$
as $\hess\phi(\vx_{\mu})$ is positive definite. Consequently, a $\mu$-centered
point is a point which maintains an approximate notion of the optimality
of $\vx_{\mu}$.

The IPM framework works by maintaining $\mu$-centered points
by controlling centrality measures as potentials. The definition of
these quantities (Definition~\ref{def:centrality}), the framework
(Algorithm~\ref{algo:rIPMAbs}), and the result from \cite{DLY21}
that we use about this framework (Theorem~\ref{thm:rIPMAbs}) are all
given below.
\begin{defn}[Centrality]
\label{def:centrality} For $\mu$-feasible $(\vx,\vs)$ we define
\emph{centrality measure} $\gamma(\vx,\vs)\in\R^{m}$ where $\gamma_{i}(\vx,\vs)\defeq\phi_{i}''(\vx)^{-1/2}(\vs_{i}+\phi_{i}'(\vx))$
and $\phi_{i}'(\vx)\defeq[\grad\phi(\vx)]_{i}$ and $\phi_{i}''(\vx)\defeq[\hess\phi(\vx)]_{ii}$.
Further, we define\emph{ centrality potential} $\Psi(\vx,\vs)\defeq\sum_{i\in[m]}\cosh(\lambda\cdot\gamma_{i}(\vx,\vs))$
where $\lambda\defeq128\log(16m)$ and $\cosh(z)\defeq\frac{1}{2}[\exp(z)+\exp(-z)]$
for all $z\in\R$. 
\end{defn}

\begin{algorithm2e}[!ht]
\SetKwProg{Proc}{procedure}{}{}
\SetKwRepeat{Do}{do}{while}

\caption{A robust interior point method in \cite{DLY21}
\label{algo:rIPMAbs}}

\Proc{$\textsc{Centering}(\mb,\vx,\vs,\vl,\vu,\mu_{\mathrm{start}},\mu_{\mathrm{end}})$}{

\tcp{Invariant: $(\vx,\vs)$ is $\mu$-centered with $\Psi(\vx,\vs)\leq\cosh(\lambda/64)$}
\tcp{(See Definitions~\ref{def:centered} and \ref{def:centrality})}

$\ $

Define step size $\alpha\defeq\frac{1}{2^{15}\lambda}$.

$\overline{\mu}=\mu=\mu_{\mathrm{start}}$, $\ox=\vx$, $\os=\vs$ 

\While{$\mu\geq\mu_{\mathrm{end}}$}{

Set weight matrix $\mw\gets\nabla^{2}\phi(\ox)^{-1}$.

Set iterate approximation $(\ox,\os)\in\R^{m}\times\R^{m}$ such that
\[
\|\mw^{-1/2}(\ox-\vx)\|_{\infty}\leq\alpha
\text{ and }
\|\mw^{1/2}(\os-\vs)\|_{\infty}\leq\alpha
\,.
\]

Set step direction $\vv\in\R^{m}$ where $\vv_{i}\gets\sinh(\lambda\gamma_{i}(\ox,\os))$
for all $i\in[m]$ and $\sinh(z)\defeq\frac{1}{2}(\exp(z)-\exp(-z))$.

Set step size $h\gets-\alpha/\|\cosh(\lambda\gamma(\ox,\os))\|_{2}$.

Set $\vv^{\|}$, $\vv^{\perp}$ such that $\mw^{-1/2}\vv^{\|}\in\mathrm{Im}\mb$,
$\mb^{\top}\mw^{1/2}\vv^{\perp}=0$, and
\[
\|\vv^{\|}-\mproj_{\mw}v\|_{2}\leq\alpha\|\vv\|_{2}\text{ and }\|\vv^{\perp}-(\mi-\mproj_{\mw})v\|_{2}\leq\alpha\|\vv\|_{2}\,
\]
where $\mproj\defeq\mw^{1/2}\mb(\mb^{\top}\mw\mb)^{-1}\mb^{\top}\mw^{1/2}$

Set $\vx\leftarrow\vx+h\mw^{1/2}\vv^{\perp}$, $\vs\gets\vs+h\mw^{-1/2}\vv^{\|}$,
$\mu\leftarrow\max\{(1-\frac{\alpha}{64\sqrt{m}})\mu,\mu_{\mathrm{end}}\}$\label{line:step}

\textbf{If} $|\overline{\mu}-\mu|\geq\alpha\overline{\mu}$, \textbf{then}
$\vs\leftarrow\frac{\mu}{\overline{\mu}}\vs$, $\overline{\mu}\leftarrow\mu$

}

\textbf{Return} $(\vx,\vs)$

}

\end{algorithm2e}
\begin{thm}[Theorem A.16 in \cite{DLY21}]
\label{thm:rIPMAbs} Using the notation in Algorithm~\ref{algo:rIPMAbs},
let $(\vx^{(0)},\vs^{(0)})$ be the value of $(\vx,\vs)$ before the
step (Line~\ref{line:step}) and let $(\vx^{(1)},\vs^{(1)})$ be
the value $(\vx,\vs)$ after the step. If $(\vx^{(0)},\vs^{(0)})$
is $\overline{\mu}$-feasible and $\Psi(\vx^{(0)},\vs^{(0)})\leq\cosh(\lambda/64)$,
then $(\vx^{(1)},\vs^{(1)})$ is $\overline{\mu}$-feasible and
\[
\Psi(\vx^{(1)},\vs^{(1)})\leq\left(1-\frac{\alpha\lambda}{8\sqrt{m}}\right)\Psi(\vx^{(0)},\vs^{(0)})+\alpha\lambda\sqrt{m}\leq\cosh(\lambda/64).
\]
\end{thm}

\begin{proof}
The proof of \cite[Theorem A.16]{DLY21} shows $\Psi^{\mu'}(\vx^{(1)},\vs^{(1)})\leq\left(1-\frac{\alpha\lambda}{8\sqrt{m}}\right)\Psi^{\mu'}(\vx^{(0)},\vs^{(0)})+\alpha\lambda\sqrt{m}$
for any $|\mu'-\mu|\leq\alpha\mu$ where $\Psi^{\mu}(\vx,\vs)\defeq\norm{\hess\phi(\vx)^{-1/2}(\vs/\mu+\grad\phi(\vx))}_{\infty}$.
We picked $\mu'=\overline{\mu}$ and replaced $\vs/\overline{\mu}$
by $\vs$.
\end{proof}

\subsection{Robust IPM Tools}

\label{sec:ipm_tools}

Here we discuss the key tools we develop in this paper to efficiently
implement the robust IPM (Algorithm~\ref{algo:rIPMAbs}) of \cite{DLY21} discussed
in the previous Section~\ref{sec:ipm_framework}. 

First, as discussed in Section \ref{sec:keypieces}, a key advance of this paper is efficient
procedures for approximately maintaining the iterates of Algorithm~\ref{algo:rIPMAbs},
i.e. approximating the result of approximate projection steps. We
formalize this maintenance problem as a data structure problem defined
below. 
\global\long\def\runtime{\mathcal{T}}%
\global\long\def\timeInit{\mathcal{T}_{\mathrm{init}}}%
\global\long\def\timePhase{\mathcal{T}_{\mathrm{phase}}}%
\global\long\def\timeApprox{\mathcal{T}_{\mathrm{approx}}}%
\global\long\def\dsInit{\textsc{Initialize}}%
\global\long\def\dsStart{\textsc{StartPhase}}%
\global\long\def\dsMove{\textsc{Move}}%
\global\long\def\dsApprox{\textsc{Approx}}%
\global\long\def\dsMaintain{\textsc{Maintain}}%

\begin{defn}[Solution Maintainer]
\label{def:sol_maintainer} We call a data structure a \emph{$(\timeInit,\timePhase)$-solution
maintainer} if it supports
the following operations against an adaptive adversary with high probability:
\begin{itemize}
\item $\textsc{\ensuremath{\dsInit}}(\mb\in\R^{m\times n},\vw^{(0)}\in\R_{+}^{m},\vx^{(0)}\in\R^{m},\vs^{(0)}\in\R^{m},\alpha,C_r,k,C_z)$:
Given input constraint matrix $\mb$, weight vector $\vw^{(0)}$,
iterate $(\vx^{(0)},\vs^{(0)})$, accuracy parameter $\alpha$,
weight range $r$, phase length $k$, sparsity of changes $z$, initialize
the data structure with $\vw:=\vw^{(0)}$, $\vx:=\vx^{(0)}$, and
$\vs:=\vs^{(0)}$ in time $O(\timeInit)$ with $\timeInit = \Omega(m)$.
\item $\dsStart(\widetilde{\vx}\in\R^{m},\widetilde{\vs}\in\R^{m})$: Given
input iterate $(\widetilde{\vx},\widetilde{\vs})$ with $\|\mw^{-1/2}(\widetilde{\vx}-\vx)\|_{2}\leq1$
and $\|\mw^{1/2}(\widetilde{\vs}-\vs)\|_{2}\leq1$, update $\vx\gets\widetilde{\vx}$
and $\vs\gets\widetilde{\vs}$ in amortized $O(\timePhase)$ time with $\timePhase=\Omega(m)$.
\item $\dsMove(\vw^{(j)}\in\R_{+}^{m},\vv^{(j)}\in\R^{m},h^{(j)}\in\R)\rightarrow\R^{m}\times\R^{m}$:
In the $j$-th call to $\dsMove$, given input weights $\vw^{(j)}$,
direction $\vv^{(j)}$, and step size $h^{(j)}$ with $h^{(j)}\|\vv^{(j)}\|_{2}\leq1$,  $\dsMove$
updates $\vw \gets \vw^{(j)}$, 
\begin{align*}
\vx \leftarrow\vx+h^{(j)}\mw_{j}^{1/2}(\mi-\mP_{j})\vv^{(j)}
 \text{, and } 
\vs \leftarrow\vs+h^{(j)}\mw_{j}^{-1/2}\mP_{j}\vv^{(j)},
\end{align*}
where $\mw_{j}\defeq\mdiag(\vw^{(j)})$ and $\mproj_{j}\defeq\mw_{j}^{1/2}\mb(\mb^{\top}\mw_{j}\mb)^{-1}\mb^{\top}\mw_{j}^{1/2}$ and  $\dsMove$ outputs $(\ox^{(j)},\os^{(j)}) \in \R^m \times \R^m$ with $\|(\mw^{(j)})^{-1/2}(\ox^{(j)}-\vx)\|_{\infty}\leq\alpha$,
$\|(\mw^{(j)})^{1/2}(\os^{(j)}-\vs)\|_{\infty}\leq\alpha$, and
the number of coordinates changed from the previous
output bounded by $O(2^{2\ell_{j+1}}\alpha^{-2}\log^{3}m+S_{j})$
where
\[
S_{j}
\defeq
\left|\left\{
i\in[m] : \vw_{i}^{(j)}\neq\vw_{i}^{(j-1)}\text{, }\ox_{i}^{(j-1)}=\ox_{i}^{(j-2)}\text{, and }\os_{i}^{(j-1)}=\os_{i}^{(j-2)}
\right\}\right|\,.
\]
The input $\vw^{(j)}$ and $\vv^{(j)}$ and output $(\ox^{(j)},\os^{(j)})$
are given implicitly as a list of changes to the previous input and
output of $\dsMove$.
\end{itemize}
Furthermore, the above operations need only be supported under the
following assumptions:
\begin{enumerate}
\item \emph{Phase length}: $\textsc{StartPhase}$ is called at least every $k$
calls to $\textsc{Move}$ and at most twice in a row.
\item \emph{Number of changes}: for all $j\geq1$ there are at most $\min\{C_z2^{2\ell_{j}},m\}$
coordinates changed in $\vw^{(j)}$, $\vv^{(j)}$ from $\vw^{(j-1)},\vv^{(j-1)}$
where $\ell_{j}$ is the largest integer with $\ell$ with $j \equiv 0 \pmod{2^\ell}.$
\item \emph{Magnitude of changes}: for any $|j_{2}-j_{1}|\leq L$, we have $\sqrt{\vw_{i}^{(j_{2})}/\vw_{i}^{(j_{1})}}\leq C_rL^{2}$
for all $i\in[m]$.
\end{enumerate}
\end{defn}
Our algorithm actually always has $S_j = 0$, but we state \cref{def:sol_maintainer} with possibly nonzero $S_j$ for more generality.

In the particular case of graphs, one of the key results of this paper is the following
efficient solution maintenance data structure in the particular case
of graphs (shown in \cref{subsec:efficientsoln}).
\begin{restatable}[Graph Solution Maintenance]{theorem}{restatesol}
\label{sec:graph_solution_maintenance}
In the special case that $\mb$ is the incidence matrix of a $m$-edge, $n$-node graph,
if $C_r, C_z = \O(1)$ and $\alpha = \widetilde{\Omega}(1)$, there is a $(\timeInit,\timePhase)$-solution maintainer (Definition~\ref{def:sol_maintainer}) with $\timeInit = \O(m)$ and $\timePhase = \O(m + m^{15/16}k^{29/8})$.
\end{restatable}

Note that in the solution maintenance data structure problem it is
required that $\dsStart$ be called at least every $k$ calls to $\dsMove$.
Consequently, to apply this data structure to implement the robust
IPM framework the input $\hat x$ and $\hat s$ to $\dsStart$, i.e.
weighted $\ell_{2}$ approximations to $(x,s)$, need to be computed
efficiently. We formalize this problem below.

\begin{defn}[Solution Approximation]
\label{def:solution_approximation} We call a procedure\emph{ $\timeApprox$-approximator}
if given $\mu$-feasible $(\vx,\vs)$, weights $\vw^{(1)},\cdots,\vw^{(k)}\in\R_{+}^{m}$,
directions $\vv^{(1)},\cdots,\vv^{(k)}\in\R^{m}$, and step sizes
$h^{(1)},\cdots,h^{(k)}$ such that
\begin{itemize}
\item $h^{(i)}\|\vv^{(i)}\|_{2}\leq1$,
\item all the changes in $\vw$ and $\vv$ are supported on $z$ many edges
and the input is given as these changes,
\item $\frac{1}{r}\leq\sqrt{\vw_{\ell}^{(i)}/\vw_{\ell}^{(j)}}\leq r$ for
all $i,j\in[k]$ and $\ell\in[m]$,
\end{itemize}
with high probability, we can compute $\mu$-feasible $(\widetilde{\vx},\widetilde{\vs})$
such that
\[
\left\|\widetilde{\vx}-\vx-\sum_{i\in[k]}h^{(i)}\mw_{i}^{1/2}(\mi-\mproj_{\mw_{i}})\vv^{(i)}\right\|_{\mw_{k}^{-1}}\leq\epsilon\text{ and }
\left\|\widetilde{\vs}-\vs-\sum_{i\in[k]}h^{(i)}\mw_{i}^{-1/2}\mproj_{\mw_{i}}v^{(i)} \right\|_{\mw_{k}}\leq\epsilon
\]
in $O(\timeApprox)$ time where $\mw_{i}\defeq\mdiag(\vw^{(i)})$.
\end{defn}

In the particular case of graphs, in Section~\ref{sec:solution_estimation}
we provide the following theorem on efficient solution approximation.
\begin{thm}
\label{thm:graph_solution_approximation} In the special case that
$\mb$ is the incidence matrix of a $m$-edge, $n$-node graph there
is \emph{$\timeApprox$-approximator} with $\timeApprox=\widetilde{O}(m+z k^3 r^2 \epsilon^{-2})$.
\end{thm}

Finally, to apply these results, we need to prove that the weights
which in turn are induced by $\hess\phi(\vx)$ do not change by too
much. For this, in Appendix~\ref{sec:ipm_stability} we prove the
following. The statement is similar to the bound in \cite[Lemma 6.5]{GLP21} generalized
to our setting. Note that this bound applies to the IPM framework
regardless of whether or not $\mb$ is the incidence matrix of a graph.
\begin{lem}
\label{lem:reldist}For $\mu^{(0)}$-centered $(\vx^{(0)},\vs^{(0)})$
and $\mu^{(1)}$-centered $(\vx^{(1)},\vs^{(1)})$ with $\mu^{(0)}\approx_{1/32}\mu^{(1)}$,
if $\veta^{(j)}\defeq\vs^{(j)}+\nabla\phi(\vx^{(j)})$ for $j\in\{0,1\}$
it follows that
\[
\sum_{\phi_{i}''(\vx^{(0)})^{1/2}\geq3\phi_{i}''(\vx^{(1)})^{1/2}}\sqrt{\frac{\phi_{i}''(\vx^{(0)})}{\phi_{i}''(\vx^{(1)})}}\leq2^{10}\sum_{i\in[m]}\frac{(\veta_{i}^{(1)}-\veta_{i}^{(0)})^{2}}{\phi''_{i}(\vx^{(0)})+\phi''_{i}(\vx^{(1)})}+2^{4}m\left(\frac{\mu^{(1)}-\mu^{(0)}}{\mu^{(0)}}\right)^{2}\,.
\]
\end{lem}

\subsection{Robust IPM Implementation}

\label{sec:ipm_implementation}

Here we show how to use the tools of Section~\ref{sec:ipm_tools}
to efficiently implement the Algorithm~\ref{algo:rIPMAbs}. The algorithm,
Algorithm~\ref{algo:rIPMAbs}, and its analysis, Lemma~\ref{lem:ipm_imp},
are given below.

\begin{algorithm2e}[!ht]
\SetKwProg{Proc}{procedure}{}{}
\SetKwRepeat{Do}{do}{while}
\SetEndCharOfAlgoLine{}

\caption{Algorithm~\ref{algo:rIPMAbs} Implementation with Solution
Maintenance and Estimation \label{algo:rIPM}}

\Proc{$\textsc{CenteringImpl}(\mb,\vx,\vs,\vl,\vu,\mu_{\mathrm{start}},\mu_{\mathrm{end}},k)$}{

Define step size $\alpha\defeq\frac{1}{2^{15}\lambda}$, weight range
$C_r=\Theta(1)$, and sparsity parameter $C_z=\Theta(\log^{5}m)$

Set $\overline{\mu}=\mu=\mu_{\mathrm{start}}$, $\ox=\vx$, $\os=\vs$,
$\vw=\diag(\nabla^{2}\phi(\ox)^{-1})$, $j=0$

$\textsc{Sol}.\textsc{Initialize}(\mb,\vw,\vx,\vs,\alpha,C_r,k,C_z)$
where \textsc{Sol} is a $(\timeInit,\timePhase)$-solution maintainer (Definition~\ref{def:sol_maintainer})

\While{$\mu\geq\mu_{\mathrm{end}}$}{

\If(\tcp*[h]{Reset every $k$ iterations}){$k$ divides $j$ or $\mu=\mu_{\mathrm{end}}$}{
Let $(\vx,\vs)$ be the solution $\textsc{Sol}$ implicitly maintained.

Find $\overline{\mu}$-feasible $(\widetilde{\vx},\widetilde{\vs})$
with $\|\mw^{-1/2}(\vx-\widetilde{\vx})\|_{2}\leq\frac{\alpha}{100}$
and $\|\mw^{1/2}(\vs-\widetilde{\vs})\|_{2}\leq\frac{\alpha}{100}$
by using a $\timeApprox$-approximator (Definition~\ref{def:solution_approximation})
with $k=k,r=O(k^{4}),z=\widetilde{O}(k^{2})$.

\If{$|\overline{\mu}-\mu|\geq\alpha\overline{\mu}$}{

$\textsc{Sol}.\textsc{Initialize}(\mb,\vw,\vx,\vs,\alpha,r,k,z)$

$\widetilde{\vs}\leftarrow\frac{\mu}{\overline{\mu}}\widetilde{\vs}$,
$\overline{\mu}\leftarrow\mu$ \tcp*{Reinitialize. All coordinates
may have changed}

}

$\textsc{Sol}.\textsc{StartPhase}(\widetilde{\vx},\widetilde{\vs})$

}

$\ $

\tcp{ Step: $\vx\leftarrow\vx+h\mw^{1/2}(\mi-\mP_{\mw})\vv,\vs\leftarrow\vs+h\mw^{-1/2}\mP_{\mw}\vv$}

Set the direction $\vv_{i}=\sinh(\lambda\gamma_{i}(\ox,\os))$ and
the step size $h=-\alpha/\|\cosh(\lambda\gamma(\ox,\os))\|_{2}$.

$(\ox,\os)\leftarrow\textsc{Sol}.\textsc{Move}(\vw,\vv,h)$

$\mu\leftarrow\max((1-\frac{\alpha}{64\sqrt{m}})\mu,\mu_{\mathrm{end}})$,
$\vw\leftarrow\diag(\nabla^{2}\phi(\ox)^{-1})$, and $j\leftarrow j+1$

}

\textbf{Return} $(\vx,\vs)$

}

\end{algorithm2e}
\begin{thm}
\label{lem:ipm_imp} For any $k\geq1$, $\mu_{\mathrm{end}}\leq\mu_{\mathrm{start}}$,
and $\mu_{\mathrm{start}}$-centered $(\vx,\vs)$ with $\Psi(\vx,\vs)\leq\cosh(\lambda/64)$,
Algorithm \ref{algo:rIPM} outputs a $\mu_{\mathrm{end}}$-centered
$(\vx',\vs')$ with $\Psi(\vx',\vs')\leq\cosh(\lambda/64)$ in time
$\widetilde{O}((\timeInit+\frac{\sqrt{m}}{k}(\timePhase+\timeApprox))\log(\mu_{\mathrm{start}}/\mu_{\mathrm{end}}))$. 
\end{thm}

\begin{proof}
First, we verify that the conditions of the solution maintenance data
structure (Definition~\ref{def:sol_maintainer}) are satisfied with
$C_r,C_z$ defined as in the Algorithm~\ref{algo:rIPMAbs}. 
\begin{itemize}
\item $C_z$: Note that both $\vw,\vv$ are entrywise functions of $\ox,\os$
and Definition~\ref{def:sol_maintainer} promises that $\ox,\os$
changes in at most $O(2^{2\ell_{j+1}}\epsilon^{-2}\log^{3}m+S_{j})$ coordinates.
Since we only change $\vw$ when $\ox$ or $\os$ changes, we have
$S_{j}=0$. Using the parameter choice $\epsilon=\Theta(1/\log m)$,
the number of changes is bounded by $O(2^{2\ell_{j+1}}\log^{5}m)$.
This verifies the condition $C_z=O(\log^{5}m)$.
\item $C_r$: For any two iterations $\vx^{(j_{1})}$
and $\vx^{(j_{2})}$ associated with path parameters $\mu^{(j_1)}$ and $\mu^{(j_2)}$, Lemma~\ref{lem:reldist} shows that 
\[
\sqrt{\frac{\phi_{i}''(\vx^{(j_{1})})}{\phi_{i}''(\vx^{(j_{2})})}}\leq3+2^{10}\sum_{i\in[m]}\frac{(\veta_{i}^{(j_{2})}-\veta_{i}^{(j_{1})})^{2}}{\phi''_{i}(\vx^{(j_{2})})+\phi''_{i}(\vx^{(j_{1})})}+2^{4}m\left(\frac{\mu^{(j_{2})}-\mu^{(j_{1})}}{\mu^{(j_{1})}}\right)^{2}
\]
where $\veta^{(j)} \defeq\vs^{(j)}+\nabla\phi(\vx^{(j)})$. Note that to apply this lemma, we used that all iterations are centered (which we will show later)
and that $\mu^{(j_{1})}\approx_{1/32}\mu^{(j_{2})}$ (since we reinitialize
the data structure every $\widetilde{\Theta}(\sqrt{m})$ steps). For $L=|j_{1}-j_{2}|$, we have
$|\mu^{(j_{2})}-\mu^{(j_{1})}|\leq\frac{\alpha L}{32\sqrt{m}}\mu^{(j_{1})}$,
so
\begin{align*}
\sqrt{\frac{\phi_{i}''(\vx^{(j_{1})})}{\phi_{i}''(\vx^{(j_{2})})}} & \leq3+2^{10}\sum_{i\in[m]}\frac{(\veta_{i}^{(j_{2})}-\veta_{i}^{(j_{1})})^{2}}{\phi''_{i}(\vx^{(j_{2})})+\phi''_{i}(\vx^{(j_{1})})}+O(\alpha^{2}L^{2}).
\end{align*}
To bound the first term, we note that every term in the summation
is bounded by $O(1)$. We split the sum into two cases. The first
case is when $\phi''_{i}(\vx)$ does not change by more than a $O(1)$ factor.
In this case, one can prove that $\|\veta^{(j+1)}-\veta^{(j)}\|_{\phi''(\vx^{(j)})}=O(\alpha)$
because $\|\vx^{(j+1)}-\vx^{(j)}\|_{(\mw^{(j)})^{-1}}\leq\alpha$,
$\|\vs^{(j+1)}-\vs^{(j)}\|_{\mw^{(j)}}\leq\alpha$ and $\mw^{(j)}\approx_{O(1)}\nabla^{2}\phi(\vx^{(j)})^{-1}$.
Therefore, after $L$ steps, the sum for the first case is bounded
by $O(\alpha^{2}L^{2})=O(L^{2})$. For the second case, we can use
$\|\vx^{(j+1)}-\vx^{(j)}\|_{(\mw^{(j)})^{-1}}\leq\alpha$ to show
that there are at most $O(L^{2}\alpha^{2})$ coordinates where $\phi''$
changes by more than a constant multiplicative factor. Hence, this shows
that 
\begin{equation}
\sqrt{\frac{\phi_{i}''(\vx^{(j_{1})})}{\phi_{i}''(\vx^{(j_{2})})}}=O(L^{2}).\label{eq:phi_ratio_proof}
\end{equation}
This verifies the condition $C_r=O(1)$.
\end{itemize}
Now, we bound the potential. Theorem \ref{thm:rIPMAbs} shows that
\[
\Psi(\vx^{\mathrm{new}},\vx^{\mathrm{new}})\leq\left(1-\frac{\alpha\lambda}{8\sqrt{m}}\right)\Psi(\vx,\vx)+\alpha\lambda\sqrt{m}
\]
for every step (excluding the effect of $\textsc{StartPhase}$). For
$\textsc{StartPhase}$, we have that $\|\mw^{-1/2}(\vx-\widetilde{\vx})\|_{2}\leq\frac{\alpha}{100}$
and $\|\mw^{1/2}(\vs-\widetilde{\vs})\|_{2}\leq\frac{\alpha}{100}$.
This increases $\Psi$ by at most $\frac{\alpha\lambda}{16\sqrt{m}}\Psi(\vx,\vs)$
additively. Finally, for the change of $\overline{\mu}$, it would
increase $\Psi$ by at most $2\alpha\lambda\Psi(\vx,\vs)$, but this happens
every $32\sqrt{m}$ steps. Therefore $\Psi$ is decreasing on average and stays polynomially bounded.

Next, we discuss the parameters for the $\timeApprox$-approximator.
The number of terms is exactly given by $k$. For the number of coordinate
changes $z$, Definition~\ref{def:sol_maintainer} promised that
$\ox,\os$ changes by $\widetilde{O}(2^{2\ell_{j+1}})$ coordinates
at the $j$-th step. Since we restart every $k$ iterations by calling \textsc{StartPhase}, by aligning
our steps numbers appropriately, we have that $\sum_{j\text{ in a phase}}\widetilde{O}(2^{2\ell_{j}})=\widetilde{O}(\max_{j\text{ in a phase}}2^{2\ell_{j}})=\widetilde{O}(k^{2})$.
Finally, the weight ratio is due to (\ref{eq:phi_ratio_proof}) with
$L=k$.

Finally, for the runtime, note that there are $\widetilde{O}(\sqrt{m}\log(\mu_{\mathrm{start}}/\mu_{\mathrm{end}}))$
steps. For every $k$ steps, we use a $\timeApprox$-approximator
and call $\textsc{Sol}.\textsc{StartPhase}$ and they cost $\timePhase$
and $\timeApprox$ respectively. All other costs are linear in the
output size of the data structure and are not bottlenecks. Therefore,
the total cost is $\widetilde{O}((\timeInit+\frac{\sqrt{m}}{k}(\timePhase+\timeApprox))\log(\mu_{\mathrm{start}}/\mu_{\mathrm{end}}))$. 
\end{proof}

\subsection{Efficient Solution Approximation}
\label{sec:solution_estimation}

In this section, we prove Theorem~\ref{thm:graph_solution_approximation}.
Our algorithms leverage
two powerful tools from algorithmic graph theory, in particular nearly
linear time algorithms for subspace sparsification \cite{LS18}.

\begin{proof}[Proof of Theorem~\ref{thm:graph_solution_approximation}]
Our algorithm for approximating $\widetilde{\vx}$ involves two steps, we first find a $\vx'$ such that it is close to the true vector $\vx^{*}\defeq\vx+\sum_{i\in[k]}h^{(i)}\mw_{i}^{1/2}(\mi-\mproj_{\mw_{i}})\vv^{(i)}$, but may not satisfies $\mb^\top \vx' = \vd$. Then, we show how to use $\vx'$ to find $\widetilde{\vx}$ that is close to $\vx^{*}$ and satisfies $\mb^\top \vx' = \vd$.

Let $S\subseteq[m]$ be the set of at most $z$ coordinates of $\vw$
and $\vv$ that change and let $C\subseteq[n]$ be an arbitrary subset
(that we set later) such that every edge in $S$ has both endpoints
in $C$. Further, let $\Delta_{1}\defeq\sum_{i\in[k]}h^{(i)}\mw_{i}^{1/2}\mproj_{\mw_{i}}\vv^{(i)}$
and $\Delta_{2}\defeq\sum_{i\in[k]}h^{(i)}\mw_{i}^{1/2}\vv^{(i)}$
so that $\sum_{i\in[k]}h^{(i)}\mw_{i}^{1/2}(\mi-\mproj_{\mw_{i}})\vv^{(i)}=\Delta_{2}-\Delta_{1}$.
Note that $\Delta_{2}$ can be computed in $O(m+zk)$ time by first
computing $\sum_{i\in[k]}h^{(i)}$ and with this computing $[\Delta_{2}]_{j}$
for $j\notin S$ in $O(1)$ time and for $j\in S$ in $O(k)$ time.
Consequently, to compute $\tilde{x}$ in the given time bound, it
suffices to approximately compute $\Delta_{1}$.

Next, let $\ml^{(i)}\defeq\mb^{\top}\mw_{i}\mb$ and $F\defeq V\setminus C$
and note that 
\begin{align*}
\Delta_{1}= & \sum_{i\in[k]}h^{(i)}\mw_{i}\mb\left[\begin{array}{cc}
\mi & -\ml_{FF}^{-1}\ml_{FC}\\
\mzero & \mi
\end{array}\right]\left[\begin{array}{cc}
\ml_{FF}^{-1} & \mzero\\
\mzero & \mSC(\ml_{i},C)^{\dagger}
\end{array}\right]\left[\begin{array}{cc}
\mi & \mzero\\
-\ml_{CF}\ml_{FF}^{-1} & \mi
\end{array}\right]\mb^{\top}\mw_{i}^{1/2}\vv^{(i)}
\end{align*}
since $\ml_{FF}=[\ml^{(i)}]_{FF}$ and $\ml_{FC}=[\ml^{(i)}]_{FC}$
for all $i$ by the definition of $C$. Further, let $\mb_{C}$ be
the incidence matrix of edges with both endpoints in $C$ and $\mb_{-C}$
be the incidence matrix of the remaining edges so that $\mb^{\top}\mw_{i}^{1/2}\vv^{(i)}=\mb_{C}^{\top}\mw_{i}^{1/2}\vv^{(i)}+\mb_{-C}^{\top}\mw^{1/2}\vv$
for all $i\in[k]$. Combining yields that, $\Delta_{1}=\va_{1}+\va_{2}+\va_{3}+\va_{4}$
where 
\begin{align}
\va_{1} & \defeq\sum_{i\in[k]}h^{(i)}\mw_{i}\mb_{C}\mSC(\ml_{i},C)^{\dagger}\mb_{C}^{\top}\mw_{i}^{1/2}\vv^{(i)}\\
\va_{2} & \defeq\sum_{i\in[k]}h^{(i)}\mw\mb_{-C}\left[\begin{array}{c}
-\ml_{FF}^{-1}\ml_{FC}\\
\mi
\end{array}\right]\mSC(\ml_{i},C)^{-1}\mb_{C}^{\top}\mw_{i}^{1/2}\vv^{(i)}\nonumber \\
\va_{3} & \defeq\sum_{i\in[k]}h^{(i)}\mw_{i}\mb_{C}\mSC(\ml_{i},C)^{\dagger}\left[\begin{array}{cc}
-\ml_{CF}\ml_{FF}^{-1} & \mi\end{array}\right]\mb_{-C}^{\top}\mw^{1/2}\vv\nonumber \\
\va_{4} & \defeq\sum_{i\in[k]}h^{(i)}\mw\mb_{-C}\left[\begin{array}{cc}
\mi & -\ml_{FF}^{-1}\ml_{FC}\\
\mzero & \mi
\end{array}\right]\left[\begin{array}{cc}
\ml_{FF}^{-1} & \mzero\\
\mzero & \mSC(\ml_{i},C)^{\dagger}
\end{array}\right]\left[\begin{array}{cc}
\mi & \mzero\\
-\ml_{CF}\ml_{FF}^{-1} & \mi
\end{array}\right]\mb_{-C}^{\top}\mw^{1/2}\vv\,.\label{eq:compute_xs_form}
\end{align}
Our algorithm simply computes $\Delta_{1}$ through the above formula
where every instance of $\mSC(\ml_{i},C)$ is replaced with
some efficiently computed $\widetilde{\msc}_{i}\approx_{\delta}\mSC(\ml_{i},C)$
for $\delta$ we set later. 

To compute the $\widetilde{\msc}_{i}$, first for each $i$, we define
$\ml_{i}(S)\defeq\mb_{S}^{\top}\mw_{i}\mb_{S}$ where $\mb_{S}$ is
the incidence matrix of edges $S$ and let $\ml_{\mathrm{ext}}\defeq\ml^{(i)}-\ml_{i}(S)$
for any $i\in[k]$. Note that this definition does not depend on $i$
by the definition of $S$. Using \cite[Theorem 1.3]{LS18}, we can compute $C\subseteq V$ such
that every edge in $S$ has both endpoints in $C$ and a Laplacian
$\widetilde{\msc}\in\R^{C\times C}$ such that $\widetilde{\msc}\approx_{\delta}\msc(\ml_{\mathrm{ext}},C)$
and $|C|\leq\nnz(\widetilde{\msc})=\tilde{O}(|S|\delta^{-2})=\tilde{O}(z\delta^{-2})$
in $\widetilde{O}(m)$ time with high probability. We use this procedure
to determine $C$ and compute $\widetilde{\msc}$. Further, we define
$\widetilde{\msc}_{i}=\widetilde{\msc}+\ml_{i}(S)$ and note that
$\widetilde{\msc}_{i}\approx_{\delta}\mSC(\ml_{i},C)$ for
all $i\in[k]$ and $\nnz(\widetilde{\msc}_{i})=\nnz(\widetilde{\msc}_{i})+|S|=\tilde{O}(z\delta^{-2})$.

Now, we let $\tilde{\va}_{1},\tilde{\va}_{2},\tilde{\va}_{3},\tilde{\va}_{4}$
be the result of computing $\va_{1},\va_{2},\va_{3},\va_{4}$ respectively
where each $\mSC(\ml_{i},C)$ is replaced with $\widetilde{\msc}_{i}$
and each matrix inversion is computed to high precision using nearly
linear time SDD-solvers for $\ml_{FF}^{-1}$ and $\mSC(\ml_{i},C)^{-1}$ (Theorem \ref{thm:lap}). Further, we let $\vx'=\vx + \sum_{i\in[4]}\tilde{\va}_{i}+\Delta_{2}$.
Note that $h^{(i)}\mb_{C}^{\top}\mw_{i}^{1/2}\vv^{(i)}$ can be computed
explicitly for all $i\in[k]$ in $\tilde{O}(m+kz\delta^{-2})$ time
by simply iterating through the changes in $\vw$ and $\vv$ and noting
that each change only effects the resulting $\tilde{O}(z\delta^{-2})$
coordinate vector in $2$ coordinates. Further, this implies that
$\vd_{i}\defeq h^{(i)}\mSC(\ml_{i},C)^{\dagger}\mb_{C}^{\top}\mw_{i}^{1/2}\vv^{(i)}$
can be computed to high precision in $\tilde{O}(kz\delta^{-2})$ time
by using a nearly linear time Laplacian system solver too apply $\mSC(\ml_{i},C)^{\dagger}$.
Next, to compute $\tilde{\va}_{1}\defeq\sum_{i\in[k]}\mw_{i}\mb_{C}d_{i}$
note that the contribution of each row of $\mb_{C}$ for $e\in S$
can be computed $O(k)$ and the contribution of all the remaining
rows can be computed in $O(m)$; thus, $\tilde{\va}_{1}$ can be computed
from the $\vd_{i}$ in $O(m+kz)$. Further, given the $\vd_{i}$ by
using a nearly linear time SDD solver to apply $\ml_{FF}^{-1}$ to
a vector we see that $\tilde{\va}_{2}$ can be computed in $\tilde{O}(m)$.
Similarly, all the $\ve_{i}\defeq\mSC(\ml_{i},C)^{\dagger}\left[\begin{array}{cc}
-\ml_{CF}\ml_{FF}^{-1} & \mi\end{array}\right]\mb_{-C}^{\top}\mw^{1/2}\vv$ can be computed in $\tilde{O}(m+zk\delta^{-2})$ and from these $\tilde{\va}_{3}$
can be computed in an additional $O(m+kz)$ time (analogous to computing
$\tilde{\va}_{1}$). Further, $\tilde{\va}_{4}$ can be computed $\tilde{O}(m+zk\delta^{-2})$
since summation can be moved to the $\mSC(\ml_{i},C)^{\dagger}$.
Putting these pieces together shows that $\tilde{\vx}$ can be computed
in $\tilde{O}(m+zk\delta^{-2})$.

Next, to determine what to set $\delta$ to. Note that
\begin{align*}
 & \|\vx^{*}-\widetilde{\vx}\|_{\mw_{k}^{-1}}\\
= & \norm{\sum_{i\in[k]}\mw_{i}\mb\left[\begin{array}{c}
-\ml_{FF}^{-1}\ml_{FC}\\
\mi
\end{array}\right](\mSC(\ml_{i},C)^{\dagger}-\widetilde{\msc}_{i}^{\dagger})\left[\begin{array}{c}
-\ml_{FF}^{-1}\ml_{FC}\\
\mi
\end{array}\right]^{\top}\mb^{\top}\mw_{i}^{1/2}(h^{(i)}\vv^{(i)})}_{\mw_{k}^{-1}}\\
\leq & r\sum_{i\in[k]}\norm{\mw_{i}\mb\left[\begin{array}{c}
-\ml_{FF}^{-1}\ml_{FC}\\
\mi
\end{array}\right](\mSC(\ml_{i},C)^{\dagger}-\widetilde{\msc}_{i}^{\dagger})\left[\begin{array}{c}
-\ml_{FF}^{-1}\ml_{FC}\\
\mi
\end{array}\right]^{\top}\mb^{\top}\mw_{i}^{1/2}(h^{(i)}\vv^{(i)})}_{\mw_{i}^{-1}}\\
\leq & r\sum_{i\in[k]}\norm{\mw_{i}^{-1/2}\mb\left[\begin{array}{c}
-\ml_{FF}^{-1}\ml_{FC}\\
\mi
\end{array}\right](\mSC(\ml_{i},C)^{\dagger}-\widetilde{\msc}_{i}^{\dagger})\left[\begin{array}{c}
-\ml_{FF}^{-1}\ml_{FC}\\
\mi
\end{array}\right]^{\top}\mb^{\top}\mw_{i}^{1/2}}_{2}\\
= & r\sum_{i\in[k]}\norm{\mSC(\ml_{i},C)^{1/2}(\mSC(\ml_{i},C)^{\dagger}-\widetilde{\msc}_{i}^{\dagger})\mSC(\ml_{i},C)^{1/2}}_{2}=O(rk\delta)
\end{align*}
where in the third line we used that assumption $\frac{1}{r}\leq\sqrt{(\vw_{i})_{l}/(\vw_{j})_{l}}\leq r$,
in the fourth we used that $\norm{h\vv^{(i)}}_{2}\leq1$, and in the
fifth we used that 
\[
\left[\begin{array}{c}
-\ml_{FF}^{-1}\ml_{FC}\\
\mi
\end{array}\right]^{\top}\mb^{\top}\mw_{i}\mb\left[\begin{array}{c}
-\ml_{FF}^{-1}\ml_{FC}\\
\mi
\end{array}\right]=\mSC(\ml_{i},C),
\]
and that $\norm{\mm_{1}\mm_{2}\mm_{3}}_{2}=\norm{(\mm_{1}^{\top}\mm_{1})^{1/2}\mm_{2}(\mm_{3}\mm_{3}^{\top})^{1/2}}_{2}$
for matrices $\mm_{1},\mm_{2},\mm_{3}$ of appropriate dimension and
that $\widetilde{\msc}_{i}\approx_{\delta}\mSC(\ml_{i},C)$.
Consequently, it suffices to set $\delta=\Theta(\epsilon/(rk))$ and this
gives the result for computing $\tilde{\vx}$.

Now, we show how to find a feasible $\widetilde{\vx}$ using $\vx'$. From the first part, 
we can find $\vx'$ such that $\|\vx'-\vx^{*}\|_{\mw_{k}^{-1/2}}\leq\frac{\epsilon}{2}$
in $\widetilde{O}(m+zr^{2}k^{3}/\epsilon^{2})$ time. Note that $\mb^{\top}\vx^{*}=\mb^{\top}\vx=\vd$.
However, we may not have $\mb^{\top}\vx'=\vd$. To fix this, we define
\[
\widetilde{\vx}\defeq\vx'+\mw_{k}\mb(\mb^{\top}\mw_{k}\mb)^{-1}(\vd-\mb^{\top}\vx').
\]
This can be found in an extra $\widetilde{O}(m)$ time (Theorem \ref{thm:lap}).
Furthermore, we have
\begin{align*}
\|\widetilde{\vx}-\vx'\|_{\mw_{k}^{-1}} & =\|\mw_{k}^{1/2}\mb(\mb^{\top}\mw_{k}\mb)^{-1}(\mb^{\top}\vx^{*}-\mb^{\top}\vx')\|_{2}\\
 & \leq\|\mw_{k}^{-1/2}(\vx^{*}-\vx')\|_{2}\leq\frac{\epsilon}{2}.
\end{align*}
Hence, we have $\|\widetilde{\vx}-\vx\|_{\mw_{k}^{-1}}\leq\epsilon$
and $\mb^{\top}\widetilde{\vx}=\vd$. 

The algorithm and analysis
for computing $\vs'$ is analogous with $\Delta_{2}$ set
to $0$ and the signs of the exponents of some $\mw_{i}$ flipped. The main difference is that $\vs'$ is automatically feasible and hence we simply set $\widetilde{\vs} = \vs'$. To see this, we note that the new $\Delta^{(\vs)}_{1}$ is given by

\begin{align*}
\Delta^{(\vs)}_{1}= & \sum_{i\in[k]}h^{(i)}\mb\left[\begin{array}{cc}
\mi & -\ml_{FF}^{-1}\ml_{FC}\\
\mzero & \mi
\end{array}\right]\left[\begin{array}{cc}
\ml_{FF}^{-1} & \mzero\\
\mzero & \mSC(\ml_{i},C)^{\dagger}
\end{array}\right]\left[\begin{array}{cc}
\mi & \mzero\\
-\ml_{CF}\ml_{FF}^{-1} & \mi
\end{array}\right]\mb^{\top}\mw_{i}^{1/2}\vv^{(i)}.
\end{align*}
Note that after we replacing $\mSC(\ml_{i},C)^{\dagger}$ by its approximation, the vector above is still in the image of $\mb$. Hence, $\vs'-\vs^{*}$ is in the image of $\mb$.

\end{proof}

\subsection{Robust IPM Stability Bound}

\label{sec:ipm_stability}

In this section we prove Lemma~\ref{lem:reldist} which bounds the
relative change in $\phi$ in each iteration of the robust IPM method
(See Section~\ref{sec:ipm}). We first provide helper Lemma~\ref{lem:reldist1dim}
and Lemma~\ref{lem:reldist1dim-all} and then use it to prove Lemma~\ref{lem:reldist}.
The first lemma is a statement about 1-dimensional log-barrier problems.
\begin{lem}
\label{lem:reldist1dim} Let $\vl,\vu\in\R^{m}$ with $\vl_{i}<\vu_{i}$
for all $i\in[m]$, $c\in\R$, $\vw_{\ell}^{(0)},\vw_{u}^{(0)},\vw_{\ell}^{(1)},\vw_{u}^{(1)}\in[\frac{7}{8},\frac{8}{7}]^{m}$, and for  $j\in\{0,1\}$ let
\[
x^{(j)}
\defeq
\argmin_{\max_{i\in[m]}\vl_{i}\leq x\leq\min_{i\in[m]}\vu_{i}}c\cdot x-\sum_{i\in[m]}\vw_{\ell,i}^{(j)}\log(x-\vl_{i})-\sum_{i\in[m]}\vw_{u,i}^{(j)}\log(\vu_{i}-x)
\]
Then, for $r(a)\defeq\max\{a-3,a^{-1}-3,0\}$ we have 
\[
\sum_{i\in[m]}r\left(\frac{x^{(1)}-\vl_{i}}{x^{(0)}-\vl_{i}}\right)+\sum_{i\in[m]}r\left(\frac{\vu_{i}-x^{(1)}}{\vu_{i}-x^{(0)}}\right)\leq16\left[\norm{\vw_{\ell}^{(0)}-\vw_{\ell}^{(1)}}_{2}^{2}+\norm{\vw_{u}^{(0)}-\vw_{u}^{(1)}}_{2}^{2}\right]
\]
\end{lem}

\begin{rem*}
Note that $r\left(\frac{x^{(1)}-\vl_{i}}{x^{(0)}-\vl_{i}}\right)$
large implies either $x^{(1)}$ or $x^{(0)}$ is much closer to $\vl_{i}$
compared to the another one. The inequality above shows that if the
weights do not change too much, then $x^{(1)}$ cannot be too much
closer to $\vl_{i}$ compared to $x^{(0)}.$
\end{rem*}
\begin{proof}
Without loss of generality, we assume $x^{(0)}<x^{(1)}$. By the optimality
condition of $x^{(j)}$, we have that
\[
c-\sum_{i\in[m]}\frac{\vw_{\ell,i}^{(j)}}{x^{(j)}-\vl_{i}}+\sum_{i\in[m]}\frac{\vw_{u,i}^{(j)}}{\vu_{i}-x^{(j)}}=0.
\]
Subtracting this equation for $j=0$ and $1$, we have
\begin{equation}
\sum_{i\in[m]}\left(\frac{\vw_{\ell,i}^{(0)}}{x^{(0)}-\vl_{i}}-\frac{\vw_{\ell,i}^{(1)}}{x^{(1)}-\vl_{i}}\right)=\sum_{i\in[m]}\left(\frac{\vw_{u,i}^{(0)}}{\vu_{i}-x^{(0)}}-\frac{\vw_{u,i}^{(1)}}{\vu_{i}-x^{(1)}}\right).\label{eq:k2main}
\end{equation}
We bound the left and right hand side above separately.

To lower bound the left hand side of (\ref{eq:k2main}), we let $\valpha_{i}=\frac{x^{(1)}-x^{(0)}}{x^{(1)}-\vl_{i}}\in(0,1)$
(since $x^{(1)}>x^{(0)}>\vl_{i}$). Note that 
\[
\frac{\vw_{\ell,i}^{(0)}}{x^{(0)}-\vl_{i}}-\frac{\vw_{\ell,i}^{(1)}}{x^{(1)}-\vl_{i}}=\frac{\valpha_{i}}{x^{(1)}-x^{(0)}}\left(\frac{\vw_{\ell,i}^{(0)}}{1-\alpha_{i}}-\vw_{\ell,i}^{(1)}\right).
\]
If $\valpha_{i}\leq2(\vw_{\ell,i}^{(1)}-\vw_{\ell,i}^{(0)})$ then
\[
\frac{\vw_{\ell,i}^{(0)}}{x^{(0)}-\vl_{i}}-\frac{\vw_{\ell,i}^{(1)}}{x^{(1)}-\vl_{i}}\geq\frac{\valpha_{i}(\vw_{\ell,i}^{(0)}-\vw_{\ell,i}^{(1)})}{x^{(1)}-x^{(0)}}\geq-\frac{2(\vw_{\ell,i}^{(0)}-\vw_{\ell,i}^{(1)})^{2}}{x^{(1)}-x^{(0)}}
\]
where we used $\vw_{\ell,i}^{(0)}-\vw_{\ell,i}^{(1)}\leq0$ in the
last inequality. Otherwise, we have $\vw_{\ell,i}^{(0)}-\vw_{\ell,i}^{(1)}>-\valpha_{i}/2$
and hence
\[
\frac{\vw_{\ell,i}^{(0)}}{1-\valpha_{i}}-\vw_{\ell,i}^{(1)}=\frac{\vw_{\ell,i}^{(0)}-\vw_{\ell,i}^{(1)}+\valpha_{i}\vw_{\ell,i}^{(1)}}{1-\valpha_{i}}\geq\frac{\valpha_{i}(\vw_{\ell,i}^{(1)}-1/2)}{1-\valpha_{i}}\geq\frac{1}{3}\frac{\valpha_{i}}{1-\valpha_{i}}\geq0
\]
where we used that $\vw_{\ell,i}^{(1)}\geq\frac{7}{8}$ and $\frac{7}{8}-\frac{1}{2}\ge\frac{1}{3}$.
Combining both cases, we have
\begin{align}
 & (x^{(1)}-x^{(0)})\cdot\sum_{i\in[m]}\left(\frac{\vw_{\ell,i}^{(0)}}{x^{(0)}-\vl_{i}}-\frac{\vw_{\ell,i}^{(1)}}{x^{(1)}-\vl_{i}}\right)\nonumber \\
\geq & -2\sum_{\valpha_{i}\leq2(\vw_{\ell,i}^{(1)}-\vw_{\ell,i}^{(0)})}(\vw_{\ell,i}^{(0)}-\vw_{\ell,i}^{(1)})^{2}+\sum_{\valpha_{i}>2(\vw_{\ell,i}^{(1)}-\vw_{\ell,i}^{(0)})}\frac{1}{3}\frac{\valpha_{i}^{2}}{1-\valpha_{i}}\nonumber \\
\geq & -2\sum_{i\in[m]}(\vw_{\ell,i}^{(0)}-\vw_{\ell,i}^{(1)})^{2}+\frac{1}{8}\sum_{\valpha_{i}\geq\frac{2}{3}}\frac{1}{1-\valpha_{i}}\label{eq:k2lhs}
\end{align}
where we used that $\valpha_{i}\geq\frac{2}{3}$ implies $\valpha_{i}>2(\vw_{\ell,i}^{(1)}-\vw_{\ell,i}^{(0)})$
at the end.

To upper bound the right hand side of (\ref{eq:k2main}), we let $\vbeta_{i}=\frac{x^{(1)}-x^{(0)}}{\vu_{i}-x^{(0)}}\in(0,1)$.
Note that 
\[
\frac{\vw_{u,i}^{(0)}}{\vu_{i}-x^{(0)}}-\frac{\vw_{u,i}^{(1)}}{\vu_{i}-x^{(1)}}=\frac{\vbeta_{i}}{x^{(1)}-x^{(0)}}\left(\vw_{u,i}^{(0)}-\frac{\vw_{u,i}^{(1)}}{1-\vbeta_{i}}\right).
\]
If $\vbeta_{i}\leq2(\vw_{u,i}^{(0)}-\vw_{u,i}^{(1)})$, we have that
\[
\frac{\vw_{u,i}^{(0)}}{\vu_{i}-x^{(0)}}-\frac{\vw_{u,i}^{(1)}}{\vu_{i}-x^{(1)}}\leq\frac{\vbeta_{i}(\vw_{u,i}^{(0)}-\vw_{u,i}^{(1)})}{x^{(1)}-x^{(0)}}\leq\frac{2(\vw_{u,i}^{(0)}-\vw_{u,i}^{(1)})^{2}}{x^{(1)}-x^{(0)}}
\]
where we used $\vw_{u,i}^{(0)}-\vw_{u,i}^{(1)}\geq0$ at the last
inequality. Otherwise, we have $\vw_{u,i}^{(0)}-\vw_{u,i}^{(1)}\leq\vbeta_{i}/2$
and hence
\[
\vw_{u,i}^{(0)}-\frac{\vw_{u,i}^{(1)}}{1-\vbeta_{i}}=\frac{\vw_{u,i}^{(0)}-\vw_{u,i}^{(1)}-\vbeta_{i}\vw_{u,i}^{(0)}}{1-\vbeta_{i}}\leq\frac{\vbeta_{i}(\frac{1}{2}-\vw_{u,i}^{(0)})}{1-\vbeta_{i}}\leq-\frac{1}{3}\frac{\vbeta_{i}}{1-\vbeta_{i}}\leq0.
\]
Combining both cases and using $\vbeta_{i}\geq\frac{2}{3}$ implies
$\vbeta_{i}>2(\vw_{u,i}^{(0)}-\vw_{u,i}^{(1)})$, we have
\begin{equation}
(x^{(1)}-x^{(0)})\cdot\sum_{i\in[m]}\left(\frac{\vw_{u,i}^{(0)}}{\vu_{i}-x^{(0)}}-\frac{\vw_{u,i}^{(1)}}{\vu_{i}-x^{(1)}}\right)\leq2\sum_{i\in[m]}(\vw_{u,i}^{(0)}-\vw_{u,i}^{(1)})^{2}-\frac{1}{8}\sum_{\vbeta_{i}\geq\frac{2}{3}}\frac{1}{1-\vbeta_{i}}\label{eq:k2rhs}
\end{equation}
Combining (\ref{eq:k2lhs}) and (\ref{eq:k2rhs}) with (\ref{eq:k2main}),
we have
\[
-2\sum_{i\in[m]}(\vw_{\ell,i}^{(0)}-\vw_{\ell,i}^{(1)})^{2}+\frac{1}{8}\sum_{\valpha_{i}\geq\frac{2}{3}}\frac{1}{1-\valpha_{i}}\leq2\sum_{i\in[m]}(\vw_{u,i}^{(0)}-\vw_{u,i}^{(1)})^{2}-\frac{1}{8}\sum_{\vbeta_{i}\geq\frac{2}{3}}\frac{1}{1-\vbeta_{i}}.
\]
Using this, $x^{(0)}\leq x^{(1)}$ and the formula of $r$, we have
\begin{align*}
\sum_{i\in[m]}r\left(\frac{x^{(1)}-\vl_{i}}{x^{(0)}-\vl_{i}}\right)+r\left(\frac{\vu_{i}-x^{(1)}}{\vu_{i}-x^{(0)}}\right)\leq & \sum_{i\in[m]}\max\left\{ \frac{x^{(1)}-\vl_{i}}{x^{(0)}-\vl_{i}}-3,0\right\} +\sum_{i\in[m]}\max\left\{ \frac{\vu_{i}-x^{(0)}}{\vu_{i}-x^{(1)}}-3,0\right\} \\
= & \sum_{i\in[m]}\max\left\{ \frac{1}{1-\valpha_{i}}-3,0\right\} +\sum_{i\in[m]}\max\left\{ \frac{1}{1-\vbeta_{i}}-3,0\right\} \\
\leq & \sum_{\valpha_{i}\geq\frac{2}{3}}\frac{1}{1-\valpha_{i}}+\sum_{\vbeta_{i}\geq\frac{2}{3}}\frac{1}{1-\vbeta_{i}} \\
\leq &16\left[\norm{\vw_{\ell}^{(0)}-\vw_{\ell}^{(1)}}_{2}^{2}+\norm{\vw_{u}^{(0)}-\vw_{u}^{(1)}}_{2}^{2}\right]\,.
\end{align*}
\end{proof}
We leverage this lemma to generalize to higher dimensions in the following
lemma.
\begin{lem}
\label{lem:reldist1dim-all} In the setting of \eqref{eq:lp_formula}, given
weights $\vw_{\ell}^{(0)},\vw_{u}^{(0)},\vw_{\ell}^{(1)},\vw_{u}^{(1)}\in[\frac{7}{8},\frac{8}{7}]^{m}$ let
\[
\vx_{j}\defeq\argmin_{\vx\in\X\,|\,\mb^{\top}\vx=\vd}\vc^{\top}\vx-\sum_{i\in[m]}\vw_{\ell,i}^{(j)}\log(\vx-\vl_{i})-\sum_{i\in[m]}\vw_{u,i}^{(j)}\log(\vu_{i}-\vx)
\]
for $j\in\{0,1\}$. Then for $r(a)\defeq\max\{a-3,a^{-1}-3,0\}$ we have 
\[
\sum_{i\in[m]}r\left(\frac{\vx_{i}^{(1)}-\vl_{i}}{\vx_{i}^{(0)}-\vl_{i}}\right)+\sum_{i\in[m]}r\left(\frac{\vu_{i}-\vx^{(1)}}{\vu_{i}-\vx^{(0)}}\right)\leq16\left[\norm{\vw_{\ell}^{(0)}-\vw_{\ell}^{(1)}}_{2}^{2}+\norm{\vw_{u}^{(0)}-\vw_{u}^{(1)}}_{2}^{2}\right]\,.
\]
\end{lem}

\begin{proof}
For all $t\in\R$ let $\vx^{(t)}\defeq\vx^{(0)}+t\vdelta_{x}$ with
$\vdelta_{x}=\vx^{(1)}-\vx^{(0)}$. By the definition of $\vx^{(j)}$,
we have that for $j\in\{0,1\}$
\begin{align}
\vx^{(j)} & =\argmin_{t\in\R|\vx^{(t)}\in\X}\vc^{\top}\vx^{(t)}-\sum_{i\in[m]}\vw_{\ell,i}^{(j)}\log(\vx_{i}^{(t)}-\vl_{i})-\sum_{i\in[m]}\vw_{u,i}^{(j)}\log(\vu_{i}-\vx_{i}^{(t)})\nonumber \\
 & =\argmin_{t\in\R|\vx^{(t)}\in\X}t\cdot\vc^{\top}\vdelta_{x}-\sum_{i\in[m]}\vw_{\ell,i}^{(j)}\log(t\vdelta_{x,i}-(\vl_{i}-\vx_{i}^{(0)}))-\sum_{i\in[m]}\vw_{u,i}^{(j)}\log((\vu_{i}-\vx_{i}^{(0)})-t\vdelta_{x,i})\nonumber \\
 & =\argmin_{t\in\R|\vx^{(t)}\in\X}t\cdot\tilde{c}-\sum_{i\in[m]\,:\,\vx_{i}^{(1)}\neq\vx_{i}^{(0)}}\tilde{\vw}_{\ell,i}^{(j)}\log(t-\tilde{\vl}_{i})-\sum_{i\in[m]}\tilde{\vw}_{u,i}^{(j)}\log(\tilde{\vu}_{i}-t)\label{eq:lem_new_min}
\end{align}
where $\tilde{c}\defeq\vc^{\top}\vdelta_{x}$ and
\[
(\tilde{\vl}_{i},\tilde{\vu}_{i},\tilde{\vw}_{\ell,i}^{(j)},\tilde{\vw}_{u,i}^{(j)})\defeq\begin{cases}
(\frac{\vl_{i}-\vx_{i}^{(0)}}{\vx_{i}^{(1)}-\vx_{i}^{(0)}},\frac{\vu_{i}-\vx_{i}^{(0)}}{\vx_{i}^{(1)}-\vx_{i}^{(0)}},\vw_{\ell,i}^{(j)},\vw_{u,i}^{(j)}) & \text{if }\vdelta_{x,i}>0\\
(\frac{\vu_{i}-\vx_{i}^{(0)}}{\vx_{i}^{(0)}-\vx_{i}^{(1)}},\frac{\vl_{i}-\vx_{i}^{(0)}}{\vx_{i}^{(0)}-\vx_{i}^{(1)}},\vw_{u,i}^{(j)},\vw_{\ell,i}^{(j)}) & \text{if }\vdelta_{x,i}<0
\end{cases}.
\]
Applying Lemma~\ref{lem:reldist1dim} to (\ref{eq:lem_new_min})
then yields the result as for all $i\in[m]$ with $\vdelta_{x,i}>0$,
we have that 
\[
\frac{1-\tilde{\vl}_{i}}{0-\tilde{\vl}_{i}}=\frac{1-\frac{\vl_{i}-\vx_{i}^{(0)}}{\vx_{i}^{(1)}-\vx_{i}^{(0)}}}{\frac{\vx_{i}^{(0)}-\vl_{i}}{\vx_{i}^{(1)}-\vx_{i}^{(0)}}}=\frac{\vx_{i}^{(1)}-\vl_{i}}{\vx_{i}^{(0)}-\vl_{i}}\text{ and }\frac{1-\tilde{\vu}_{i}}{0-\tilde{\vu}_{i}}=\frac{1-\frac{\vu_{i}-\vx_{i}^{(0)}}{\vx_{i}^{(1)}-\vx_{i}^{(0)}}}{\frac{\vx_{i}^{(0)}-\vu_{i}}{\vx_{i}^{(1)}-\vx_{i}^{(0)}}}=\frac{\vu_{i}-\vx_{i}^{(1)}}{\vu_{i}-\vx_{i}^{(0)}}
\]
and similarly when $\vdelta_{x,i}<0$, we have
\[
\frac{1-\tilde{\vl}_{i}}{0-\tilde{\vl}_{i}}=\frac{1-\frac{\vu_{i}-\vx_{i}^{(0)}}{\vx_{i}^{(0)}-\vx_{i}^{(1)}}}{\frac{\vx_{i}^{(0)}-\vu_{i}}{\vx_{i}^{(0)}-\vx_{i}^{(1)}}}=\frac{\vu_{i}-\vx_{i}^{(1)}}{\vu_{i}-\vx_{i}^{(0)}}\text{ and }\frac{1-\tilde{\vu}_{i}}{0-\tilde{\vu}_{i}}=\frac{1-\frac{\vl_{i}-\vx_{i}^{(0)}}{\vx_{i}^{(0)}-\vx_{i}^{(1)}}}{\frac{\vx_{i}^{(0)}-\vl_{i}}{\vx_{i}^{(0)}-\vx_{i}^{(1)}}}=\frac{\vx_{i}^{(1)}-\vl_{i}}{\vx_{i}^{(0)}-\vl_{i}}\,.
\]
\end{proof}
Leveraging Lemma~\ref{lem:reldist1dim-all}, we prove the main result
of this section, Lemma~\ref{lem:reldist}.
\begin{proof}[Proof of Lemma~\ref{lem:reldist}]
To apply Lemma~\ref{lem:reldist1dim-all} for $j\in\{0,1\}$ we
define $\overline{\mu}\defeq\frac{1}{2}[\mu^{(0)}+\mu^{(1)}]$, $\overline{\vy}^{(j)}=\frac{\mu^{(j)}}{\overline{\mu}}\vy^{(j)}$,
$\overline{\veta}^{(j)}=\frac{\mu^{(j)}}{\overline{\mu}}\veta^{(j)}+\overline{\vc}-\vc/\overline{\mu}$
and
\[
\overline{\vc}_{i}=\begin{cases}
\vc_{i}/\overline{\mu}-\frac{\mu^{(0)}}{\overline{\mu}}\veta_{i}^{(0)} & \text{if }\phi_{i}''(\vx^{(0)})<\phi_{i}''(\vx^{(1)})\\
\vc_{i}/\overline{\mu}-\frac{\mu^{(1)}}{\overline{\mu}}\veta_{i}^{(1)} & \text{else}
\end{cases}.
\]
With these definitions, we note that $\mb\overline{\vy}^{(j)}+\overline{\vc}+\nabla\phi(\vx^{(j)})=\overline{\veta}^{(j)}.$
Since Lemma~\ref{lem:reldist1dim} considers only exact minimizers
of weighted log barriers, we remove the term $\overline{\beta}^{(j)}$
using the weights as follows, we define $\phi_{\ell,i}(\vx)\defeq-\sum_{i\in[m]}\log(\vx_{i}-\vl_{i})$,
$\phi_{u,i}(\vx)\defeq-\sum_{i\in[m]}\log(\vu_{i}-\vx_{i})$, $\phi_{\ell}(\vx)=\sum_{i\in[m]}\phi_{\ell,i}(\vx)$,
and $\phi_{u}(\vx)=\sum_{i\in[m]}\phi_{u,i}(\vx)$. Then, we define
the weights
\[
\valpha_{i}^{(j)}=1-\frac{\mathrm{sign}([\grad\phi{}_{\ell,i}(\vx^{(j)})]_{i})}{|[\grad\phi{}_{\ell,i}(\vx^{(j)})]_{i}|+|[\grad\phi{}_{u,i}(\vx^{(j)})]_{i}|}\overline{\veta}_{i}^{(j)}\text{ and }\vbeta_{i}^{(j)}=1-\frac{\mathrm{sign}([\grad\phi{}_{u,i}(\vx^{(j)})]_{i})}{|[\grad\phi{}_{\ell,i}(\vx^{(j)})]_{i}|+|[\grad\phi{}_{u,i}(\vx^{(j)})]_{i}|}\overline{\veta}_{i}^{(j)}
\]
so that $\sum_{i\in[m]}(\valpha_{i}^{(j)}\nabla\phi_{\ell,i}(\vx^{(j)})+\vbeta_{i}^{(j)}\nabla\phi_{u,i}(\vx^{(j)}))=\nabla\phi(\vx^{(j)})-\bar{\veta}$
and 
\[
\mb\overline{\vy}^{(j)}+\overline{\vc}+\sum_{i\in[m]}(\valpha_{i}^{(j)}\nabla\phi_{l,i}(\vx^{(j)})+\vbeta_{i}^{(j)}\nabla\phi_{u,i}(\vx^{(j)}))=0.
\]
Consequently,
\[
\vx^{(j)}=\argmin_{\vx\in\X\,|\,\mb^{\top}\vx=\vd}\overline{\vc}^{\top}\vx+\sum_{i\in[m]}\valpha_{i}^{(j)}\phi_{\ell,i}(\vx)+\sum_{i\in[m]}\vbeta_{i}^{(j)}\phi_{u,i}(\vx).
\]
Hence, we can apply Lemma \ref{lem:reldist1dim} with $\vw_{\ell,i}^{(j)}=\valpha_{i}^{(j)}$
and $\vw_{u,i}^{(j)}=\vbeta_{i}^{(j)}$ provided $\frac{7}{8}\leq\vw_{\ell,i}^{(j)}\leq\frac{8}{7}$
and $\frac{7}{8}\leq\vw_{u,i}^{(j)}\leq\frac{8}{7}$ for all $i\in[m]$
and $j\in\{1,2\}$. To show this, note that
\begin{equation}
\left|\valpha_{i}^{(j)}-1\right|\leq\frac{|\overline{\veta}_{i}^{(j)}|}{|\phi'_{l,i}(\vx^{(j)})|+|\phi'_{u,i}(\vx^{(j)})|}=\frac{|\overline{\veta}_{i}^{(j)}|}{\sqrt{\phi''_{u,i}(\vx^{(j)})}+\sqrt{\phi''_{u,i}(\vx^{(j)})}}\leq\frac{|\overline{\veta}_{i}^{(j)}|}{\sqrt{\phi''_{i}(\vx^{(j)})}}\label{eq:k2lemalphabound}
\end{equation}
where we used the definition of $\phi_{u}$ and $\phi_{l}$ in the
equality. Now, using the definition of $\overline{\veta}_{i}^{(j)}$,
we have
\begin{align*}
|\overline{\veta}_{i}^{(j)}| & =\begin{cases}
\overline{\mu}^{-1}\left|\veta_{i}^{(j)}\mu^{(j)}-\veta_{i}^{(0)}\mu^{(0)}\right| & \text{if }\phi_{i}''(\vx^{(0)})<\phi_{i}''(\vx^{(1)})\\
\overline{\mu}^{-1}\left|\veta_{i}^{(j)}\mu^{(j)}-\veta_{i}^{(1)}\mu^{(1)}\right| & \text{otherwise}
\end{cases}\\
 & =\begin{cases}
0 & \text{if }\phi_{i}''(\vx^{(j)})<\phi_{i}''(\vx^{(1-j)})\\
\overline{\mu}^{-1}\left|\veta_{i}^{(1)}\mu^{(1)}-\veta_{i}^{(0)}\mu^{(0)}\right| & \text{otherwise}
\end{cases}.
\end{align*}
Hence, we have
\begin{equation}
\frac{|\overline{\veta}_{i}^{(j)}|}{\sqrt{\phi''_{i}(\vx^{(j)})}}\leq\frac{\left|\veta_{i}^{(1)}\mu^{(1)}-\veta_{i}^{(0)}\mu^{(0)}\right|/\overline{\mu}}{\max(\sqrt{\phi''_{i}(\vx^{(0)})},\sqrt{\phi''_{i}(\vx^{(1)})})}.\label{eq:k2lemetabound}
\end{equation}
Using that $\mu^{(0)}\approx_{1/32}\mu^{(1)}$ and $\left\Vert \eta^{(j)}/\sqrt{\phi''(\vx^{(j)})}\right\Vert _{\infty}\leq\frac{1}{32}$,
we have that $\frac{|\overline{\veta}_{i}^{(j)}|}{\sqrt{\phi''_{i}(\vx^{(j)})}}\leq\frac{1}{8}$.
Hence, (\ref{eq:k2lemalphabound}) shows that $|\vw_{\ell,i}^{(j)}-1|\leq1/8$
for all $i\leq m$ and $j$. The same proof gives the bound of $\vbeta_{i}^{(j)}$. 

Consequently, Lemma~\ref{lem:reldist1dim} shows that
\[
\sum_{i\in[m]}r\left(\frac{\vx_{i}^{(1)}-\vl_{i}}{\vx_{i}^{(0)}-\vl_{i}}\right)+\sum_{i\in[m]}r\left(\frac{\vu_{i}-\vx_{i}^{(1)}}{\vu_{i}-\vx_{i}^{(0)}}\right)\leq16\left[\norm{\vw_{\ell}^{(0)}-\vw_{\ell}^{(1)}}_{2}^{2}+\norm{\vw_{u}^{(0)}-\vw_{u}^{(1)}}_{2}^{2}\right].
\]
Using (\ref{eq:k2lemalphabound}) and (\ref{eq:k2lemetabound}), we
have that
\begin{align}
\norm{\vw_{\ell}^{(0)}-\vw_{\ell}^{(1)}}_{2}^{2}+\norm{\vw_{u}^{(0)}-\vw_{u}^{(1)}}_{2}^{2} & \leq8\sum_{i\in[m]}\frac{\overline{\mu}^{-2}(\veta_{i}^{(1)}\mu^{(1)}-\veta_{i}^{(0)}\mu^{(0)})^{2}}{\max(\phi''_{i}(\vx^{(0)}),\phi''_{i}(\vx^{(1)}))}\nonumber \\
 & \leq32\overline{\mu}^{-2}\sum_{i\in[m]}\frac{(\veta_{i}^{(1)}\mu^{(1)}-\veta_{i}^{(0)}\mu^{(1)})^{2}+(\veta_{i}^{(0)}\mu^{(1)}-\veta_{i}^{(0)}\mu^{(0)})^{2}}{\phi''_{i}(\vx^{(0)})+\phi''_{i}(\vx^{(1)})}\nonumber \\
 & \leq64\sum_{i\in[m]}\frac{(\veta_{i}^{(1)}-\veta_{i}^{(0)})^{2}}{\phi''_{i}(\vx^{(0)})+\phi''_{i}(\vx^{(1)})}+\frac{1}{2}\overline{\mu}^{-2}m(\mu^{(1)}-\mu^{(0)})^{2}\nonumber \\
 & \leq64\sum_{i\in[m]}\frac{(\veta_{i}^{(1)}-\veta_{i}^{(0)})^{2}}{\phi''_{i}(\vx^{(0)})+\phi''_{i}(\vx^{(1)})}+m\left(\frac{\mu^{(1)}-\mu^{(0)}}{\mu^{(0)}}\right)^{2}.\label{eq:reldist1}
\end{align}
Finally, we note that $\phi_{i}''(\vx^{(0)})\leq\max(\frac{\vx^{(1)}-\vl_{i}}{\vx^{(0)}-\vl_{i}},\frac{\vu_{i}-\vx^{(1)}}{\vu_{i}-\vx^{(0)}})^{2}\cdot\phi_{i}''(\vx^{(1)})$.
Hence, we have
\begin{align}
\sum_{\phi_{i}''(\vx^{(0)})^{1/2}\geq3\phi_{i}''(\vx^{(1)})^{1/2}}\sqrt{\frac{\phi_{i}''(\vx^{(0)})}{\phi_{i}''(\vx^{(1)})}} & \leq\sum_{\max(\frac{\vx^{(1)}-\vl_{i}}{\vx^{(0)}-\vl_{i}},\frac{\vu_{i}-\vx^{(1)}}{\vu_{i}-\vx^{(0)}})\geq3}\max\left\{ \frac{\vx^{(1)}-\vl_{i}}{\vx^{(0)}-\vl_{i}},\frac{\vu_{i}-\vx^{(1)}}{\vu_{i}-\vx^{(0)}}\right\} \nonumber \\
 & \leq\sum_{i\in[m]}r\left(\frac{\vx^{(1)}-\vl_{i}}{\vx^{(0)}-\vl_{i}}\right)+\sum_{i\in[m]}r\left(\frac{\vu_{i}-\vx^{(1)}}{\vu_{i}-\vx^{(0)}}\right)\nonumber \\
 & \leq16\left[\norm{\vw_{\ell}^{(0)}-\vw_{\ell}^{(1)}}_{2}^{2}+\norm{\vw_{u}^{(0)}-\vw_{u}^{(1)}}_{2}^{2}\right].\label{eq:reldist2}
\end{align}
The result then follows from (\ref{eq:reldist1}) and (\ref{eq:reldist2}).
\end{proof}

\section{Final Runtime Bound}
\label{sec:combine}

In this section we show \cref{sec:graph_solution_maintenance} which describes how efficiently the data structures we developed in \cref{sec:dynamicsc,sec:er,sec:adaptive} can implement an IPM step. Our final runtime is then achieved via \cref{lem:ipm_imp}. Finally, we cite previous work to explain how to get an initial point for the IPM, and how to get a mincost flow after running $\O(\sqrt{m})$ IPM iterations.

\subsection{Efficient Solution Maintenance}
\label{subsec:efficientsoln}

\restatesol*

We first make several useful definitions. We let $\xx^{(j)}, \ss^{(j)}$ be the true iterates after the $j$-th call to \textsc{Move}. Our algorithm will explicitly approximate iterates $\hat{\xx}^{(j)}, \hat{\ss}^{(j)}$. Using these approximate iterates, the algorithm will output $\bar{\xx}^{(j)}, \bar{\ss}^{(j)}$ satisfying the desired update schedule, i.e. at most $\O(2^{2\ell_j}\alpha^{-2})$ coordinates are updated after the $j$-th call to \textsc{Move}. Additionally, call an index $j$ corresponding to the $j$-th \textsc{Move} operation \emph{special} if it occurs immediately following a \textsc{StartPhase} operation.  The formal construction of $\hat{\xx}^{(j)}$ and $\hat{\ss}^{(j)}$ is given in the following definition.

\begin{definition}[Approximate iterates]
\label{def:hatxx}
Say there have been $j$ \textsc{Move} operations so far. If the next operation is $\textsc{StartPhase}(\tilde{\xx}, \tilde{\ss})$, then set $\hat{\xx}^{(j+1)} \assign \tilde{\xx}$ and $\hat{\ss}^{(j+1)} \assign \tilde{\ss}$.
If the next operation is \textsc{Move} operation $(j+1)$, then let $\bdelta_{\xx}, \bdelta_{\ss}$ satisfy
\begin{align}
&\left\|\bdelta_{\xx} - h^{(j+1)}(\mi-\mP_{j+1})\vv^{(j+1)}\right\|_\infty \le \epsilon \enspace \text{ and } \enspace \left\|\bdelta_{\ss} - h^{(j+1)}\mP_{j+1}\vv^{(j+1)} \right\|_\infty \le \epsilon, \label{eq:defdelta}
\end{align}
with $\bdelta_{\xx}, \bdelta_{\ss}$ supported on $O(\eps^{-2})$ coordinates. If the previous operation was \text{Move}, define
and let $\hat{\xx}^{(j+1)} \assign \hat{\xx}^{(j)} + \mW_{j+1}^{1/2}\bdelta_{\xx}$ and $\hat{\ss}^{(j+1)} \assign \hat{\ss}^{(j)} + \mW_{j+1}^{-1/2}\bdelta_{\ss}$.
Otherwise, if the previous operation was \textsc{StartPhase} define $\hat{\xx}^{(j+1)} \assign \hat{\xx}^{(j+1)} + \mW_{j+1}^{1/2}\bdelta_{\xx}$ and $\hat{\ss}^{(j+1)} \assign \hat{\ss}^{(j+1)} + \mW_{j+1}^{-1/2}\bdelta_{\ss}$ (so that we redefine $\hat{\xx}^{(j+1)}, \hat{\ss}^{(j+1)}$).
\end{definition}
We show that the $\hat{\xx}^{(j)}$ are slowly changing, except potentially at special indices. This is because $\left\|\bdelta_{\xx}\right\|_2 = O(1)$ as it is supported on $\O(\eps^{-2})$ nonzeros and $h^{(j+1)}\|\vv^{(j+1)}\|_2 \le 1.$

We now argue that $\bdelta_{\xx}$ and $\bdelta_{\ss}$ can be computed efficiently.
\begin{lemma}[Computation of $\bdelta_{\xx}, \bdelta_{\ss}$]
\label{lemma:costdelta}
In the context of \cref{sec:graph_solution_maintenance}, there is an operation that computes $\bdelta_{\xx}, \bdelta_{\ss}$ satisfying \eqref{eq:defdelta}
in average amortized time $\O(m^{15/16}\eps^{-7/8})$ and succeeds with high probability against adaptive adversaries.
\end{lemma}
\begin{proof}
We first write $h^{(j+1)}(\mi-\mP_{j+1})\vv^{(j+1)} = h^{(j+1)}\vv^{(j+1)} - h^{(j+1)}\mP_{j+1}\vv^{(j+1)}$, and handle both parts separately up to error $\epsilon/2$. The first part can be trivially handled, as it can be explicitly maintained in time proportional to the number of changes in $\vv^{(j)}$, and $\|h^{(j+1)}\vv^{(j+1)}\|_2 \le 1$. For the second part, we first call the dynamic \textsc{Locator}~(\cref{thm:locator}) to get a set $S$ of size $O(\eps^{-2})$. Then we call the dynamic \textsc{Evaluator}s~(\cref{thm:evaluator}) wrapped inside \cref{thm:dp} with $\eps \assign \eps/(C\log^2 n)$ on $S$ by calling $\textsc{Query}()$ on $S$. The algorithm for $\bdelta_{\ss}$ follows exactly as the second term. Also, $\bdelta_{\xx}, \bdelta_{\ss}$ are supported on $|S| = O(\eps^{-2})$ coordinates by \cref{thm:dp}.

Correctness follows directly from the guarantees of \cref{thm:locator,thm:evaluator,thm:dp}. It suffices to analyze the amortized runtime. We focus on the cost of applying \cref{thm:evaluator} inside \cref{thm:dp}, as the cost of \cref{thm:locator} is less. Let $\delta_i \defeq 2^{-i}$ so that the $i$-th \textsc{Evaluator} is run with accuracy $\eps_i \defeq \delta_i\eps$ in \cref{thm:dp}. Let $\beta_i$ be the terminal size parameter for the $i$-th \textsc{Evaluator}.

There are two possible ways to run the $i$-th \textsc{Evaluator}. Either it pays $\O(m)$ time per call to solve a Laplacian exactly (while this algorithm is randomized, we can hide randomness by adding polynomially small noise that is larger than the error we solve the Laplacian to \cite{LS15}) or applies \cref{thm:evaluator}. Let us calculate the runtime of the latter approach. After $\beta_i m$ edge updates or marking, the data structure must re-initialize. Thus, after $T$ \textsc{Move} updates, because $C_z = \O(1)$, there are at most $\O(T^2 + T\eps^{-2})$ total edges we have queried or updated: $\O(T^2)$ from updates, and $\O(T\eps^{-2})$ from the set $S$ returned by \textsc{Locator}. We assume for now that the $\O(T^2)$ term dominates -- thus the data structure must reinitialize every $\sqrt{\beta_i m}$ iterations, where each initialization costs $\O(m\beta_i^{-2}\eps_i^{-2})$ time. Thus the amortized reinitialization time per \text{Move} is \[ \O(m\beta_i^{-2}\eps_i^{-2}/\sqrt{\beta_i m}) = \O(\sqrt{m}\beta_i^{-5/2}\eps^{-2}\delta_i^{-2}). \]
By \cref{thm:dp}, the $i$-th \textsc{Evaluator} is queried with probability $O(\delta_i)$, hence the expected query time is $\O(\beta_i m\eps_i^{-2}\delta_i) = \O(\beta_im\eps^{-2}\delta_i^{-1})$ by \cref{thm:evaluator} $\textsc{Query}()$, or $\O(\delta_i m)$ if \textsc{Evaluator} simply solves a Laplacian every iteration. Thus, the amortized runtime for the $i$-th \textsc{Evaluator} is
\[ \O\left(\min\left\{\delta_i m, \beta_im\eps^{-2}\delta_i^{-1} + \sqrt{m}\beta_i^{-5/2}\eps^{-2}\delta_i^{-2} \right\}\right). \]
For the choice $\beta_i = m^{-1/7}\delta_i^{-2/7}$, this becomes
\[ \O\left(\min\left\{\delta_i m, m^{6/7}\eps^{-2}\delta_i^{-9/7} \right\}\right). \]
This is maximized when the two expressions are equal at $\delta_i = m^{-1/16}\eps^{-7/8}$, yielding a runtime of $\O(m^{15/16}\eps^{-7/8})$ as desired.
Finally, note that this means that $\eps \ge m^{-1/14}$ or the previous runtime is trivial. All $\beta_i \ge m^{-1/7}$, so $T\eps^{-2} \le T^2$ for the choice $T = \sqrt{\beta_i m} \ge m^{3/7} > \eps^{-2}$, so the $\O(T^2)$ term dominated earlier, as desired.
\end{proof}

We now show that $\hat{\xx}^{(j)}$ and $\hat{\ss}^{(j)}$ are close to $\xx^{(j)}, \ss^{(j)}$.
\begin{lemma}
\label{lemma:hatxtox}
For $\eps = \frac{\alpha}{10C_rk^3}$ and $\hat{\xx}^{(j)}, \hat{\ss}^{(j)}$ defined in \cref{def:hatxx}, $\left\|\mW_j^{-1/2}\left(\hat{\xx}^{(j)} - \xx^{(j)}\right)\right\|_\infty \le \alpha/10$ and $\left\|\mW_j^{1/2}\left(\hat{\ss}^{(j)} - \ss^{(j)}\right)\right\|_\infty \le \alpha/10$.
\end{lemma}
\begin{proof}
It suffices to analyze $j$ between \textsc{StartPhase}s, as $\hat{\xx}^{(j)}, \hat{\ss}^{(j)}$ and $\xx^{(j)}, \ss^{(j)}$ are both set to $\widetilde{\xx}, \widetilde{\ss}$ during a \textsc{StartPhase}. Over $L$ steps between \textsc{StartPhase}s (from $j_1$ to $j_2 = j_1 + L$), we have that $\left\|\mW_{j_2}^{-1/2}\left(\hat{\xx}^{(j_2)} - \xx^{(j_2)}\right)\right\|_\infty$ is at most
\begin{align*}
&\left\|\mW_{j_2}^{-1/2} \sum_{j \in [j_1, j_2)} \mW_j^{1/2} \left(\bdelta_{\xx}^{(j)} - h^{(j+1)}(\mi-\mP_{j+1})\vv^{(j+1)}\right)\right\|_\infty \\
\le~&\max_{j \in [j_1, j_2]} \sqrt{\left\|\ww_j/\ww_{j_2}\right\|_\infty} \left\| \sum_{j \in [j_1, j_2)} \bdelta_{\xx}^{(j)} - h^{(j+1)}(\mi-\mP_{j+1})\vv^{(j+1)} \right\|_\infty \\
\overset{(i)}{\le}~& C_rL^2 \eps L = \eps C_rL^3 \le \alpha/10,
\end{align*}
where $(i)$ follows from the guarantee that of \cref{def:sol_maintainer} that $\sqrt{\left\|\vw_{(j_{2})}/\vw_{j_1}\right\|_\infty} \le rL^2$ and \eqref{eq:defdelta}. The bound on the error for $\ss^{(j)}$ follows similarly.
\end{proof}

\begin{proof}[Proof of \cref{sec:graph_solution_maintenance}]
We show this by carefully defining $\bar{\xx}^{(j)}, \bar{\ss}^{(j)}$ given $\hat{\xx}^{(j)}, \hat{\ss}^{(j)}$. We mimic the approach based on binary expansions given in previous works on robust IPMs, for example \cite[Theorem 8]{BLSS20}. Precisely, we first calculate $\sum_{j'=j-2^{\ell_j}}^j \bdelta_{\xx}^{(j')}$, i.e. the sum of errors in the last $2^{\ell_j}$ steps. If this exceeds $\frac{\alpha}{100\log n}$, then we set $\bar{\xx}^{(j)} \assign \hat{\xx}^{(j)}$, otherwise we set $\bar{\xx}^{(j)} \assign \bar{\xx}^{(j-1)}$ (no change). We do the same for $\bar{\ss}^{(j)}$. Now, the bounds on number of changes follows from the bounds $\|\bdelta_{\xx}\|_2 \le O(1)$ and that $\|\mW_j^{-1/2}(\xx-\tilde{\xx})\|_2 \le 1$ in \textsc{StartPhase}. More precisely, over $2^{\ell_j}$ steps, only $O(2^{2\ell_j}\alpha^{-2}\log^2n)$ could change by $\frac{\alpha}{100\log n}$, because every step satisfies $\|\bdelta_{\xx}\|_2 \le 1$ or $\|\mW_j^{-1/2}(\xx-\tilde{\xx})\|_2 \le 1$. This completes the proof of the number of changes.

We now claim that $\left\| \mW_j^{-1/2}\left(\bar{\xx}^{(j)} - \hat{\xx}^{(j)}\right)\right\|_\infty \le \alpha/10$, so $\left\| \mW_j^{-1/2}\left(\bar{\xx}^{(j)} - \xx^{(j)}\right)\right\|_\infty \le \alpha$ by combining with \cref{lemma:hatxtox} for $\eps = \frac{\alpha}{10C_rk^3} = \tilde{\Theta}(1/k^3)$ as $C_r = \O(1)$. This claim follows from the same argument as \cite[Theorem 8]{BLSS20}: each interval $[j_1, j_2]$ can be split into $\log n$ intervals contained in intervals $[j-2^{\ell_j}, j]$ for $j \le j_2$. Each of these has at most $\alpha/(100\log n)$ error, so the total error is at most $\alpha/(100\log n) \cdot \log n \le \alpha/10$.

Finally, we must calculate the runtime of $\timePhase$. The first cost is $\O(m)$ (eg. for reading $\tilde{\xx}, \tilde{\ss}$). The second cost is calling \cref{lemma:costdelta} $k$ times (as there are at most $k$ \textsc{Move} operations between \textsc{StartPhase}). For our choice $\eps = \tilde{\Theta}(1/k^3)$, the total time for this is $\O(m^{15/16}\eps^{-7/8}k) = \O(m^{15/16}k^{29/8})$ as desired.
\end{proof}

\subsection{Initial Point, Final Point, and Proof of Main Theorem}

It is standard to get an initial $\mu$-centered feasible pair $(\xx, \ss)$ for path parameter $\mu = \mu_{\mathrm{start}} \ge (mU)^{O(1)}$. Additionally, given a $\mu$-centered feasible pair $(\xx, \ss)$ for path parameter $\mu = \mu_{\mathrm{end}} \le (mU)^{-O(1)}$ we can recover a high-accuracy mincost flow (and hence round to an exact solution).

\begin{lemma}[\!\!{\cite[Lemma 7.5, Lemma 7.8]{BLLSSSW21}}]
\label{lemma:initialfinal}
Given a graph $G = (V, E)$ and mincost flow instance with demand $\dd \in [-U, \dots, U]^V$, costs $\cc \in [-U, \dots, U]^E$, and capacities $\bell, \uu \in [-U, \dots, U]^E$, we can build a mincost flow instance on a graph $G'$ with at most $O(m)$ edges with demands, costs, and capacities bounded by $\poly(mU)$.
Additionally, we can construct a $\mu_{\mathrm{start}}$-centered pair $(\ff, \ss)$ on $G'$ for $\mu_{\mathrm{start}} = \poly(mU)$. Additionally, given a $1/\poly(mU)$-accurate mincost flow on $G'$ we can recover an exact mincost flow on $G$ in time $\O(m \log U)$.
\end{lemma}

\begin{proof}[Proof of \cref{thm:main}]
We apply \cref{lemma:initialfinal} to get an initial point for \cref{algo:rIPM}. Then, we run \cref{algo:rIPM} and round to an exact mincost flow using \cref{lemma:initialfinal}. This succeeds by \cref{lem:ipm_imp} in time 
\begin{align} \widetilde{O}\left(\left(\timeInit+\frac{\sqrt{m}}{k}(\timePhase+\timeApprox)\right)\log
	\left(\frac{\mu_{\mathrm{start}}}{\mu_{\mathrm{end}}}\right)
	\right) 
	= 
\widetilde{O}\left(\left(\timeInit+\frac{\sqrt{m}}{k}(\timePhase+\timeApprox)\right)
\log U\right). \label{eq:runtimeformula} 
\end{align}
It suffices to plug in the values of $\timeInit, \timePhase$ from \cref{sec:graph_solution_maintenance} and $\timeApprox$ from \cref{thm:graph_solution_approximation}.

We take $k = m^{1/58}$ so $\timePhase = \O(m^{15/16}k^{29/8} + m) = \O(m)$ by \cref{sec:graph_solution_maintenance}. Also by \cref{thm:graph_solution_approximation}, $\timeApprox = \O(m + zr^2k^3/\eps^2)$ for $z = \O(k^2), r = \O(k^4), \eps = \widetilde{\Omega}(1)$, so $\timeApprox = \O(m + zr^2k^3/\eps^2) = \O(m + k^{13}) = \O(m)$ for $k = m^{1/58}$. Thus, the expression in \eqref{eq:runtimeformula} evaluates to $\O(m^{3/2-1/58}\log U)$ as desired.
\end{proof}

\section*{Acknowledgments}

Aaron Sidford is supported in part by a Microsoft Research Faculty Fellowship, NSF CAREER Award CCF-1844855, NSF Grant CCF-1955039, a PayPal research award, and a Sloan Research Fellowship. 
Yang P.~Liu is supported by the Department of Defense (DoD) through the National Defense Science and Engineering Graduate Fellowship, and NSF CAREER Award CCF-1844855 and NSF Grant CCF-1955039. Yin Tat Lee is supported in part by NSF awards CCF-1749609, DMS-1839116, DMS-2023166, CCF-2105772, a Microsoft Research Faculty Fellowship, a Sloan Research Fellowship, and a Packard Fellowship.
Jan van den Brand is funded by ONR BRC grant N00014-18-1-2562 and by the Simons Institute for the Theory of Computing through a Simons-Berkeley Postdoctoral Fellowship.

Part of this work was done while Jan van den Brand, Yang P.~Liu and Aaron Sidford were visiting the Discrete Optimization Trimester Program at the Hausdorff Research Institute for Mathematics.

{\small
\bibliographystyle{alpha}
\bibliography{refs}}

\appendix

\end{document}